\documentclass[aps,prb, superscriptaddress, twocolumn, nofootinbib]{revtex4-2}%Use revtex4-2

\usepackage[T1]{fontenc}
\usepackage{amsmath,amsfonts,amssymb,amsthm,bbm,bm, braket}
\usepackage{wasysym}
\usepackage{comment}
\usepackage{scalerel}
\usepackage{enumerate}
\usepackage{graphicx}
\usepackage{mathtools}
\usepackage{ifthen}
\usepackage{tensor}
\usepackage{tikz}
\usepackage{tikz-network}
\usetikzlibrary{patterns,decorations.pathreplacing}

\usepackage{color}
\definecolor{myred}{RGB}{232,102,102}
\definecolor{myblue}{RGB}{187,187,255}
\definecolor{myorange}{RGB}{202,52,51}
\definecolor{mygrey}{RGB}{105,105,105}
\definecolor{OliveGreen}{RGB}{85,107,47}
\definecolor{NavyBlue}{RGB}{0,0,128}
\definecolor{mygreen}{RGB}{34,139,34}
\definecolor{myY}{RGB}{220,255,203}
\definecolor{myYO}{RGB}{255, 220, 151}
\definecolor{myblue1}{RGB}{176,223,229}
\definecolor{myblue2}{RGB}{0,0,128}
\definecolor{myblue3}{RGB}{0,108,255}
\definecolor{myblue4}{RGB}{101,147,245}
\definecolor{myblue5}{RGB}{115,194,251}
\definecolor{myblue6}{RGB}{87,160,211}
\definecolor{myblue7}{RGB}{137,207,240}
\definecolor{myblue8}{RGB}{29,41,81}
\definecolor{myblue9}{RGB}{14,77,146}
\definecolor{myblue10}{RGB}{15,82,186}
\definecolor{myred1}{RGB}{255,36,0}
\definecolor{myred2}{RGB}{205,92,92}
\definecolor{myred3}{RGB}{178,34,34}
\definecolor{myred4}{RGB}{164,90,82}
\definecolor{myred5}{RGB}{255,8,0}
\definecolor{myred6}{RGB}{202,52,51}
\definecolor{myred7}{RGB}{66,13,9}
\definecolor{myred8}{RGB}{141,2,31}
\definecolor{myred9}{RGB}{250,128,114}
\definecolor{myred10}{RGB}{237,41,57}
\definecolor{myyellow1}{RGB}{254,220,86}
\definecolor{myyellow2}{RGB}{255,229,180}
\definecolor{myyellow3}{RGB}{238,220,130}
\definecolor{myyellow4}{RGB}{253,165,15}
\definecolor{myyellow5}{RGB}{255,195,11}
\definecolor{myyellow6}{RGB}{218,165,32}
\definecolor{myyellow7}{RGB}{255,211,0}
\definecolor{myyellow8}{RGB}{248,222,126}
\definecolor{myyellow9}{RGB}{245,245,220}
\definecolor{myyellow10}{RGB}{248,228,115}

\definecolor{mygray1}{RGB}{246,246,246}
\definecolor{mygray2}{RGB}{32,32,32}
\definecolor{mygray3}{RGB}{64,64,64}
\definecolor{mygray4}{RGB}{96,96,96}
\definecolor{mygray5}{RGB}{128,128,128}
\definecolor{mygray6}{RGB}{160,160,160}
\definecolor{mygray7}{RGB}{224,224,224}
\definecolor{mygray8}{RGB}{180,180,180}

\usepackage{longtable}%Added!
\usepackage[normalem]{ulem} % either use this (simple) 

\theoremstyle{break}

\usepackage{color}
\definecolor{myred}{RGB}{232,102,102}
\definecolor{myblue}{RGB}{187,187,255}
\definecolor{myorange0}{RGB}{252,226,5}
\definecolor{myorange0c}{RGB}{255,255,255}
\definecolor{myorange}{RGB}{255,165,0}
\definecolor{mygrey}{RGB}{105,105,105}
\definecolor{OliveGreen}{RGB}{85,107,47}
\definecolor{NavyBlue}{RGB}{0,0,128}
\definecolor{mygreen}{RGB}{34,139,34}
\definecolor{myY}{RGB}{220,255,203}
\definecolor{myYO}{RGB}{255, 220, 151}

\definecolor{mygreenc}{RGB}{150,50,50}

\usepackage{tensor}
\newcommand{\be}{\begin{equation}}
\newcommand{\ee}{\end{equation}}
\newcommand{\ba}{\begin{aligned}}
\newcommand{\ea}{\end{aligned}}
\newcommand{\bw}{\begin{widetext}}
\newcommand{\ew}{\end{widetext}}
\newcommand{\1}{\mathbbm{1}}
\newtheorem{theorem}{Theorem}
\theoremstyle{plain}

\theoremstyle{plain}
\newtheorem{lemma}{Lemma}
\theoremstyle{plain}

\usepackage[colorlinks,bookmarks=false,citecolor=NavyBlue,linkcolor=OliveGreen,urlcolor=blue]{hyperref}

\newcommand{\Wgatedagger}[2]{
\draw[very thick] (#1-0.5, #2 +0.5) -- (#1+0.5,#2-0.5);
\draw[very thick] (#1-0.5,#2-0.5) -- (#1+0.5,#2+0.5);
\draw[ thick, fill=mygreenc, rounded corners=2pt] (#1-0.25,#2+0.25) rectangle (#1+0.25,#2-0.25);
\draw[thick] (#1,#2+0.15) -- (#1+0.15,#2+0.15) -- (#1+0.15,#2);
}

\newcommand{\Wgateorange}[2]{
\draw[very thick] (#1-0.5, #2 +0.5) -- (#1+0.5,#2-0.5);
\draw[very thick] (#1-0.5,#2-0.5) -- (#1+0.5,#2+0.5);
\draw[thick, fill=myorange, rounded corners=2pt] (#1-0.25,#2+0.25) rectangle (#1+0.25,#2-0.25);
\draw[thick] (#1,#2+0.15) -- (#1+0.15,#2+0.15) -- (#1+0.15,#2);
}
\newcommand{\Wgategreen}[2]{
\draw[very thick] (#1-0.5, #2 +0.5) -- (#1+0.5,#2-0.5);
\draw[very thick] (#1-0.5,#2-0.5) -- (#1+0.5,#2+0.5);
\draw[ thick, fill=mygreen, rounded corners=2pt] (#1-0.25,#2+0.25) rectangle (#1+0.25,#2-0.25);
\draw[thick] (#1,#2+0.15) -- (#1+0.15,#2+0.15) -- (#1+0.15,#2);
}
\newcommand{\Wgategrey}[2]{
\draw[very thick] (#1-0.5, #2 +0.5) -- (#1+0.5,#2-0.5);
\draw[very thick] (#1-0.5,#2-0.5) -- (#1+0.5,#2+0.5);
\draw[ thick, fill=mygrey, rounded corners=2pt] (#1-0.25,#2+0.25) rectangle (#1+0.25,#2-0.25);
\draw[thick] (#1,#2+0.15) -- (#1+0.15,#2+0.15) -- (#1+0.15,#2);
}

\definecolor{myblue1}{RGB}{176,223,229}
\definecolor{myblue2}{RGB}{0,0,128}
\definecolor{myblue3}{RGB}{0,108,255}
\definecolor{myblue4}{RGB}{101,147,245}
\definecolor{myblue5}{RGB}{115,194,251}
\definecolor{myblue6}{RGB}{87,160,211}
\definecolor{myblue7}{RGB}{137,207,240}
\definecolor{myblue8}{RGB}{29,41,81}
\definecolor{myblue9}{RGB}{14,77,146}
\definecolor{myblue10}{RGB}{15,82,186}

\definecolor{myyellow1}{RGB}{254,220,86}
\definecolor{myyellow2}{RGB}{255,229,180}
\definecolor{myyellow5}{RGB}{255,195,11}
\definecolor{myyellow6}{RGB}{218,165,32}

\newcommand{\mcirc}{\mathbin{\scalerel*{\fullmoon}{G}}}

\newcommand{\tr}{\text{tr} \, }

\def\scale{0.3}

\hypersetup{
pdftitle={Thermalisation Dynamics and Spectral Statistics of Extended Systems with Thermalising Boundaries},
pdfsubject={Statistical mechanics},
pdfauthor={Pavel Kos, Tomaz Prosen, and Bruno Bertini},
pdfkeywords={thermalisation, chaos, quantum many-body systems}
}

\newcommand{\ch}[1]{#1} %Changed

\begin{document}

\title{Thermalisation Dynamics and Spectral Statistics of Extended Systems with Thermalising Boundaries}
\date{\today}

 \author{Pavel Kos}
 \affiliation{Department of Physics, Faculty of Mathematics and Physics, University of Ljubljana, Jadranska 19, SI-1000 Ljubljana, Slovenia}
  \author{Toma\v z Prosen}
 \affiliation{Department of Physics, Faculty of Mathematics and Physics, University of Ljubljana, Jadranska 19, SI-1000 Ljubljana, Slovenia}
 \author{Bruno Bertini}
 \affiliation{Rudolf Peierls Centre for Theoretical Physics, Clarendon Laboratory, Oxford University, Parks Road, Oxford OX1 3PU, United Kingdom}

\begin{abstract}
We study thermalisation and spectral properties of extended systems connected, through their boundaries, to a thermalising Markovian bath. Specifically, we consider periodically driven systems modelled by brickwork quantum circuits where a finite section (block) \ch{of the circuit is constituted by arbitrary} local unitary gates while its complement, \ch{which plays the role of the bath,} is dual-unitary. We show that the evolution of local observables and the spectral form factor are determined by the same quantum channel, which we use to characterise the system's dynamics and spectral properties. In particular, we identify a family of fine-tuned quantum circuits --- which we call \ch{strongly non-ergodic} --- that fails to thermalise even in this controlled setting, and, accordingly, their spectral form factor does not follow the random matrix theory prediction. We provide a set of necessary conditions on the local quantum gates that lead to strong \ch{non-ergodicity}, and in the case of qubits, we provide a complete classification of \ch{strongly non-ergodic} circuits. We also study the opposite extreme case of circuits that are almost dual-unitary, i.e., \ch{where thermalisation occurs with the fastest possible rate}. We show that, in these systems, local observables and spectral form factor approach respectively thermal values and random matrix theory prediction exponentially fast. We provide a perturbative characterisation of the dynamics and, in particular, of the time-scale for thermalisation. 
\end{abstract}

\maketitle

%------------------------------------------
%\onecolumngrid
\pagebreak 
%-----------------------------------------

%-----------------------------------------
%\tableofcontents

\section{Introduction}

One of the most natural questions arising when facing an isolated quantum many-body system out of equilibrium is whether or not it eventually thermalises~\cite{polkovnikov2011,eisert2015quantum, gogolin2016equilibration, calabrese2016introduction, dalessio2016quantum, yukalov2011equilibration}. As it is now well understood~\cite{essler2016quench}, in the presence of local interactions this question is generically well-posed only when looking at finite regions of space, finite ``subsystems'', in the limit where the whole system is taken to be infinitely large and acts as an effective bath. If a subsystem remembers some of its past, i.e.\ it acts as a quantum memory, then it does not thermalise and cannot be described using the principles of statistical mechanics. In the usual scenarios, however, the action of the effective bath is enough to assure thermalisation. Therefore, for times much larger than a certain thermalisation time-scale $\tau_{\rm thm}$, subsystems are well described by a (time-independent) thermal density matrix. 

This \ch{explanation is intuitive}, but examples where one can actually establish its validity from first principles in the presence of non-trivial interactions are very scarce. In fact, the only instances where this has been achieved are dual-unitary circuits~\cite{bertini2019entanglement, piroli2020exact, suzuki2021computational}, where the dynamics remains unitary upon switching space and time~\cite{bertini2019exact}, and a special class of integrable models~\cite{klobas2021exact, klobas2021exactII, klobas2021entanglement, pozsgay2014quantum, pozsgay2016real, zadnik2021folded, zadnik2021foldedII} with a particularly simple scattering matrix.

This paper focuses on a simple setting where one can study these phenomena in great detail, even away from the aforementioned solvable points. 
\ch{We consider a one-dimensional system of size $L$ consisting of qudits on integer and half-integer positions and subdivided into two parts: a finite region of interest, $A$, surrounded by a region $B$ (bath). The sizes of the two regions are respectively denoted by $L_A$ and $L_B$. The time evolution is given by a quantum circuit made of local gates. Specifically, the region $A$ is evolved with arbitrary two-qudit gates, whereas $B$ is driven by a dual-unitary quantum circuit.} 
In this setting, the dual-unitary part of the circuit simplifies to a perfect Markovian \ch{ (i.e.\ memoryless)} bath. Thus, the thermalisation dynamics of the central system is given exactly in terms of a \ch{non-unitary} {boundary time evolution map}, $\mathbb B_A$, which implements unitary evolution with Markovian dissipators \ch{(totally depolarising quantum channels)} at the two boundaries. 
The task is then to characterise the \ch{largest eigenvalues and eigenvectors} of $\mathbb B_A$: as we will see, this is easier than solving a generic quench problem, even when restricting to the setting of brickwork quantum circuits. Here, in particular, we characterise two opposite limiting cases: perturbed dual unitary circuits where the time-scale for thermalisation is small but finite (and goes to zero with vanishing perturbation strength) and \ch{``strongly non-ergodic''} circuits, where it is infinitely long. 
We also show that in our setting $\mathbb B_A$ bears information about the spectral statistics of the system. Indeed, as pointed out in \cite{garratt2021local}, one can use it to express the so-called spectral form factor: the Fourier transform of the two-point correlation function of the quasi-energies. This object furnishes a convenient measure of spectral correlations and has recently attracted substantial attention in the context of many-body physics~\cite{kos2018many, bertini2018exact, bertini2021random, chan2018spectral, suntajs2020quantum, friedman2019spectral, roy2020random, garratt2021many, flack2020statistics,chan2021spectral, singh2021subdiffusion}. Interestingly, we find that in the cases where the system thermalises, the spectral form factor agrees with the result obtained in the circular unitary ensemble of random matrices. \ch{These systems are therefore quantum chaotic according to the extension of the quantum chaos conjecture~\cite{casati1980on,berry1981quantizing,bohigas1984characterization} to non-semiclassical quantum many-body systems. 
}

The paper is structured as follows. In Sec.~\ref{sec:setting} we describe our setting and introduce the necessary notation. In particular, we present the central object of study of this paper: the boundary time evolution map $\mathbb B_A$. Next, in Sec.~\ref{sec:pertDU}, we characterise the map $\mathbb B_A$ analytically in the case of perturbed-dual unitary circuits, deriving analytic predictions for spectral form factor and thermalisation dynamics of local observables. In Sec.~\ref{sec:StrongL} we discuss the opposite limiting case, which we dub  \ch{\emph{strongly non-ergodic}}, where finite subsystems do not thermalise. We derive a set of necessary conditions that the local dynamics has to fulfil to have \ch{strong non-ergodicity} and determine all \ch{strongly non-ergodic} quantum circuits of qubits. Lastly, we report our concluding remarks in Sec.~\ref{sec:conclusion}.

%-----------------------------------------
\section{Setting}
\label{sec:setting}
%-----------------------------------------
In this work, we study one-dimensional brick-work quantum circuits. These systems \ch{consist of qudits with $d$ internal states (with on-site Hilbert space $\mathcal{H}_1=\mathbb{C}^d$),} placed on the sites --- labelled by half-integer numbers --- of a one-dimensional lattice and are evolved periodically in time. 

The evolution over one period --- set to one in the following --- is specified by a Floquet operator $\mathbb{U}$. \ch{It is expressed in terms of $d^2 \times d^2$ unitary matrices (\emph{local gates})}, which couple together two neighbouring qudits. 
\ch{We depict the the Floquet operator following the diagrammatic representation of Ref.~\cite{bertini2021random}, which resembles that of tensor network theory~\cite{cirac2020matrix}. In particular  we represent the gates, elements of ${\rm U}(d^2)$, as 
}
\be
\label{eq:Ugate}
U_{x,\tau}=\begin{tikzpicture}[baseline=(current  bounding  box.center), scale=.7]
\draw[ thick] (-4.25,0.5) -- (-3.25,-0.5);
\draw[ thick] (-4.25,-0.5) -- (-3.25,0.5);
\draw[ thick, fill=myblue, rounded corners=2pt] (-4,0.25) rectangle (-3.5,-0.25);
\draw[thick] (-3.75,0.15) -- (-3.75+0.15,0.15) -- (-3.75+0.15,0);
\Text[x=-4.25,y=-0.75]{}
\end{tikzpicture}\, ,\qquad 
U_{x,\tau}^\dagger = 
\begin{tikzpicture}[baseline=(current  bounding  box.center), scale=.7]
\draw[thick] (-2.25,0.5) -- (-1.25,-0.5);
\draw[thick] (-2.25,-0.5) -- (-1.25,0.5);
\draw[ thick, fill=myred, rounded corners=2pt] (-2,0.25) rectangle (-1.5,-0.25);
\draw[thick] (-1.75,0.15) -- (-1.6,0.15) -- (-1.6,0);
\Text[x=-2.25,y=-0.75]{}
\end{tikzpicture}\,, 
\ee
where the matrix multiplication goes from bottom to top and the mark $\begin{tikzpicture}[baseline=(current  bounding  box.center), scale=.7]\draw[thick] (-1.75,0.15) -- (-1.6,0.15) -- (-1.6,0); \Text[x=-1.75,y=-0.25]{}\end{tikzpicture}$ is stressing that the gates are \ch{generally} not symmetric under reflection and time reversal. The Floquet operator is then represented as 
\begin{align}
&\mathbb U =
\begin{tikzpicture}[baseline=(current  bounding  box.center), scale=0.55]
\draw[thick, dotted] (-9.5,-1.7) -- (0.4,-1.7);
\draw[thick, dotted] (-9.5,-1.3) -- (0.4,-1.3);
\foreach \i in {3,...,13}{
\draw[gray, dashed] (-12.5+\i,-2.1) -- (-12.5+\i,0.3);
}
\foreach \i in {3,...,5}{
\draw[gray, dashed] (-9.75,3-\i) -- (.75,3-\i);
}
\foreach \i in{0.5}
{
\draw[ thick] (0.5,1+2*\i-0.5-3.5) arc (-45:90:0.17);
\draw[ thick] (-10+0.5,1+2*\i-0.5-3.5) arc (270:90:0.15);
}
\foreach \i in{1.5}{
\draw[thick] (0.5,2*\i-0.5-3.5) arc (45:-90:0.17);
\draw[thick] (-10+0.5+0,2*\i-0.5-3.5) arc (90:270:0.15);
}
\foreach \i in {0,1}{
\Text[x=1.25,y=-2+2*\i]{\scriptsize$\i$}
}
\foreach \i in {1}{
\Text[x=1.25,y=-2+\i]{\small$\frac{\i}{2}$}
}
\foreach \i in {1,3,5}{
\Text[x=-7.5+\i+1-3,y=-2.6]{\small$\frac{\i}{2}$}
}
\foreach \i in {1,2,3}{
\Text[x=-7.5+2*\i-2,y=-2.6]{\scriptsize${\i}$}
}
\Text[x=-7.5+2*5-2,y=-2.6]{\small $L$}
%NonDU gates
\foreach \jj[evaluate=\jj as \j using -2*(ceil(\jj/2)-\jj/2)] in {0}
\foreach \i in {1,...,5}
{
\draw[thick] (.5-2*\i-1*\j,-2-1*\jj) -- (1-2*\i-1*\j,-1.5-\jj);
\draw[thick] (1-2*\i-1*\j,-1.5-1*\jj) -- (1.5-2*\i-1*\j,-2-\jj);
\draw[thick] (.5-2*\i-1*\j,-1-1*\jj) -- (1-2*\i-1*\j,-1.5-\jj);
\draw[thick] (1-2*\i-1*\j,-1.5-1*\jj) -- (1.5-2*\i-1*\j,-1-\jj);
\draw[thick, fill=myblue, rounded corners=2pt] (0.75-2*\i-1*\j,-1.75-\jj) rectangle (1.25-2*\i-1*\j,-1.25-\jj);
\draw[thick] (-2*\i+1,-1.35-\jj) -- (-2*\i+1.15,-1.35-\jj) -- (-2*\i+1.15,-1.5-\jj);%
}
\foreach \i in {0,...,4}{
\draw[thick] (-.5-2*\i,-1) -- (0.525-2*\i,0.025);
\draw[thick] (-0.525-2*\i,0.025) -- (0.5-2*\i,-1);
\draw[thick, fill=myblue, rounded corners=2pt] (-0.25-2*\i,-0.25) rectangle (.25-2*\i,-0.75);
\draw[thick] (-2*\i,-0.35) -- (-2*\i+0.15,-0.35) -- (-2*\i+0.15,-0.5);%
}
\Text[x=-2,y=-2.6]{$\cdots$}
\end{tikzpicture}\, ,
\label{eq:Floquet0}
\end{align}
where \ch{we consider periodic boundary conditions. In} principle, gates at different positions can be different: we denote them as $U_{x,\tau}$, with indexes denoting the position of the top-right corner, with \ch{ $x\in \mathbb Z_{2 L}/2$, $\tau\in\!\{{1}/{2},1\}$ \ch{(later they repeat because the evolution is periodic)}, and $x+\tau\in\mathbb Z$.
Concretely, all $U_{x,\tau}$ in the region $A$ might be arbitrary and independent. In Sections \ref{sec:pertDU} and \ref{sec:SLd2} we take all gates to be the same $U_{x,\tau}=U$. In Sec. \ref{sec:pertDU} this restriction is made out of convenience, and analogous computations can be performed in the general case. In contrast, this restriction has physical significance in Sec. \ref{sec:SLd2}.%, as it significally restricts the class of models under consideration.
}

%-----------------------------------------
%\subsection{Thermalisation after a quench}
%-----------------------------------------

Here we consider the Floquet operator given by $2L_A$ general gates and  $2L_B$ \emph{dual-unitary gates} denoted in yellow
\ch{
\begin{align}
&\mathbb U = 
\begin{tikzpicture}[baseline=(current  bounding  box.center), scale=0.55]
\draw[thick, dotted] (-9.5,-1.7) -- (0.4,-1.7);
\draw[thick, dotted] (-9.5,-1.3) -- (0.4,-1.3);
\foreach \i in {3,...,13}{
\draw[gray, dashed] (-12.5+\i,-2.1) -- (-12.5+\i,0.3);
}
\foreach \i in {3,...,5}{
\draw[gray, dashed] (-9.75,3-\i) -- (.75,3-\i);
}
\foreach \i in{0.5}
{
\draw[ thick] (0.5,1+2*\i-0.5-3.5) arc (-45:90:0.17);
\draw[ thick] (-10+0.5,1+2*\i-0.5-3.5) arc (270:90:0.15);
}
\foreach \i in{1.5}{
\draw[thick] (0.5,2*\i-0.5-3.5) arc (45:-90:0.17);
\draw[thick] (-10+0.5+0,2*\i-0.5-3.5) arc (90:270:0.15);
}
\foreach \i in {0,1}{
\Text[x=1.25,y=-2+2*\i]{\scriptsize$\i$}
}
\foreach \i in {1}{
\Text[x=1.25,y=-2+\i]{\small$\frac{\i}{2}$}
}
%\foreach \i in {1,3,5}{
%\Text[x=-7.5+\i+1-3,y=-2.6]{\small$\frac{\i}{2}$}
%}
%\foreach \i in {1,2,3}{
%\Text[x=-7.5+2*\i-2,y=-2.6]{\scriptsize${\i}$}
%}
%NonDU gates
\foreach \jj[evaluate=\jj as \j using -2*(ceil(\jj/2)-\jj/2)] in {0}
\foreach \i in {3,...,5}
{
\draw[thick] (.5-2*\i-1*\j,-2-1*\jj) -- (1-2*\i-1*\j,-1.5-\jj);
\draw[thick] (1-2*\i-1*\j,-1.5-1*\jj) -- (1.5-2*\i-1*\j,-2-\jj);
\draw[thick] (.5-2*\i-1*\j,-1-1*\jj) -- (1-2*\i-1*\j,-1.5-\jj);
\draw[thick] (1-2*\i-1*\j,-1.5-1*\jj) -- (1.5-2*\i-1*\j,-1-\jj);
\draw[thick, fill=myblue, rounded corners=2pt] (0.75-2*\i-1*\j,-1.75-\jj) rectangle (1.25-2*\i-1*\j,-1.25-\jj);
\draw[thick] (-2*\i+1,-1.35-\jj) -- (-2*\i+1.15,-1.35-\jj) -- (-2*\i+1.15,-1.5-\jj);%
}
\foreach \i in {2,...,4}{
\draw[thick] (-.5-2*\i,-1) -- (0.525-2*\i,0.025);
\draw[thick] (-0.525-2*\i,0.025) -- (0.5-2*\i,-1);
\draw[thick, fill=myblue, rounded corners=2pt] (-0.25-2*\i,-0.25) rectangle (.25-2*\i,-0.75);
\draw[thick] (-2*\i,-0.35) -- (-2*\i+0.15,-0.35) -- (-2*\i+0.15,-0.5);%
}
%DU gates
\foreach \i in {0,...,1}{
\draw[thick] (-.5-2*\i,-1) -- (0.525-2*\i,0.025);
\draw[thick] (-0.525-2*\i,0.025) -- (0.5-2*\i,-1);
\draw[thick, fill=myorange0, rounded corners=2pt] (-0.25-2*\i,-0.25) rectangle (.25-2*\i,-0.75);
\draw[thick] (-2*\i,-0.35) -- (-2*\i+0.15,-0.35) -- (-2*\i+0.15,-0.5);%
}
\foreach \jj[evaluate=\jj as \j using -2*(ceil(\jj/2)-\jj/2)] in {0}
\foreach \i in {1,...,2}
{
\draw[thick] (.5-2*\i-1*\j,-2-1*\jj) -- (1-2*\i-1*\j,-1.5-\jj);
\draw[thick] (1-2*\i-1*\j,-1.5-1*\jj) -- (1.5-2*\i-1*\j,-2-\jj);
\draw[thick] (.5-2*\i-1*\j,-1-1*\jj) -- (1-2*\i-1*\j,-1.5-\jj);
\draw[thick] (1-2*\i-1*\j,-1.5-1*\jj) -- (1.5-2*\i-1*\j,-1-\jj);
\draw[thick, fill=myorange0, rounded corners=2pt] (0.75-2*\i-1*\j,-1.75-\jj) rectangle (1.25-2*\i-1*\j,-1.25-\jj);
\draw[thick] (-2*\i+1,-1.35-\jj) -- (-2*\i+1.15,-1.35-\jj) -- (-2*\i+1.15,-1.5-\jj);%
}
%\foreach \jj[evaluate=\jj as \j using -2*(ceil(\jj/2)-\jj/2)] in {0}
%\foreach \i in {1,...,5}{
%\draw[thick] (-2*\i+1.4,-.95-\jj) -- (-2*\i+1.55,-.95-\jj) -- (-2*\i+1.55,-1.1-\jj);
%}
%\foreach \jj[evaluate=\jj as \j using -2*(ceil(\jj/2)-\jj/2)] in {0}
%\foreach \i in {1,...,5}{
%\draw[thick] (-2*\i+2.4,0.05-\jj) -- (-2*\i+2.55,0.05-\jj) -- (-2*\i+2.55,-0.1-\jj);
%}
%\foreach \jj[evaluate=\jj as \j using -2*(ceil(\jj/2)-\jj/2)] in {0}
%\foreach \i in {1,...,5}{
%\draw[thick]  (-2*\i+1.45,-0.1-\jj) -- (-2*\i+1.45,0.05-\jj) -- (-2*\i+1.6,0.05-\jj);
%}
%\foreach \jj[evaluate=\jj as \j using -2*(ceil(\jj/2)-\jj/2)] in {0}
%\foreach \i in {0,...,4}{
%\draw[thick]  (-2*\i+.45,-1.1-\jj) -- (-2*\i+.45,-0.95-\jj) -- (-2*\i+.6,-0.95-\jj);
%}
%\Text[x=-2,y=-2.6]{$\cdots$}
\draw [thick, black,decorate,decoration={brace,amplitude=10pt,mirror},xshift=0.0pt,yshift=-0.0pt](-8.5,-2.2) -- (-4.5,-2.2) node[black,midway,yshift=-0.6cm] { A,  $2L_A-1$ sites};
\draw [thick, black,decorate,decoration={brace,amplitude=10pt,mirror},xshift=0.0pt,yshift=-0.0pt](-3.45,-2.2) -- (.5,-2.2) node[black,midway,yshift=-0.6cm] { B, $2L_B+1$ sites};
\end{tikzpicture},
\label{eq:diagramFloquet}
\end{align}
}
\ch{where we showed explicitly that the region $A$ is composed of the first $2L_A-1$ sites  $\{ \frac{1}{2}, \dots, L_A- \frac{1}{2}\}$, and the region $B$ contains the rest of the system.}
Dual unitary gates~\cite{bertini2019exact} 
\be
U_{x,\tau}^{\rm du}=\begin{tikzpicture}[baseline=(current  bounding  box.center), scale=.7]
\draw[ thick] (-4.25,0.5) -- (-3.25,-0.5);
\draw[ thick] (-4.25,-0.5) -- (-3.25,0.5);
\draw[ thick, fill=myorange0, rounded corners=2pt] (-4,0.25) rectangle (-3.5,-0.25);
\draw[thick] (-3.75,0.15) -- (-3.75+0.15,0.15) -- (-3.75+0.15,0);
\Text[x=-4.25,y=-0.75]{}
\end{tikzpicture}\, ,\qquad\qquad
(U_{x,\tau}^{\rm du})^\dagger = 
\begin{tikzpicture}[baseline=(current  bounding  box.center), scale=.7]
\draw[thick] (-2.25,0.5) -- (-1.25,-0.5);
\draw[thick] (-2.25,-0.5) -- (-1.25,0.5);
\draw[ thick, fill=myorange0c, rounded corners=2pt] (-2,0.25) rectangle (-1.5,-0.25);
\draw[thick] (-1.75,0.15) -- (-1.6,0.15) -- (-1.6,0);
\Text[x=-2.25,y=-0.75]{}
\end{tikzpicture},
\ee
are defined as the family of gates fulfilling the following four conditions 
\begin{align}
&\begin{tikzpicture}[baseline=(current  bounding  box.center), scale=.65]
\def\ep{1.3};
\def\eps{0.8};
\draw[ thick] (-2.25,1) -- (-1.75,0.5);
\draw[ thick] (-1.75,0.5) -- (-1.25,1);
\draw[ thick] (-2.25,-1) -- (-1.75,-0.5);
\draw[ thick] (-1.75,-0.5) -- (-1.25,-1);
\draw[ thick] (-1.9,0.35) to[out=170, in=-170] (-1.9,-0.35);
\draw[ thick] (-1.6,0.35) to[out=10, in=-10] (-1.6,-0.35);
\draw[thick, fill=myorange0c, rounded corners=2pt] (-2,0.25) rectangle (-1.5,0.75);
\draw[thick, fill=myorange0, rounded corners=2pt] (-2,-0.25) rectangle (-1.5,-0.75);
\draw[thick] (-1.75,0.65) -- (-1.6,0.65) -- (-1.6,0.5);
\draw[thick] (-1.75,-0.35) -- (-1.6,-0.35) -- (-1.6,-0.5);
\Text[x=-1.25,y=-0.5-\ep]{}
\Text[x=-0.4,y=0]{$=$}
\draw[ thick] (.7,-.75) -- (.7,.75) (1.3,-0.75) -- (1.3,0.75) (.7,-.75) -- (.45,.-1) (1.3,-.75) -- (1.55,.-1) (.7,.75) -- (.45,1) (1.3,.75) -- (1.55,1);
\Text[x=1.6,y=0]{,}
\end{tikzpicture}
& 
&\begin{tikzpicture}[baseline=(current  bounding  box.center), scale=.65]
\def\ep{1.3};
\def\eps{0.8};
\draw[ thick] (-2.25,1) -- (-1.75,0.5);
\draw[ thick] (-1.75,0.5) -- (-1.25,1);
\draw[ thick] (-2.25,-1) -- (-1.75,-0.5);
\draw[ thick] (-1.75,-0.5) -- (-1.25,-1);
\draw[ thick] (-1.9,0.35) to[out=170, in=-170] (-1.9,-0.35);
\draw[ thick] (-1.6,0.35) to[out=10, in=-10] (-1.6,-0.35);
\draw[thick, fill=myorange0, rounded corners=2pt] (-2,0.25) rectangle (-1.5,0.75);
\draw[thick, fill=myorange0c, rounded corners=2pt] (-2,-0.25) rectangle (-1.5,-0.75);
\draw[thick] (-1.75,0.65) -- (-1.6,0.65) -- (-1.6,0.5);
\draw[thick] (-1.75,-0.35) -- (-1.6,-0.35) -- (-1.6,-0.5);
\Text[x=-1.25,y=-0.5-\ep]{}
\Text[x=-0.4,y=0]{$=$}
\draw[thick] (.7,-.75) -- (.7,.75) (1.3,-0.75) -- (1.3,0.75) (.7,-.75) -- (.45,.-1) (1.3,-.75) -- (1.55,.-1) (.7,.75) -- (.45,1) (1.3,.75) -- (1.55,1);
\Text[x=1.6,y=0]{,}
\end{tikzpicture}\\
&\begin{tikzpicture}[baseline=(current  bounding  box.center), scale=0.65]
\def\eps{0.4};
\def\ep{1.8}
\Text[x=-0.125,y=-\ep]{}
\draw[ thick] (-1.75+2,0.5+\eps) -- (-1.25+2,1+\eps);
\draw[ thick] (-1.25+2,0+\eps) -- (-1.75+2,0.5+\eps);
\draw[ thick] (-1.25+2,0-\eps) -- (-1.75+2,-0.5-\eps);
\draw[ thick] (-1.75+2,-0.5-\eps) -- (-1.25+2,-1-\eps);
\draw[ thick] (-1.9+2,0.7+\eps) to[out=170, in=-170] (-1.9+2,-0.7-\eps);
\draw[ thick] (-1.9+2,0.3+\eps) to[out=170, in=-170] (-1.9+2,-0.3-\eps);
\draw[thick, fill=myorange0c, rounded corners=2pt] (-2+2,0.25+\eps) rectangle (-1.5+2,0.75+\eps);
\draw[thick, fill=myorange0, rounded corners=2pt] (-2+2,-0.25-\eps) rectangle (-1.5+2,-0.75-\eps);
\draw[thick] (.25,0.65+\eps) -- (.4,0.65+\eps) -- (.4,0.5+\eps);
\draw[thick] (.25,-0.35-\eps) -- (.4,-0.35-\eps) -- (.4,-0.5-\eps);
\draw[ thick] (-1.45+4,0.7+\eps) to[out=170, in=-170] (-1.45+4,-0.7-\eps);
\draw[ thick] (-1.45+4,0.3+\eps) to[out=170, in=-170] (-1.45+4,-0.3-\eps);
\draw[ thick] (-1.1+4,0.7+\eps) -- (-1.45+4,0.7+\eps)(-1.1+4,0.7+\eps) -- (-.75+4,1+\eps);
\draw[ thick] (-1.1+4,0.3+\eps) -- (-1.45+4,0.3+\eps)(-1.1+4,0.3+\eps) -- (-.75+4,+\eps);
\draw[ thick] (-1.1+4,-0.3-\eps) -- (-1.45+4,-0.3-\eps)(-1.1+4,-0.3-\eps) -- (-.75+4,-\eps);
\draw[ thick] (-1.1+4,-0.7-\eps) -- (-1.45+4,-0.7-\eps)(-1.1+4,-0.7-\eps) -- (-.75+4,-1-\eps);
\Text[x=1.45,y=0]{$=$}
\Text[x=3.45,y=0]{,}
\end{tikzpicture}
& &
\begin{tikzpicture}[baseline=(current  bounding  box.center), scale=.65]
\def\eps{0.4};
\def\ep{1.8}
\Text[x=-0.125,y=-\ep]{}
\draw[ thick] (-2.25+2,1+\eps) -- (-1.75+2,0.5+\eps);
\draw[ thick] (-1.75+2,0.5+\eps) -- (-2.25+2,0+\eps);
\draw[ thick] (-1.75+2,-0.5-\eps) -- (-2.25+2,0-\eps);
\draw[ thick] (-2.25+2,-1-\eps) -- (-1.75+2,-0.5-\eps);
\draw[ thick] (-1.6+2,0.7+\eps) to[out=10, in=-10] (-1.6+2,-0.7-\eps);
\draw[ thick] (-1.6+2,0.3+\eps) to[out=10, in=-10] (-1.6+2,-0.3-\eps);
\draw[thick, fill=myorange0, rounded corners=2pt] (-2+2,0.25+\eps) rectangle (-1.5+2,0.75+\eps);
\draw[thick, fill=myorange0c, rounded corners=2pt] (-2+2,-0.25-\eps) rectangle (-1.5+2,-0.75-\eps);
\draw[thick] (.25,0.65+\eps) -- (.4,0.65+\eps) -- (.4,0.5+\eps);
\draw[thick] (.25,-0.35-\eps) -- (.4,-0.35-\eps) -- (.4,-0.5-\eps);
\draw[ thick] (-.7+4,0.7+\eps) to[out=10, in=-10] (-.7+4,-0.7-\eps);
\draw[ thick] (-.7+4,0.3+\eps) to[out=10, in=-10] (-.7+4,-0.3-\eps);
\draw[ thick] (-.7+4,0.7+\eps) -- (-1.05+4,0.7+\eps)(-1.35+4,1+\eps) -- (-1.05+4,0.7+\eps);
\draw[ thick] (-.7+4,0.3+\eps) -- (-1.05+4,0.3+\eps)(-1.35+4,+\eps) -- (-1.05+4,0.3+\eps);
\draw[ thick] (-.7+4,-0.3-\eps) -- (-1.05+4,-0.3-\eps)(-1.35+4,-\eps) -- (-1.05+4,-0.3-\eps);
\draw[ thick] (-.7+4,-0.7-\eps) -- (-1.05+4,-0.7-\eps)(-1.35+4,-1-\eps) -- (-1.05+4,-0.7-\eps);
\Text[x=1.8,y=0]{$=$}
\Text[x=4.2,y=0]{}
\end{tikzpicture}
\end{align}
\ch{where the first two conditions encode unitarity of the gates $U_{x,\tau} U_{x,\tau}^\dagger=U_{x,\tau}^\dagger U_{x,\tau} = \1$. The second two conditions encode the unitarity in space direction, i.e. unitarity of the gates with reshuffled indices.}
These properties enforce unitarity \ch{of global evolution} in both the spatial and temporal direction and have been shown to lead to exact calculations of several many-body properties such as spectral statistics~\cite{bertini2018exact, flack2020statistics, bertini2021random, fritzsch2021eigenstate}, operator spreading~\cite{bertini2020operator, bertini2020operator2, claeys2020maximum, bertini2020scrambling}, entanglement growth~\cite{bertini2019entanglement, gopalakrishnan2019unitary, piroli2020exact, jonay2021triunitary}, and thermalisation dynamics~\cite{piroli2020exact, claeys2021ergodic, suzuki2021computational}.

Here we are interested in the physical properties of the \ch{Floquet operator} \eqref{eq:diagramFloquet} in the limit of infinite $L$ with fixed $L_A$. For simplicity, we will consider the case where the dual unitary part of the circuit is homogeneous, i.e. $U^{\rm du}_{x,\tau}=U^{\rm du}$, while the gates $U_{x,\tau}$ \ch{in the central part $A$} can be position dependent. 

\subsection{Quantum quench}
To begin with let us look at a \emph{quantum quench} problem: we prepare the system in \ch{some state} and evolve it with the propagator \eqref{eq:diagramFloquet}. Specifically let us consider initial  matrix product states (MPSs) of the form  
\begin{align}
\ket{\psi_{L}} = \frac{1}{d^{\frac{L}{2}}}
\begin{tikzpicture}[baseline=(current  bounding  box.center), scale=.65]
\draw[very thick] (-2.05,0) -- (4.5,0);
\draw[very thick] (-2.05,0.15) arc (90:270:0.075);
\draw[very thick] (4.5,0.15) arc (90:-90:0.075);
\foreach \i in {0,1.35,2.7}
{
\draw[thick] (-2+\i,0.5) -- (-1.5+\i,0);
\draw[thick] (-1.5+\i,0) -- (-1+\i,0.5);
\draw[ thick, fill=myblue, rounded corners=2pt] (-1.5-0.35+\i,0.2-0.25) rectangle (-1.5+0.35+\i,0.2+0.2);
\draw[thick] (-1.5+0.1+\i,0.15+.18)-- (-1.35+0.1+\i,0.15+.18) -- (-1.35+0.1+\i,0+.18);}
\foreach \i in {4.05,5.4}
{
\draw[thick] (-2+\i,0.5) -- (-1.5+\i,0);
\draw[thick] (-1.5+\i,0) -- (-1+\i,0.5);
\draw[ thick, fill=myorange0, rounded corners=2pt] (-1.5-0.35+\i,0.2-0.25) rectangle (-1.5+0.35+\i,0.2+0.2);
\draw[thick] (-1.5+0.1+\i,0.15+.18)-- (-1.35+0.1+\i,0.15+.18) -- (-1.35+0.1+\i,0+.18);}
\Text[x=-1.5,y=-1.65]{}
\draw [thick, black,decorate,decoration={brace,amplitude=3pt,mirror},xshift=0.0pt,yshift=-0.0pt](1.75,.65) -- (-2.05,.65) node[black,midway,yshift=0.4cm] { $2L_A$};
\draw [thick, black,decorate,decoration={brace,amplitude=3pt,mirror},xshift=0.0pt,yshift=-0.0pt](4.35,.65) -- (2,.65) node[black,midway,yshift=0.4cm] { $2L_B$};
\end{tikzpicture}\,,
\label{eq:psi0diagram}
\end{align}
where the thick line represents \ch{bond} space of dimension $\chi$, while 
\be
[T_x]_{ab}^{cd}=\begin{tikzpicture}[baseline=(current  bounding  box.center), scale=.75]
\draw[thick] (-2,0.5) -- (-1.5,0);
\draw[thick] (-1.5,0) -- (-1,0.5);
\Text[x=-1.5,y=-0.5]{}
\Text[x=-2.25,y=0]{$a$}
\Text[x=-.75,y=0]{$c$}
\Text[x=-2.25,y=0.65]{$b$}
\Text[x=-.75,y=0.65]{$d$}
\draw[very thick] (-2.05,0) -- (-0.95,0);
\draw[ thick, fill=myblue, rounded corners=2pt] (-1.5-0.35,0.2-0.25) rectangle (-1.5+0.35,0.2+0.2);
\draw[thick] (-1.5+0.1,0.15+.18)-- (-1.35+0.1,0.15+.18) -- (-1.35+0.1,0+.18);
\end{tikzpicture}
\ee
are \ch{arbitrary} tensors and 
\be
[T_x^{\rm s}]_{ab}^{cd}=\begin{tikzpicture}[baseline=(current  bounding  box.center), scale=.75]
\draw[thick] (-2,0.5) -- (-1.5,0);
\draw[thick] (-1.5,0) -- (-1,0.5);
\draw[very thick] (-2.05,0) -- (-0.95,0);
\Text[x=-1.5,y=-0.5]{}
\Text[x=-2.25,y=0]{$a$}
\Text[x=-.75,y=0]{$c$}
\Text[x=-2.25,y=0.65]{$b$}
\Text[x=-.75,y=0.65]{$d$}
\draw[ thick, fill=myorange0, rounded corners=2pt] (-1.5-0.35,0.2-0.25) rectangle (-1.5+0.35,0.2+0.2);
\draw[thick] (-1.5+0.1,0.15+.18)-- (-1.35+0.1,0.15+.18) -- (-1.35+0.1,0+.18);
\end{tikzpicture},
\ee
are the ``solvable tensors'' of Ref.~\cite{piroli2020exact}. Namely, they are \emph{unitary} when seen as matrices implementing left-to-right multiplication and  \ch{they build the space transfer matrix}
 \be
 E_x(0)= \frac{1}{d}\,\,\begin{tikzpicture}[baseline=(current  bounding  box.center), scale=.75]
\def\eps{1.};
\draw[thick, rounded corners]  (-1.5,0.35+\eps) -- (-2,-0.15+\eps) -- (-2,0.5) -- (-1.5,0);
\draw[thick, rounded corners] (-1.5,0) -- (-1,0.5) -- (-1,-0.15+\eps) -- (-1.5,0.35+\eps);
\draw[very thick] (-2,0+\eps+.35) -- (-1,0+\eps+.35);
\draw[very thick] (-2,0) -- (-1,0);
\draw[ thick, fill=myorange0c, rounded corners=2pt] (-1.5-0.35,0.2-0.25+\eps) rectangle (-1.5+0.35,0.2+0.2+\eps);
%\draw[thick] (-1.75,0.17+\eps) -- (-1.75,0.02+\eps)-- (-1.6,0.02+\eps) ;
\draw[ thick, fill=myorange0, rounded corners=2pt] (-1.5-0.35,0.2-0.25) rectangle (-1.5+0.35,0.2+0.2);
\draw[thick] (-1.5+0.1,0.15+.18)-- (-1.35+0.1,0.15+.18) -- (-1.35+0.1,0+.18);
\draw[thick] (-1.5+0.1,0.15+.18+\eps)-- (-1.35+0.1,0.15+.18+\eps) -- (-1.35+0.1,0+.18+\eps);
\Text[x=-1.5,y=-.35]{}
\end{tikzpicture},
 \ee
 \ch{which has a unique fixed point (a single left/right eigenvector associated to eigenvalue $1$).} Here the white tensor is defined as 
 \be
[(T_x^{\rm s})^\dag]_{ab}^{cd}=\begin{tikzpicture}[baseline=(current  bounding  box.center), scale=.75]
\draw[thick] (-2,0.5) -- (-1.5,0);
\draw[thick] (-1.5,0) -- (-1,0.5);
\draw[very thick] (-2.05,0) -- (-0.95,0);
\Text[x=-1.5,y=-0.5]{}
\Text[x=-2.25,y=0]{$a$}
\Text[x=-.75,y=0]{$c$}
\Text[x=-2.25,y=0.65]{$b$}
\Text[x=-.75,y=0.65]{$d$}
\draw[ thick, fill=myorange0c, rounded corners=2pt] (-1.5-0.35,0.2-0.25) rectangle (-1.5+0.35,0.2+0.2);
\draw[thick] (-1.4,0.05) -- (-1.25,0.05) -- (-1.25,0.2);
\end{tikzpicture}=\left (\begin{tikzpicture}[baseline=(current  bounding  box.center), scale=.75]
\draw[thick] (-2,0.5) -- (-1.5,0);
\draw[thick] (-1.5,0) -- (-1,0.5);
\draw[very thick] (-2.05,0) -- (-0.95,0);
\Text[x=-1.5,y=-0.5]{}
\Text[x=-2.25,y=0]{$c$}
\Text[x=-.75,y=0]{$a$}
\Text[x=-2.25,y=0.65]{$d$}
\Text[x=-.75,y=0.65]{$b$}
\draw[ thick, fill=myorange0, rounded corners=2pt] (-1.5-0.35,0.2-0.25) rectangle (-1.5+0.35,0.2+0.2);
\draw[thick] (-1.5+0.1,0.15+.18)-- (-1.35+0.1,0.15+.18) -- (-1.35+0.1,0+.18);
\end{tikzpicture}\right)^*.
\ee
Our focus is on the dynamics of the reduced density matrix of the subsystem $A$ \ch{(cf. Eq.~\eqref{eq:diagramFloquet})}
% composed of the first $2L_A-1$ sites \ch{   $\{ \frac{1}{2}, \dots, L_A- \frac{1}{2}\}$:}
\be
\rho_{A, L}(t) = \tr_{\!\!B}[\mathbb{U}^t\ket{\psi_{L}} \bra{\psi_{L}} (\mathbb{U^\dagger})^t] \, ,
\label{eq:rhoL}
\ee
where ${\rm tr}_{B}$ denotes the partial trace over the subsystem $B$.
%of size $2L_B +1$
 %\be
%B=\{L_A,\ldots, L\}.
%\ee 

To represent \eqref{eq:rhoL} diagrammatically it is convenient to ``fold'' the representation of $(\mathbb{U^\dagger})^t$ behind that of $\mathbb{U}^t$ (this corresponds to folding in Schr\"odinger picture, so operators travel downwards). Formally, we achieve this by an operator-to-state (vectorization) mapping. One maps operators acting on $k$ consecutive qubits, i.e.\ acting on $\mathbb{C}^{d^k} \otimes \mathbb{C}^{d^k}$, to vectors in $\mathbb{C}^{d^{2k}}$ as
\be
a \overset{\rm vec}{\mapsto} \ket{a}\, .
\ee
\ch{Note that, from now on, we will exclusively use Dirac ket $\ket{\cdot}$ to represent vectorised operators. We will use the standard Hilbert-Schmidt scalar product
$\langle a|b\rangle = \tr a^\dagger b$, which also induces a norm.}
The mapping is specified by fixing a basis $\{\ket{n}\}$ of $\mathbb{C}^{d^k}$ which induces a basis $\{\ket{m}\bra{n}\}$ for  $\mathbb{C}^{d^k}\!\!\! \otimes \mathbb{C}^{d^k}$. Then, we define the following mapping of basis elements   
\be
\ket{m}\bra{n} \overset{\rm vec}{\mapsto} \ket{m} \otimes \ket{n}^* \,,
\ee
and extend it by linearity. In the following, we will use the following notation for the vectorised \ch{2-norm} normalised identity operator 
\ch{
\be
\ket{\mcirc}= \frac{1}{\sqrt{d}}  \ket{\1}=
\begin{tikzpicture}[baseline=(current  bounding  box.center), scale=.7]
\draw[thick] (0,0) -- (0,0.35);
\draw[thick, fill=white] (0,0) circle (0.1cm); 
\end{tikzpicture}
.
\label{eq:Id}
\ee 
}
After folding, we have the local gate $U_{x,\tau}$ and the time reversed gate $U_{x,\tau}^*$ ($(\cdot)^*$ denotes the complex conjugation in the canonical basis) at the same position. Therefore, we express the evolution in terms of ``doubled gates''
\begin{align}
\label{eq:doublegate}
W_{x,\tau} &=
\begin{tikzpicture}[baseline=(current  bounding  box.center), scale=.7]
\def\eps{0.5};
\Wgategreen{-3.75}{0};
\Text[x=-3.75,y=-0.6]{}
\end{tikzpicture}
=
\begin{tikzpicture}[baseline=(current  bounding  box.center), scale=0.7]
\draw[thick] (-1.65,0.65) -- (-0.65,-0.35);
\draw[thick] (-1.65,-0.35) -- (-0.65,0.65);
\draw[ thick, fill=myred, rounded corners=2pt] (-1.4,0.4) rectangle (-.9,-0.1);
\draw[thick] (-1.15,0) -- (-1,0) -- (-1,0.15);
\draw[thick] (-2.25,0.5) -- (-1.25,-0.5);
\draw[thick] (-2.25,-0.5) -- (-1.25,0.5);
\draw[ thick, fill=myblue, rounded corners=2pt] (-2,0.25) rectangle (-1.5,-0.25);
\draw[thick] (-1.75,0.15) -- (-1.6,0.15) -- (-1.6,0);
\Text[x=-2.25,y=-0.6]{}
\end{tikzpicture}
= U_{x,\tau}\otimes U_{x,\tau}^{*},\\
%U^\dagger \otimes U^{T}.\\
\label{eq:doublegatedu}
W_{x,\tau}^{\rm du} &=
\begin{tikzpicture}[baseline=(current  bounding  box.center), scale=.7]
\def\eps{0.5};
\Wgateorange{-3.75}{0};
\Text[x=-3.75,y=-0.6]{}
\end{tikzpicture}
=
\begin{tikzpicture}[baseline=(current  bounding  box.center), scale=0.7]
\draw[thick] (-1.65,0.65) -- (-0.65,-0.35);
\draw[thick] (-1.65,-0.35) -- (-0.65,0.65);
\draw[ thick, fill=myorange0c, rounded corners=2pt] (-1.4,0.4) rectangle (-.9,-0.1);
\draw[thick] (-1.15,0) -- (-1,0) -- (-1,0.15);
\draw[thick] (-2.25,0.5) -- (-1.25,-0.5);
\draw[thick] (-2.25,-0.5) -- (-1.25,0.5);
\draw[ thick, fill=myorange0, rounded corners=2pt] (-2,0.25) rectangle (-1.5,-0.25);
\draw[thick] (-1.75,0.15) -- (-1.6,0.15) -- (-1.6,0);
\Text[x=-2.25,y=-0.6]{}
\end{tikzpicture}
= 
(U_{x,\tau}^{\rm du}) \otimes (U_{x,\tau}^{\rm du})^{*},\\
Y_x &=
\begin{tikzpicture}[baseline=(current  bounding  box.center), scale=.65]
\draw[thick] (-2,0.5) -- (-1.5,0);
\draw[thick] (-1.5,0) -- (-1,0.5);
\draw[very thick] (-2.05,0) -- (-0.95,0);
\Text[x=-1.5,y=-0.5]{}
\draw[ thick, fill=mygreen, rounded corners=2pt] (-1.5-0.35,0.2-0.25) rectangle (-1.5+0.35,0.2+0.2);
\draw[thick] (-1.5+0.1,0.15+.18)-- (-1.35+0.1,0.15+.18) -- (-1.35+0.1,0+.18);
\end{tikzpicture}
=
\begin{tikzpicture}[baseline=(current  bounding  box.center), scale=0.7]
\draw[thick] (-2+0.35,0.5+0.2) -- (-1.5+0.35,0+0.2);
\draw[thick] (-1.5+0.35,0+0.2) -- (-1+0.35,0.5+0.2);
\draw[very thick] (-2.05+0.35,0+0.2) -- (-0.95+0.35,0+0.2);
\draw[ thick, fill=myred, rounded corners=2pt] (-1.5-0.35+0.35,0.2-0.25+0.2) rectangle (-1.5+0.35+0.35,0.2+0.2+0.2);
\draw[thick] (-1.4+0.35,0.05+0.2) -- (-1.25+0.35,0.05+0.2) -- (-1.25+0.35,0.2+0.2);
\draw[thick] (-2,0.5) -- (-1.5,0);
\draw[thick] (-1.5,0) -- (-1,0.5);
\draw[very thick] (-2.05,0) -- (-0.95,0);
\draw[ thick, fill=myblue, rounded corners=2pt] (-1.5-0.35,0.2-0.25) rectangle (-1.5+0.35,0.2+0.2);
\draw[thick] (-1.5+0.1,0.15+.18)-- (-1.35+0.1,0.15+.18) -- (-1.35+0.1,0+.18);
\Text[x=-2.25,y=-0.6]{}
\end{tikzpicture}
= T_x \otimes (T_x)^{\dag},\\
Y_x^{\rm s} &=
\begin{tikzpicture}[baseline=(current  bounding  box.center), scale=.65]
\draw[thick] (-2,0.5) -- (-1.5,0);
\draw[thick] (-1.5,0) -- (-1,0.5);
\draw[very thick] (-2.05,0) -- (-0.95,0);
\Text[x=-1.5,y=-0.5]{}
\draw[ thick, fill=myorange, rounded corners=2pt] (-1.5-0.35,0.2-0.25) rectangle (-1.5+0.35,0.2+0.2);
\draw[thick] (-1.5+0.1,0.15+.18)-- (-1.35+0.1,0.15+.18) -- (-1.35+0.1,0+.18);
\end{tikzpicture}
=
\begin{tikzpicture}[baseline=(current  bounding  box.center), scale=0.7]
\draw[thick] (-2+0.35,0.5+0.2) -- (-1.5+0.35,0+0.2);
\draw[thick] (-1.5+0.35,0+0.2) -- (-1+0.35,0.5+0.2);
\draw[very thick] (-2.05+0.35,0+0.2) -- (-0.95+0.35,0+0.2);
\draw[ thick, fill=myorange0c, rounded corners=2pt] (-1.5-0.35+0.35,0.2-0.25+0.2) rectangle (-1.5+0.35+0.35,0.2+0.2+0.2);
\draw[thick] (-1.4+0.35,0.05+0.2) -- (-1.25+0.35,0.05+0.2) -- (-1.25+0.35,0.2+0.2);
\draw[thick] (-2,0.5) -- (-1.5,0);
\draw[thick] (-1.5,0) -- (-1,0.5);
\draw[very thick] (-2.05,0) -- (-0.95,0);
\draw[ thick, fill=myorange0, rounded corners=2pt] (-1.5-0.35,0.2-0.25) rectangle (-1.5+0.35,0.2+0.2);
\draw[thick] (-1.5+0.1,0.15+.18)-- (-1.35+0.1,0.15+.18) -- (-1.35+0.1,0+.18);
\Text[x=-2.25,y=-0.6]{}
\end{tikzpicture}
= (T_x^{\rm s}) \otimes (T_x^{\rm s})^{\dag}.
\label{eq:Ydef}
\end{align} 
\ch{Explicitly, considering a two-site operator $O$, we have 
\be
W_{x,\tau}\ket{O}=  \ket{U_{x,t} O U_{x,t}^\dag}.
\ee}
The gates \eqref{eq:doublegate}--\eqref{eq:Ydef} result from folding in the Schr\"odinger picture and differ slightly from the Heisenberg gates in~\cite{kos2021correlations, reid2021entanglement}. 

Expressing \eqref{eq:rhoL} in terms of the above gates we find 
\bw
\be
\ket{\rho_{A,L}(t)} =\frac{1}{d^{L_B+\frac{1}{2}}}
\begin{tikzpicture}[baseline=(current  bounding  box.center), scale=0.6]
\foreach \i in {0,...,1}{
\Wgateorange{-2}{0+2*\i}\Wgategreen{0}{0+2*\i}\Wgategreen{2}{0+2*\i}\Wgategreen{4}{0+2*\i}\Wgateorange{6}{0+2*\i}\Wgateorange{8}{0+2*\i}
\Wgateorange{-1}{1+2*\i}\Wgategreen{1}{1+2*\i}\Wgategreen{3}{1+2*\i}\Wgategreen{5}{1+2*\i}\Wgateorange{7}{1+2*\i}\Wgateorange{9}{1+2*\i}
}
\draw[very thick] (-2.25,-1) -- (9.5,-1);
\foreach \i in {2.5,4.5,6.5}
{
\draw[very thick] (-2+\i,0.5-1) -- (-1.5+\i,0-1);
\draw[very thick] (-1.5+\i,0-1) -- (-1+\i,0.5-1);
\draw[ thick, fill=mygreen, rounded corners=2pt] (-1.5-0.35+\i,0.2-0.25-1) rectangle (-1.5+0.35+\i,0.2+0.2-1);
\draw[thick] (-1.5+0.1+\i,0.15+.18-1)-- (-1.35+0.1+\i,0.15+.18-1) -- (-1.35+0.1+\i,0+.18-1);}
\foreach \i in {0.5,8.5,10.5}
{
\draw[very thick] (-2+\i,0.5-1) -- (-1.5+\i,0-1);
\draw[very thick] (-1.5+\i,0-1) -- (-1+\i,0.5-1);
\draw[ thick, fill=myorange, rounded corners=2pt] (-1.5-0.35+\i,0.2-0.25-1) rectangle (-1.5+0.35+\i,0.2+0.2-1);
\draw[thick] (-1.5+0.1+\i,0.15+.18-1)-- (-1.35+0.1+\i,0.15+.18-1) -- (-1.35+0.1+\i,0+.18-1);}  
 \foreach \i in {6,...,10}{
 \draw[thick, fill=white] (\i-0.5,4-0.5) circle (0.1cm); 
}
\draw[thick, fill=white] (-1-0.5,4-0.5) circle (0.1cm); 
\draw[thick, fill=white] (-0-0.5,4-0.5) circle (0.1cm);
 %\draw[thick, fill=white] (0-0.25,-0.75) circle (0.1cm);
 %Denote E(t)
 \draw [thick, black,decorate,decoration={brace,amplitude=2.5pt,mirror},xshift=0.0pt,yshift=-0.0pt](7.6,-1.2) -- (9.6,-1.2) node[black,midway,yshift=-0.4cm] {$E_x(t)$};
\draw[dashed] (7.6,-1.2) -- (7.6,3.8);
\draw[dashed] (9.6,-1.2) -- (9.6,3.8);
\end{tikzpicture}.
\label{eq:rhodiagram}
\ee
\ew
The dual-unitary part of the circuit (traced out) can be expressed in terms of products of the dual-unitary space transfer matrix $E_x(t)$ depicted in Eq.~\eqref{eq:rhodiagram} \ch{(for a more explicit formula see~\cite{piroli2020exact})}. This matrix implements left-to-right multiplication and, crucially, has a unique fixed points, while all its other eigenvalues are strictly contained in the unit circle (cf.\ Ref.~\cite{piroli2020exact}). Moreover, right and left fixed points are written as $\ket{e_0}$ and $\bra{e_0}=(\ket{e_0})^\dag$ where 
\be
\ket{e_0}=\ket{{\mcirc}_\chi \mcirc \mcirc \dotsc \mcirc \mcirc},
\ee
and $\ket{{\mcirc}_\chi}$ is the vectorized identity state on the auxiliary space of the MPS. This means that in the limit \ch{of infinite bath, when} $L \to \infty$ with $L_A$ fixed, we have  
\be\begin{aligned}
&\ket{\rho_A(t)}  =
\begin{tikzpicture}[baseline=(current  bounding  box.center), scale=0.55]
\foreach \i in {0,...,1}{
\Wgategreen{0}{0+2*\i}\Wgategreen{2}{0+2*\i}\Wgategreen{4}{0+2*\i}
\Wgategreen{1}{1+2*\i}\Wgategreen{3}{1+2*\i}\Wgategreen{5}{1+2*\i}
\draw[thick, fill=white] (-0.5,-0.5+2*\i) circle (0.1cm); 
\draw[thick, fill=white] (-0.5,0.5+2*\i) circle (0.1cm); 
\draw[thick, fill=white] (5.5,1.5+2*\i) circle (0.1cm); 
\draw[thick, fill=white] (5.5,0.5+2*\i) circle (0.1cm);
}
\draw[very thick] (-0,-1) -- (6,-1);
\draw[thick, fill=white] (6.,-1) circle (0.15cm);
\draw[thick, fill=white] (-.0,-1) circle (0.15cm);
\Text[x=6+0.25,y=-1-0.35]{\scriptsize$\chi$}
\Text[x=-0-0.25,y=-1-0.35]{\scriptsize$\chi$}
\draw[thick] (5.5,-0.5)--(5.5,-0.25);
\draw[thick, fill=white] (5.5,-0.25) circle (0.1cm);
\foreach \i in {2.5,4.5,6.5}
{
\draw[very thick] (-2+\i,0.5-1) -- (-1.5+\i,0-1);
\draw[very thick] (-1.5+\i,0-1) -- (-1+\i,0.5-1);
\draw[ thick, fill=mygreen, rounded corners=2pt] (-1.5-0.35+\i,0.2-0.25-1) rectangle (-1.5+0.35+\i,0.2+0.2-1);
\draw[thick] (-1.5+0.1+\i,0.15+.18-1)-- (-1.35+0.1+\i,0.15+.18-1) -- (-1.35+0.1+\i,0+.18-1);}
\Text[x=6+0.5,y=-0-0.35]{,}
\end{tikzpicture}
\end{aligned}
\ee
\ch{where empty circles denote the vectorised identity operators from Eq.~\eqref{eq:Id}.}
We see that this expression can be written as a boundary driven evolution for the reduced density matrix $\ket{\rho_A(t)}$, i.e.
\be
\ket{\rho_A(t)} = (\mathbb B_A)^t \ket{\rho_A(0)}\, .
\label{eq:Markovrho}
\ee
Here
\be
\ket{\rho_A(0)}  =
\begin{tikzpicture}[baseline=(current  bounding  box.center), scale=0.55]
\draw[very thick] (-0,-1) -- (6,-1);
\draw[thick, fill=white] (6.,-1) circle (0.15cm);
\draw[thick, fill=white] (-.0,-1) circle (0.15cm);
\Text[x=6+0.25,y=-1-0.35]{\scriptsize$\chi$}
\Text[x=-0-0.25,y=-1-0.35]{\scriptsize$\chi$}
\foreach \i in {2.5,4.5,6.5}
{
\draw[very thick] (-2+\i,0.5-1) -- (-1.5+\i,0-1);
\draw[very thick] (-1.5+\i,0-1) -- (-1+\i,0.5-1);
\draw[ thick, fill=mygreen, rounded corners=2pt] (-1.5-0.35+\i,0.2-0.25-1) rectangle (-1.5+0.35+\i,0.2+0.2-1);
\draw[thick] (-1.5+0.1+\i,0.15+.18-1)-- (-1.35+0.1+\i,0.15+.18-1) -- (-1.35+0.1+\i,0+.18-1);}
\end{tikzpicture}
\ee
denotes the thermodynamic limit of the reduced density matrix at time $t=0$. We introduced the \emph{boundary time evolution map} $\mathbb B_A$, a many-body quantum channel defined as 
\begin{align}
\!\!\!\!\mathbb B_A &= 
\begin{tikzpicture}[baseline=(current  bounding  box.center), scale=0.55]
\Wgategreen{0}{0}\Wgategreen{2}{0}\Wgategreen{4}{0}
\Wgategreen{1}{1}\Wgategreen{3}{1}\Wgategreen{5}{1}
\draw[thick, fill=white] (-0.5,-0.5) circle (0.1cm); 
\draw[thick, fill=white] (-0.5,0.5) circle (0.1cm); 
\draw[thick, fill=white] (5.5,1.5) circle (0.1cm); 
\draw[thick, fill=white] (5.5,0.5) circle (0.1cm); 
\draw [thick, black,decorate,decoration={brace,amplitude=5pt,mirror},xshift=0.0pt,yshift=-0.0pt](4.5,1.6) -- (.5,1.6) node[black,midway,yshift=0.4cm] { $A$};
\end{tikzpicture}\notag\\
&=  \left[\bigotimes_{x}  W_{x,1/2} \!\otimes\! m_{L_A-{1}/{2},r}\right]\!\!\cdot\!\!\left[m_{{1}/{2},l} \!\otimes \!\bigotimes_{x}  W_{x,1}\right]\!\!,
\label{eq:Bmap}
\end{align}
with boundary quantum channels denoted by 
\be
m_{{1}/{2},l}=\begin{tikzpicture}[baseline=(current  bounding  box.center), scale=0.7]
\Wgategreen{0}{0}\draw[thick, fill=white] (-0.5,-0.5) circle (0.1cm);
\draw[thick, fill=white] (-0.5,0.5) circle (0.1cm); 
\end{tikzpicture} \, ,
\qquad
m_{L_A-{1}/{2},r}=\begin{tikzpicture}[baseline=(current  bounding  box.center), scale=0.7]
\Wgategreen{0}{0}
\draw[thick, fill=white] (0.5,-0.5) circle (0.1cm);
\draw[thick, fill=white] (0.5,0.5) circle (0.1cm); 
\end{tikzpicture}\, ,
\ee
where the first index depicts the position on which the channel act%[PK: we could inherit the labels from the gate, but it would be different for the $r$ channels...]
, while the second labels a left or right edge channel.

Firstly we remark that $\mathbb B_A $ is \ch{a unital\footnote{It maps identity to identity.} completely positive trace-preserving (CPT)} map, i.e., a legitimate quantum channel. This can be seen immediately by noting that it is directly written in the so called environmental representation with the environment in the maximally mixed state~\cite{bengtsson2008geometry}. Secondly, we note that, since the evolution of $\ket{\rho_A(t)}$ does not depend on the state of the bath, Eq.~\eqref{eq:Markovrho} implies that the bath formed by the subsystem $B$ is Markovian~\cite{lerose2020influence}. This is a consequence of maximally fast thermalisation of dual-unitary circuits evolving from solvable MPSs~\cite{piroli2020exact}. Interestingly, it has been recently shown that tracing out the rest of the system one can also achieve concrete realisations of non-Markovian baths~\cite{klobas2021exact, klobas2021exactII}. 

Before discussing the main properties of the many-body map $\mathbb B_A$, let us show that it is relevant also for an \ch{\emph{ostensibly}} unrelated quantity, i.e., the spectral form factor of the time evolution operator \eqref{eq:diagramFloquet} in the thermodynamic limit.  

%-----------------------------------------
\subsection{Spectral form factor}
%-----------------------------------------

%We know that the spectral form factor (SFF) follows the RMT result exactly in dual-unitary models \cite{BKP:kickedIsing,BKP:SFFDU}. A natural question arises what happens if a finite part of the circuit is, in fact, not dual-unitary. Similar to the previous subsection, the Floquet operator is given by $2L$ general gates and  $2L$ dual-unitary gates. As the SFF is not self-averaging, we perform local ${\rm U}(1)$ averaging on the dual-unitary part of the circuit, the same as in \cite{BKP:SFFDU}.

The spectral form factor of a given Floquet operator $\mathbb U_{ L}$ is defined as 
\be
K(t,L) = \mathbb E\left[|{\rm tr}[\mathbb U_{L}^t]|^2\right]
\ee
where $\mathbb E[\cdot]$ denotes an average either over time or over an ensemble of similar systems. Here we follow Ref.~\cite{bertini2021random} and consider an average over an ensemble of similar systems obtained by introducing independently-distributed random single-site operators $\{u_{x,\tau}\}$ on the dual-unitary part of the circuit (for more details see Ref.~\cite{bertini2021random}). Namely we consider 
\be
\mathbb E[\cdot]= \mathbb E_{\{u_{x,\tau}\}}[\cdot],\qquad \mathbb U_{L} \mapsto \mathbb U_{L, \{u_{x,\tau}\}},
\ee
with 
\be
\mathbb U_{L,\{u_{x,\tau}\}}\!\! = %\prod_{x \in \mathbb Z_{L}} \eta_{x,L}(U_{x,1}) \prod_{x \in \mathbb Z_{L}+\feac{1}{2}} \eta_{x,L}(U_{x,1/2}) =
\begin{tikzpicture}[baseline=(current  bounding  box.center), scale=0.55]
\draw[thick, dotted] (-9.5,-1.7) -- (0.4,-1.7);
\draw[thick, dotted] (-9.5,-1.3) -- (0.4,-1.3);
\foreach \i in {3,...,13}{
\draw[gray, dashed] (-12.5+\i,-2.1) -- (-12.5+\i,0.3);
}
\foreach \i in {3,...,5}{
\draw[gray, dashed] (-9.75,3-\i) -- (.75,3-\i);
}
\foreach \i in{0.5}
{
\draw[ thick] (0.5,1+2*\i-0.5-3.5) arc (-45:90:0.17);
\draw[ thick] (-10+0.5,1+2*\i-0.5-3.5) arc (270:90:0.15);
}
\foreach \i in{1.5}{
\draw[thick] (0.5,2*\i-0.5-3.5) arc (45:-90:0.17);
\draw[thick] (-10+0.5+0,2*\i-0.5-3.5) arc (90:270:0.15);
}
%\foreach \i in {0,1}{
%\Text[x=1.25,y=-2+2*\i]{\scriptsize$\i$}
%}
%\foreach \i in {1}{
%\Text[x=1.25,y=-2+\i]{\small$\frac{\i}{2}$}
%}
%\foreach \i in {1,3,5}{
%\Text[x=-7.5+\i+1-3,y=-2.6]{\small$\frac{\i}{2}$}
%}
%\foreach \i in {1,2,3}{
%\Text[x=-7.5+2*\i-2,y=-2.6]{\scriptsize${\i}$}
%}
%NonDU gates
\foreach \jj[evaluate=\jj as \j using -2*(ceil(\jj/2)-\jj/2)] in {0}
\foreach \i in {3,...,5}
{
\draw[thick] (.5-2*\i-1*\j,-2-1*\jj) -- (1-2*\i-1*\j,-1.5-\jj);
\draw[thick] (1-2*\i-1*\j,-1.5-1*\jj) -- (1.5-2*\i-1*\j,-2-\jj);
\draw[thick] (.5-2*\i-1*\j,-1-1*\jj) -- (1-2*\i-1*\j,-1.5-\jj);
\draw[thick] (1-2*\i-1*\j,-1.5-1*\jj) -- (1.5-2*\i-1*\j,-1-\jj);
\draw[thick, fill=myblue, rounded corners=2pt] (0.75-2*\i-1*\j,-1.75-\jj) rectangle (1.25-2*\i-1*\j,-1.25-\jj);
\draw[thick] (-2*\i+1,-1.35-\jj) -- (-2*\i+1.15,-1.35-\jj) -- (-2*\i+1.15,-1.5-\jj);%
}
\foreach \i in {2,...,4}{
\draw[thick] (-.5-2*\i,-1) -- (0.525-2*\i,0.025);
\draw[thick] (-0.525-2*\i,0.025) -- (0.5-2*\i,-1);
\draw[thick, fill=myblue, rounded corners=2pt] (-0.25-2*\i,-0.25) rectangle (.25-2*\i,-0.75);
\draw[thick] (-2*\i,-0.35) -- (-2*\i+0.15,-0.35) -- (-2*\i+0.15,-0.5);%
}
%DU gates
\foreach \i in {0,...,1}{
\draw[thick] (-.5-2*\i,-1) -- (0.525-2*\i,0.025);
\draw[thick] (-0.525-2*\i,0.025) -- (0.5-2*\i,-1);
\draw[thick, fill=myorange0, rounded corners=2pt] (-0.25-2*\i,-0.25) rectangle (.25-2*\i,-0.75);
\draw[thick] (-2*\i,-0.35) -- (-2*\i+0.15,-0.35) -- (-2*\i+0.15,-0.5);%
}
\foreach \jj[evaluate=\jj as \j using -2*(ceil(\jj/2)-\jj/2)] in {0}
\foreach \i in {1,...,2}
{
\draw[thick] (.5-2*\i-1*\j,-2-1*\jj) -- (1-2*\i-1*\j,-1.5-\jj);
\draw[thick] (1-2*\i-1*\j,-1.5-1*\jj) -- (1.5-2*\i-1*\j,-2-\jj);
\draw[thick] (.5-2*\i-1*\j,-1-1*\jj) -- (1-2*\i-1*\j,-1.5-\jj);
\draw[thick] (1-2*\i-1*\j,-1.5-1*\jj) -- (1.5-2*\i-1*\j,-1-\jj);
\draw[thick, fill=myorange0, rounded corners=2pt] (0.75-2*\i-1*\j,-1.75-\jj) rectangle (1.25-2*\i-1*\j,-1.25-\jj);
\draw[thick] (-2*\i+1,-1.35-\jj) -- (-2*\i+1.15,-1.35-\jj) -- (-2*\i+1.15,-1.5-\jj);%
}
\draw[ thick, fill=myred2, rounded corners=2pt] (0.5,-1) circle (.15);
\draw[ thick, fill=myred2, rounded corners=2pt] (-0.5,-1) circle (.15);
\draw[ thick, fill=myred2, rounded corners=2pt] (-1.5,-1) circle (.15);
\draw[ thick, fill=myred2, rounded corners=2pt] (-2.5,-1) circle (.15);
%\draw[ thick, fill=myred5, rounded corners=2pt] (-3.5,-1) circle (.15);
\draw[ thick, fill=myred2, rounded corners=2pt] (0.5,0) circle (.15);
\draw[ thick, fill=myred2, rounded corners=2pt] (-0.5,0) circle (.15);
\draw[ thick, fill=myred2, rounded corners=2pt] (-1.5,0) circle (.15);
\draw[ thick, fill=myred2, rounded corners=2pt] (-2.5,0) circle (.15);
%Hooks on circles?
\foreach \i in {0,...,1}{
\draw[ thick] (-2*\i+0.4,0.05) -- (-2*\i+0.55,0.05) -- (-2*\i+0.55,-0.1);
\draw[ thick] (-2*\i+0.4-1,0.05-1) -- (-2*\i+0.55-1,0.05-1) -- (-2*\i+0.55-1,-0.1-1);
\foreach \jj in {0,...,1}{
\draw[ thick]  (-2*\i+1.45-2+\jj,-0.1-\jj) -- (-2*\i+1.45-2+\jj,0.05-\jj) -- (-2*\i+1.6-2+\jj,0.05-\jj);
}
}
\foreach \jj[evaluate=\jj as \j using -2*(ceil(\jj/2)-\jj/2)] in {0}
\foreach \i in {1,...,5}{
%\draw[thick] (-2*\i+1.4,-.95-\jj) -- (-2*\i+1.55,-.95-\jj) -- (-2*\i+1.55,-1.1-\jj);
}
\foreach \jj[evaluate=\jj as \j using -2*(ceil(\jj/2)-\jj/2)] in {0}
\foreach \i in {1,...,5}{
%\draw[thick] (-2*\i+2.4,0.05-\jj) -- (-2*\i+2.55,0.05-\jj) -- (-2*\i+2.55,-0.1-\jj);
}
\foreach \jj[evaluate=\jj as \j using -2*(ceil(\jj/2)-\jj/2)] in {0}
\foreach \i in {1,...,5}{
%\draw[thick]  (-2*\i+1.45,-0.1-\jj) -- (-2*\i+1.45,0.05-\jj) -- (-2*\i+1.6,0.05-\jj);
}
\foreach \jj[evaluate=\jj as \j using -2*(ceil(\jj/2)-\jj/2)] in {0}
\foreach \i in {0,...,4}{
%\draw[thick]  (-2*\i+.45,-1.1-\jj) -- (-2*\i+.45,-0.95-\jj) -- (-2*\i+.6,-0.95-\jj);
}
%\Text[x=-2,y=-2.6]{$\cdots$}
%\draw [thick, black,decorate,decoration={brace,amplitude=5pt,mirror},xshift=0.0pt,yshift=-0.0pt](-9.5,-2.2) -- (-4.5,-2.2) node[black,midway,yshift=-0.375cm] { $A$};
%\draw [thick, black,decorate,decoration={brace,amplitude=5pt,mirror},xshift=0.0pt,yshift=-0.0pt](-3.5,-2.2) -- (.5,-2.2) node[black,midway,yshift=-0.375cm] { $B$};
\end{tikzpicture}.
\label{eq:diagramFloquet2}
\ee
where blue and yellow denote again \ch{possibly inhomogenous arbitrary and homogenous dual-unitary} gates respectively, while red circles represent \ch{time periodic} on-site disorder \ch{$\{u_{x,\tau}\}$}.

For this choice of factorised-in-space average the spectral form factor is represented as 
\begin{align}
\label{eq:SFFdiag}
\!\!\!\!\!\!\!K(t,L)\!=&
\begin{tikzpicture}[baseline=(current  bounding  box.center), scale=0.5]
\foreach \i in {1,...,3}{
\draw[very thick, dotted] (2*\i+2-12.5+0.255,-1.75-0.1) -- (2*\i+2-12.5+0.255,4.25-0.1);
\draw[very thick, dotted] (2*\i+2-11.5-0.255,-1.75-0.1) -- (2*\i+2-11.5-0.255,4.25-0.1);
}
\foreach \i in {1,...,3}{
\draw[very thick] (2*\i+2-11.5,4) arc (-45:175:0.15);
\draw[very thick] (2*\i+2-11.5,-2) arc (315:180:0.15);
\draw[very thick] (2*\i+2-0.5-12,-2) arc (-135:0:0.15);
}
\foreach \i in {2,...,4}{
\draw[very thick] (2*\i+2-2.5-12,4) arc (225:0:0.15);
}
\foreach \i in{1.5,2.5,3.5}{
%\draw[very thick] (0.5,2*\i-0.5-3.5) arc (45:-90:0.15);
\draw[very thick] (-10+0.5+0,2*\i-0.5-3.5) arc (45:270:0.15);
}
\foreach \i in{0.5,1.5,2.5}
{
%\draw[very thick] (0.5,1+2*\i-0.5-3.5) arc (-45:90:0.15);
\draw[very thick] (-10+0.5,1+2*\i-0.5-3.5) arc (315:90:0.15);
}
\foreach \jj[evaluate=\jj as \j using -2*(ceil(\jj/2)-\jj/2)] in {-1,-3,-5}{
\foreach \i in {3,...,5}
{
\draw[very thick] (.5-2*\i-1*\j,-2-1*\jj) -- (1-2*\i-1*\j,-1.5-\jj);
\draw[very thick] (1-2*\i-1*\j,-1.5-1*\jj) -- (1.5-2*\i-1*\j,-2-\jj);
\draw[very thick] (.5-2*\i-1*\j,-1-1*\jj) -- (1-2*\i-1*\j,-1.5-\jj);
\draw[very thick] (1-2*\i-1*\j,-1.5-1*\jj) -- (1.5-2*\i-1*\j,-1-\jj);
\draw[thick, fill=mygreen, rounded corners=2pt] (0.75-2*\i-1*\j,-1.75-\jj) rectangle (1.25-2*\i-1*\j,-1.25-\jj);
\draw[ thick] (-2*\i+2,-1.35-\jj) -- (-2*\i+2.15,-1.35-\jj) -- (-2*\i+2.15,-1.5-\jj);%
}
}
\foreach \jj[evaluate=\jj as \j using -2*(ceil(\jj/2)-\jj/2)] in {-4,-2,0}{
\foreach \i in {3,...,5}
{
\draw[very thick] (.5-2*\i-1*\j,-2-1*\jj) -- (1-2*\i-1*\j,-1.5-\jj);
\draw[very thick] (1-2*\i-1*\j,-1.5-1*\jj) -- (1.5-2*\i-1*\j,-2-\jj);
\draw[very thick] (.5-2*\i-1*\j,-1-1*\jj) -- (1-2*\i-1*\j,-1.5-\jj);
\draw[very thick] (1-2*\i-1*\j,-1.5-1*\jj) -- (1.5-2*\i-1*\j,-1-\jj);
\draw[thick, fill=mygreen, rounded corners=2pt] (0.75-2*\i-1*\j,-1.75-\jj) rectangle (1.25-2*\i-1*\j,-1.25-\jj);
\draw[ thick] (-2*\i+1,-1.35-\jj) -- (-2*\i+1.15,-1.35-\jj) -- (-2*\i+1.15,-1.5-\jj);%
}
}
\end{tikzpicture}
\mathbb{E}\Bigl[\begin{tikzpicture}[baseline=(current  bounding  box.center), scale=0.5]
\foreach \i in {5}{
\draw[very thick, dotted] (2*\i+2-12.5+0.255,-1.75-0.1) -- (2*\i+2-12.5+0.255,4.25-0.1);
\draw[very thick, dotted] (2*\i+2-11.5-0.255,-1.75-0.1) -- (2*\i+2-11.5-0.255,4.25-0.1);}
\foreach \i in {5}{
\draw[very thick] (2*\i+2-11.5,4) arc (-45:175:0.15);
\draw[very thick] (2*\i+2-11.5,-2) arc (315:180:0.15);
\draw[very thick] (2*\i+2-0.5-12,-2) arc (-135:0:0.15);
}
\foreach \i in {6}{
\draw[very thick] (2*\i+2-2.5-12,4) arc (225:0:0.15);
}
\foreach \jj[evaluate=\jj as \j using -2*(ceil(\jj/2)-\jj/2)] in {-1,-3,-5}{
\foreach \i in {1}
{
\draw[very thick] (.5-2*\i-1*\j,-2-1*\jj) -- (1-2*\i-1*\j,-1.5-\jj);
\draw[very thick] (1-2*\i-1*\j,-1.5-1*\jj) -- (1.5-2*\i-1*\j,-2-\jj);
\draw[very thick] (.5-2*\i-1*\j,-1-1*\jj) -- (1-2*\i-1*\j,-1.5-\jj);
\draw[very thick] (1-2*\i-1*\j,-1.5-1*\jj) -- (1.5-2*\i-1*\j,-1-\jj);
\draw[thick, fill=myorange, rounded corners=2pt] (0.75-2*\i-1*\j,-1.75-\jj) rectangle (1.25-2*\i-1*\j,-1.25-\jj);
\draw[ thick] (-2*\i+2,-1.35-\jj) -- (-2*\i+2.15,-1.35-\jj) -- (-2*\i+2.15,-1.5-\jj);%
}
}
\foreach \jj[evaluate=\jj as \j using -2*(ceil(\jj/2)-\jj/2)] in {-4,-2,0}{
\foreach \i in {1}
{
\draw[very thick] (.5-2*\i-1*\j,-2-1*\jj) -- (1-2*\i-1*\j,-1.5-\jj);
\draw[very thick] (1-2*\i-1*\j,-1.5-1*\jj) -- (1.5-2*\i-1*\j,-2-\jj);
\draw[very thick] (.5-2*\i-1*\j,-1-1*\jj) -- (1-2*\i-1*\j,-1.5-\jj);
\draw[very thick] (1-2*\i-1*\j,-1.5-1*\jj) -- (1.5-2*\i-1*\j,-1-\jj);
\draw[thick, fill=myorange, rounded corners=2pt] (0.75-2*\i-1*\j,-1.75-\jj) rectangle (1.25-2*\i-1*\j,-1.25-\jj);
\draw[ thick] (-2*\i+1,-1.35-\jj) -- (-2*\i+1.15,-1.35-\jj) -- (-2*\i+1.15,-1.5-\jj);%
}
}
\foreach \jj in {0,2,4}{
\draw[ thick, fill=myyellow9, rounded corners=2pt] (0.5,-1+\jj) circle (.15);
\draw[ thick, fill=myyellow9, rounded corners=2pt] (-0.5,-1+\jj) circle (.15);
\draw[ thick, fill=myyellow9, rounded corners=2pt] (0.5,\jj) circle (.15);
\draw[ thick, fill=myyellow9, rounded corners=2pt] (-0.5,\jj) circle (.15);
}
\foreach \jj in {0,-2,-4}{
\foreach \i in {1}{
\draw[ thick] (-2*\i+1.4,-.95-\jj) -- (-2*\i+1.55,-.95-\jj) -- (-2*\i+1.55,-1.1-\jj);}
\foreach \i in {1}{
\draw[ thick] (-2*\i+2.4,0.05-\jj) -- (-2*\i+2.55,0.05-\jj) -- (-2*\i+2.55,-0.1-\jj);}
\foreach \i in {1}{
\draw[ thick]  (-2*\i+1.45,-0.1-\jj) -- (-2*\i+1.45,0.05-\jj) -- (-2*\i+1.6,0.05-\jj);}
\foreach \i in {0}{
\draw[ thick]  (-2*\i+.45,-1.1-\jj) -- (-2*\i+.45,-0.95-\jj) -- (-2*\i+.6,-0.95-\jj);}
}
\end{tikzpicture}\Bigr]\mathbb{E}\Bigl[\begin{tikzpicture}[baseline=(current  bounding  box.center), scale=0.5]
\foreach \i in {5}{
\draw[thick, dotted] (2*\i+2-12.5+0.255,-1.75-0.1) -- (2*\i+2-12.5+0.255,4.25-0.1);
\draw[thick, dotted] (2*\i+2-11.5-0.255,-1.75-0.1) -- (2*\i+2-11.5-0.255,4.25-0.1);}
\foreach \i in {5}{
\draw[very thick] (2*\i+2-11.5,4) arc (-45:175:0.15);
\draw[very thick] (2*\i+2-11.5,-2) arc (315:180:0.15);
\draw[very thick] (2*\i+2-0.5-12,-2) arc (-135:0:0.15);
}
\foreach \i in {6}{
\draw[very thick] (2*\i+2-2.5-12,4) arc (225:0:0.15);
}
\foreach \jj[evaluate=\jj as \j using -2*(ceil(\jj/2)-\jj/2)] in {-1,-3,-5}{
\foreach \i in {1}
{
\draw[very thick] (.5-2*\i-1*\j,-2-1*\jj) -- (1-2*\i-1*\j,-1.5-\jj);
\draw[very thick] (1-2*\i-1*\j,-1.5-1*\jj) -- (1.5-2*\i-1*\j,-2-\jj);
\draw[very thick] (.5-2*\i-1*\j,-1-1*\jj) -- (1-2*\i-1*\j,-1.5-\jj);
\draw[very thick] (1-2*\i-1*\j,-1.5-1*\jj) -- (1.5-2*\i-1*\j,-1-\jj);
\draw[thick, fill=myorange, rounded corners=2pt] (0.75-2*\i-1*\j,-1.75-\jj) rectangle (1.25-2*\i-1*\j,-1.25-\jj);
\draw[ thick] (-2*\i+2,-1.35-\jj) -- (-2*\i+2.15,-1.35-\jj) -- (-2*\i+2.15,-1.5-\jj);%
}
}
\foreach \jj[evaluate=\jj as \j using -2*(ceil(\jj/2)-\jj/2)] in {-4,-2,0}{
\foreach \i in {1}
{
\draw[very thick] (.5-2*\i-1*\j,-2-1*\jj) -- (1-2*\i-1*\j,-1.5-\jj);
\draw[very thick] (1-2*\i-1*\j,-1.5-1*\jj) -- (1.5-2*\i-1*\j,-2-\jj);
\draw[very thick] (.5-2*\i-1*\j,-1-1*\jj) -- (1-2*\i-1*\j,-1.5-\jj);
\draw[very thick] (1-2*\i-1*\j,-1.5-1*\jj) -- (1.5-2*\i-1*\j,-1-\jj);
\draw[thick, fill=myorange, rounded corners=2pt] (0.75-2*\i-1*\j,-1.75-\jj) rectangle (1.25-2*\i-1*\j,-1.25-\jj);
\draw[ thick] (-2*\i+1,-1.35-\jj) -- (-2*\i+1.15,-1.35-\jj) -- (-2*\i+1.15,-1.5-\jj);%
}
}
\foreach \jj in {0,2,4}{
\draw[ thick, fill=myyellow9, rounded corners=2pt] (0.5,-1+\jj) circle (.15);
\draw[ thick, fill=myyellow9, rounded corners=2pt] (-0.5,-1+\jj) circle (.15);
\draw[ thick, fill=myyellow9, rounded corners=2pt] (0.5,\jj) circle (.15);
\draw[ thick, fill=myyellow9, rounded corners=2pt] (-0.5,\jj) circle (.15);
}
\foreach \jj in {0,-2,-4}{
\foreach \i in {1}{
\draw[ thick] (-2*\i+1.4,-.95-\jj) -- (-2*\i+1.55,-.95-\jj) -- (-2*\i+1.55,-1.1-\jj);}
\foreach \i in {1}{
\draw[ thick] (-2*\i+2.4,0.05-\jj) -- (-2*\i+2.55,0.05-\jj) -- (-2*\i+2.55,-0.1-\jj);}
\foreach \i in {1}{
\draw[ thick]  (-2*\i+1.45,-0.1-\jj) -- (-2*\i+1.45,0.05-\jj) -- (-2*\i+1.6,0.05-\jj);}
\foreach \i in {0}{
\draw[ thick]  (-2*\i+.45,-1.1-\jj) -- (-2*\i+.45,-0.95-\jj) -- (-2*\i+.6,-0.95-\jj);}
}
\end{tikzpicture}\Bigr]\!
\begin{tikzpicture}[baseline=(current  bounding  box.center), scale=0.5]
\Text[x=0.5,y=4.15]{}
\foreach \i in{1.5,2.5,3.5}{
\draw[very thick] (0.5,2*\i-0.5-3.5) arc (45:-90:0.15);
%\draw[very thick] (-10+0.5+0,2*\i-0.5-3.5) arc (45:270:0.15);
}
\foreach \i in{0.5,1.5,2.5}
{
\draw[very thick] (0.5,1+2*\i-0.5-3.5) arc (-45:90:0.15);
%\draw[very thick] (-10+0.5,1+2*\i-0.5-3.5) arc (315:90:0.15);
}
\end{tikzpicture}\notag\\
&= {\rm tr}\left[ \left(\prod_{x=1}^{L_A}\mathbb{T}_x \right) \mathbb{\widetilde{T}}^{L_B}\right],
\end{align}
where white circles \ch{with arrows} denote folded disorder gates \ch{$u_{x,\tau} \otimes u_{x,\tau}^*$ (not to be confused with identity \eqref{eq:Id}),} and we introduced the transfer matrices 
\be
\mathbb{T}_x =\begin{tikzpicture}[baseline=(current  bounding  box.center), scale=0.5]
\foreach \i in {5}{
\draw[thick, dotted] (2*\i+2-12.5+0.255,-1.75-0.1) -- (2*\i+2-12.5+0.255,4.25-0.1);
\draw[thick, dotted] (2*\i+2-11.5-0.255,-1.75-0.1) -- (2*\i+2-11.5-0.255,4.25-0.1);}

\foreach \i in {5}{
\draw[very thick] (2*\i+2-11.5,4) arc (-45:175:0.15);
\draw[very thick] (2*\i+2-11.5,-2) arc (315:180:0.15);
\draw[very thick] (2*\i+2-0.5-12,-2) arc (-135:0:0.15);
}
\foreach \i in {6}{
\draw[very thick] (2*\i+2-2.5-12,4) arc (225:0:0.15);
}
\foreach \jj[evaluate=\jj as \j using -2*(ceil(\jj/2)-\jj/2)] in {-1,-3,-5}{
\foreach \i in {1}
{
\draw[very thick] (.5-2*\i-1*\j,-2-1*\jj) -- (1-2*\i-1*\j,-1.5-\jj);
\draw[very thick] (1-2*\i-1*\j,-1.5-1*\jj) -- (1.5-2*\i-1*\j,-2-\jj);
\draw[very thick] (.5-2*\i-1*\j,-1-1*\jj) -- (1-2*\i-1*\j,-1.5-\jj);
\draw[very thick] (1-2*\i-1*\j,-1.5-1*\jj) -- (1.5-2*\i-1*\j,-1-\jj);
\draw[thick, fill=mygreen, rounded corners=2pt] (0.75-2*\i-1*\j,-1.75-\jj) rectangle (1.25-2*\i-1*\j,-1.25-\jj);
\draw[ thick] (-2*\i+2,-1.35-\jj) -- (-2*\i+2.15,-1.35-\jj) -- (-2*\i+2.15,-1.5-\jj);%
}
}
\foreach \jj[evaluate=\jj as \j using -2*(ceil(\jj/2)-\jj/2)] in {-4,-2,0}{
\foreach \i in {1}
{
\draw[very thick] (.5-2*\i-1*\j,-2-1*\jj) -- (1-2*\i-1*\j,-1.5-\jj);
\draw[very thick] (1-2*\i-1*\j,-1.5-1*\jj) -- (1.5-2*\i-1*\j,-2-\jj);
\draw[very thick] (.5-2*\i-1*\j,-1-1*\jj) -- (1-2*\i-1*\j,-1.5-\jj);
\draw[very thick] (1-2*\i-1*\j,-1.5-1*\jj) -- (1.5-2*\i-1*\j,-1-\jj);
\draw[thick, fill=mygreen, rounded corners=2pt] (0.75-2*\i-1*\j,-1.75-\jj) rectangle (1.25-2*\i-1*\j,-1.25-\jj);
\draw[ thick] (-2*\i+1,-1.35-\jj) -- (-2*\i+1.15,-1.35-\jj) -- (-2*\i+1.15,-1.5-\jj);%
}
}
\end{tikzpicture}\,,
\qquad
\mathbb{\widetilde{T}} =\mathbb{E}\Bigl[\begin{tikzpicture}[baseline=(current  bounding  box.center), scale=0.5]
\foreach \i in {5}{
\draw[thick, dotted] (2*\i+2-12.5+0.255,-1.75-0.1) -- (2*\i+2-12.5+0.255,4.25-0.1);
\draw[thick, dotted] (2*\i+2-11.5-0.255,-1.75-0.1) -- (2*\i+2-11.5-0.255,4.25-0.1);}

\foreach \i in {5}{
\draw[very thick] (2*\i+2-11.5,4) arc (-45:175:0.15);
\draw[very thick] (2*\i+2-11.5,-2) arc (315:180:0.15);
\draw[very thick] (2*\i+2-0.5-12,-2) arc (-135:0:0.15);
}
\foreach \i in {6}{
\draw[very thick] (2*\i+2-2.5-12,4) arc (225:0:0.15);
}
\foreach \jj[evaluate=\jj as \j using -2*(ceil(\jj/2)-\jj/2)] in {-1,-3,-5}{
\foreach \i in {1}
{
\draw[very thick] (.5-2*\i-1*\j,-2-1*\jj) -- (1-2*\i-1*\j,-1.5-\jj);
\draw[very thick] (1-2*\i-1*\j,-1.5-1*\jj) -- (1.5-2*\i-1*\j,-2-\jj);
\draw[very thick] (.5-2*\i-1*\j,-1-1*\jj) -- (1-2*\i-1*\j,-1.5-\jj);
\draw[very thick] (1-2*\i-1*\j,-1.5-1*\jj) -- (1.5-2*\i-1*\j,-1-\jj);
\draw[thick, fill=myorange, rounded corners=2pt] (0.75-2*\i-1*\j,-1.75-\jj) rectangle (1.25-2*\i-1*\j,-1.25-\jj);
\draw[ thick] (-2*\i+2,-1.35-\jj) -- (-2*\i+2.15,-1.35-\jj) -- (-2*\i+2.15,-1.5-\jj);%
}
}
\foreach \jj[evaluate=\jj as \j using -2*(ceil(\jj/2)-\jj/2)] in {-4,-2,0}{
\foreach \i in {1}
{
\draw[very thick] (.5-2*\i-1*\j,-2-1*\jj) -- (1-2*\i-1*\j,-1.5-\jj);
\draw[very thick] (1-2*\i-1*\j,-1.5-1*\jj) -- (1.5-2*\i-1*\j,-2-\jj);
\draw[very thick] (.5-2*\i-1*\j,-1-1*\jj) -- (1-2*\i-1*\j,-1.5-\jj);
\draw[very thick] (1-2*\i-1*\j,-1.5-1*\jj) -- (1.5-2*\i-1*\j,-1-\jj);
\draw[thick, fill=myorange, rounded corners=2pt] (0.75-2*\i-1*\j,-1.75-\jj) rectangle (1.25-2*\i-1*\j,-1.25-\jj);
\draw[ thick] (-2*\i+1,-1.35-\jj) -- (-2*\i+1.15,-1.35-\jj) -- (-2*\i+1.15,-1.5-\jj);%
}
}
\foreach \jj in {0,2,4}{
\draw[ thick, fill=myyellow9, rounded corners=2pt] (0.5,-1+\jj) circle (.15);
\draw[ thick, fill=myyellow9, rounded corners=2pt] (-0.5,-1+\jj) circle (.15);
\draw[ thick, fill=myyellow9, rounded corners=2pt] (0.5,\jj) circle (.15);
\draw[ thick, fill=myyellow9, rounded corners=2pt] (-0.5,\jj) circle (.15);
}
\foreach \jj in {0,-2,-4}{
\foreach \i in {1}{
\draw[ thick] (-2*\i+1.4,-.95-\jj) -- (-2*\i+1.55,-.95-\jj) -- (-2*\i+1.55,-1.1-\jj);}
\foreach \i in {1}{
\draw[ thick] (-2*\i+2.4,0.05-\jj) -- (-2*\i+2.55,0.05-\jj) -- (-2*\i+2.55,-0.1-\jj);}
\foreach \i in {1}{
\draw[ thick]  (-2*\i+1.45,-0.1-\jj) -- (-2*\i+1.45,0.05-\jj) -- (-2*\i+1.6,0.05-\jj);}
\foreach \i in {0}{
\draw[ thick]  (-2*\i+.45,-1.1-\jj) -- (-2*\i+.45,-0.95-\jj) -- (-2*\i+.6,-0.95-\jj);}
}
\end{tikzpicture}\Bigr]\,.
\ee
\ch{The first transfer matrix depends on its position $x$, as the region $A$ may be inhomogeneous.} We stress that the noise in \eqref{eq:SFFdiag} is periodic in time, i.e. $u_{x,\tau}=u_{x,\tau+1}$, as required for Floquet evolution.

Taking ${L\to \infty}$ again simplifies the problem\ch{, because only the leading eigenvalues and the corresponding eigenvectors of $\mathbb{\widetilde{T}}$ contribute}. Indeed Ref.~\cite{bertini2021random} proved that $\mathbb{\widetilde{T}}$ has $t$ eigenvalues $1$ while all other eigenvalues of  $\mathbb{\widetilde{T}}$ have magnitude smaller than $1$. The eigenvectors corresponding to the eigenvalue one are $\{\ket{\Pi_{2t}^{2\tau}}\}_{\tau=1,\ldots,t}$, where $\Pi_{2t}$ denotes the operator performing a one-site translation in the time direction. For example, $\ket{\Pi_{2t}^{2}}$ is graphically depicted as
\be
\ket{\Pi_{2t}^{2}}=
\begin{tikzpicture}[baseline=(current  bounding  box.center), scale=.7]
\foreach \j in {0,...,3}{
\draw[ thick]   (0,0+\j) -- (.3,0+\j) .. controls (1,-1.07+\j+.5) and (1,-1.07+\j-.5).. (0.5,-2-.15+\j) -- (0.2,-2-.15+\j);
}
\draw[ thick]    (1,-1.07+.5+4) .. controls (1,-1.07+.5+4) and (1,-1.07-.5+4).. (0.5,-.15+2) -- (0.2,-.15+2);
\draw[ thick]     (1-.2,-1.07-.5+5) ..controls  (1-.2,-1.07-.5+5) and  (1-.2,-1.07-.9+5)  .. (0.5,-.15+3) -- (0.2,-.15+3);
\draw[ thick]    (0,-1)--(0.3,-1) .. controls (1,-1.07+.5-1) and (1,-1.07-.5-1).. (1,-1.07-.5-1);
\draw[ thick]    (0,-2)--(0.3,-2) .. controls (1-0.2,-1.07-1.0) and (1-.2,-1.07-1.3).. (1-0.2,-1.07-.5-1);
\end{tikzpicture}\, ,
\ee
\ch{where the lines are connecting the two time-sheets at  positions shifted by two in the time direction. }

This means 
\be 
K(t)\equiv \lim_{L\to \infty} K(t,L) = \sum_{\tau=0}^{t-1} \bra{\Pi_{2t}^{2\tau}}  \prod_{j=1}^{L_A}\mathbb{T}_j  \ket{ \Pi_{2t}^{2\tau}}\,,
\ee
\ch{where we used the fact that } since the left and right eigenvectors are shifting by the same amount, and $\tr\mathbb{U}^t$ is invariant under translations in time, all terms in the sum are equal and, in particular, they are all equal to $\tr \mathbb B_A^t$. 
Therefore, we arrive a the following simple expression for $K(t)$ in terms of the map $\mathbb B_A$
\be
K(t)=t \, \tr \mathbb B_A^t\, .
\label{eq:SFFBA}
\ee
An analogous relation was first derived in Sec. III F of \cite{garratt2021local}, where the authors used it to express the leading eigenvalue of their space-transfer matrix. 

%-----------------------------------------
\subsection{The map $\mathbb B_A$}
\label{sec:map}
%-----------------------------------------

Having shown that $\mathbb B_A$ determines both the dynamics of local observables in the subsystem $A$ and the spectral statistics of \eqref{eq:diagramFloquet} let us now proceed to discuss the features of $\mathbb B_A$ that dominate the large-time behaviour. To this aim we make two observations. Firstly, we note that the map $\mathbb{B}_A$ is non-expanding for any unitary $\{U_{x,\tau}\}$, namely 
\begin{lemma}
For all normalised $\ket{\Omega}$, we have $\left| \braket{\Omega|\mathbb{B}_A| \Omega}\right|\leq 1$.
\end{lemma}
\begin{proof}
\begin{align}
& \left| \braket{\Omega|\mathbb{B}_A| \Omega}\right| \notag\\
&= |\!\braket{\mcirc \Omega \!\mcirc\! | (\1 \otimes \prod_{x}  W_{x,1} ) ( \prod_{x}  W_{x,1/2} \otimes \1 ) \!|\!\mcirc\! \Omega \mcirc}\!| \notag\\
&\leq
\|\! \1 \otimes \prod_{x}  W_{x,1}\!\| \|\! \prod_{x} W_{x,1/2} \otimes \1 \!\| \leq 1.
\end{align}
\end{proof}
\noindent 
\ch{Intuitively, this condition means that upon the application of $\mathbb{B}_A$ the norm of a vector cannot grow.} Secondly, we note that, using the unitarity of $U_{x,\tau}$, one can directly see that
\begin{lemma} The \ch{vectorised normalised identity operator} 
 \be
\ket{\mcirc}_{A}=\ket{\overbrace{\mcirc \dots \mcirc}^{2L_A-1}},
\label{eq:inftemp}
\ee
is an eigenvector of $\mathbb{B}_A$ corresponding to eigenvalue $1$, \ch{i.e. $\mathbb{B}_A$ is unital}. 
\end{lemma}
\noindent 
\ch{This follows directly from the unitarity of the local gates which induces the unitality of the double gates $W_{x,\tau} \ket{\mcirc \mcirc} = \ket{\mcirc \mcirc}$.
Upon unfolding, the vector $\ket{\mcirc}_{A}$ is proportional to an infinite temperature density matrix ${\rho_{\infty,A}= \1_A/d^{2L_A-1}}$ (where $\1_A$ is the identity operator in region $A$).} This means that we have either one of the following two cases: 
\begin{itemize}
\item[(i)] $\ket{\mcirc}_{A}$ is the only  eigenvector of $\mathbb{B}_A$ of unimodular\ch{\footnote{Of \ch{magnitude} equal to one.}} eigenvalue; 
\item[(ii)] There is at least one additional  eigenvector of $\mathbb{B}_A$ of unimodular eigenvalue\ch{\footnote{It might be complex and different from $1$.}};
\end{itemize}
In the case (i) the subsystem $A$ eventually thermalises to the infinite-temperature state. In particular, considering the Hilbert-Schmidt norm of the difference between the reduced density matrix and the thermal state $\rho_{\infty,A}$ we have 
\be
\|\rho_A(t)-\rho_{\infty,A}\|^2= \braket{\rho_A(0)|(\mathbb{B}_A^\dagger)^t\mathbb{B}_A^t|\rho_A(0)}-\frac{1}{D}\, ,
\ee
where $D$ is the dimension \ch{of the Hilbert space of the region $A$}. %To get this expression we expanded the product on the l.h.s.\ and used that the last three terms sum up to $-1/D$. 
The trivial eigenvalue $1$ of $\mathbb{B}_A$ gives the result $1/D$, therefore, at long times we have \ch{
\be
||\rho_A(t)-\rho_{\infty,A}||^2 \propto  \sum_j |c_j|^2 \Lambda_j^{2t} + \dots \, ,
\label{eq:normbound}
\ee
where $\Lambda_j$,  with ${|\Lambda_j|=\Lambda<1}$, are, possibly several, second largest (in magnitude) eigenvalues of $\mathbb B_A$, $c_j$ are the overlaps of the corresponding eigenvectors with $\ket{\rho_A(0)}$, and we denoted by dots the sub-leading terms.} Note that \ch{this asymptotic result} only applies for times larger than the size of $\mathbb B_A$'s largest Jordan's block\footnote{\ch{$\mathbb B_A$ is not necessarily diagonalizable.}} ($t \gg L_A$).

Eq.~\eqref{eq:normbound} implies that $\tau_{\rm thm}$, the thermalisation time-scale for all local observables, is bounded by 
\be
\tau_{\rm thm} \leq \tau_\Lambda \equiv  -\frac{1}{\log \Lambda} \, .
\label{eq:taulambda}
\ee
Concomitantly, the spectral form factor of \eqref{eq:diagramFloquet2}
\ch{
\be
K(t)=t \,  \tr \mathbb{B}_A^t = t\Bigl[1+\sum_j \Lambda_j^t+\dots\Bigr]
\ee
}approaches the result found in the circular unitary ensemble \ch{(CUE)} of random matrices, $K(t)=t$, \ch{indicating chaotic behaviour}. In particular, the timescale of the approach (called Thouless time in this context~\cite{kos2018many, chan2018spectral, kos2021chaos}) is given by $ \tau_\Lambda$. 

In the case (ii) the subsystem $A$ does not thermalise (and if the additional unimodular \ch{eigenvalue is $\neq 1$} the subsystem does not even equilibrate) and the spectral form factor of \eqref{eq:diagramFloquet2} does not approach the random-matrix-theory prediction \ch{$K(t)=t$, so the model is not chaotic.
Instead, the spectral form factor reads as
\be
K(t)= t \Bigr[1+\sum_j e^{i\phi_j t} +\dots\Bigl],
\ee
where $\Lambda_j=e^{i\phi_j}$ are the additional unimodular eigenvalues. In particular, in Sec.~\ref{sec:StrongL} we will see that if a circuit fulfils (ii) and the gates in $A$ are the same, $U_{x,t}=U$, there are exponentially many (in $L_A$) unimodular eigenvalues. If these eigenvalues are equal to one, as it is typically the case at least for $d=2$, this yields 
\be
 K(t) \sim t g^{L_A},
\ee
where $g$ is some constant.}

The case (i) is \ch{believed to be} the generic one and can be argued to occur for almost all choices of $\{U_{x,\tau}\}$ in  \eqref{eq:diagramFloquet2}. \ch{Indeed, on physical grounds, it is natural to expect a finite system coupled to a thermal bath to eventually thermalise. On a more technical level one could also argue that having additional unimodular eigenvalues of $\mathbb{B}_A$ requires additional constraints on $\{U_{x,\tau}\}$ and is hence non-generic.} A quantitative characterisation of $\Lambda$, however, has been achieved only for dual unitary circuits~\cite{piroli2020exact} (where it is exactly 0, Sec.~\ref{sec:pertDU}) and for the quantum cellular automaton Rule 54~\cite{klobas2021exact, klobas2021exactII}. In Sec.~\ref{sec:pertDU} we provide a perturbative characterisation of $\Lambda$ for $\{U_{x,\tau}\}$ close to (but not exactly at) the dual-unitary point. 

Moreover, in Sec.~\ref{sec:StrongL} we characterise the class of gates for which, instead, (ii) occurs. \ch{ As these systems do not thermalise even when connected to a thermal bath, we call them \emph{strongly non-ergodic}.} In Sec.~\ref{sec:NecSL} we prove a set of necessary conditions on $\{U_{x,\tau}\}$ in order to have \ch{strong non-ergodicity} for generic $d$, while in Sec.~\ref{sec:SLd2} we present \ch{the results of} a full classification of all \ch{strongly non-ergodic} gates for $d=2$ in the presence of translational invariance\ch{, which is carried out explicitly in App. \ref{app:d2homogenous}}.

To conclude this brief survey we point out that, interestingly, a quantum map very similar to $\mathbb B_A$ (with $m_{1/2,l}$ in \eqref{eq:Bmap} replaced by the identity) has been shown in Ref.~\cite{ippoliti2021fractal} to describe the entanglement growth under non-unitary local quantum circuits with unitary duals (i.e.\ that are unitary in the space direction but not in the time one~\cite{ippoliti2021postselection}). There, the authors focused on the late time regimes of the non-unitary evolution, which correspond to ${t< L_A}$. Note that the dynamics in this regime are determined by the full structure of $\mathbb B_A$ (Jordan blocks and full spectrum). On the other hand, here we are interested in the complementary regime $t\gg L_A$, which is completely specified by the largest eigenvalues of $\mathbb B_A$. This means that, even though this regime is perhaps less rich from the physical point of view, it is amenable to a more rigorous analysis.

%-------------------------------------------------------------------------------------------------------------
\section{$\mathbb B_A$ in perturbed dual-unitary circuits}
\label{sec:pertDU}
%-------------------------------------------------------------------------------------------------------------

Let us characterise the leading part of the spectrum of $\mathbb B_A$ for circuits made of perturbed dual-unitary gates. We will show that the magnitude of the second leading eigenvalue is proportional to the perturbation strength and that, in certain perturbed dual unitary circuits, one can obtain the leading part of the spectrum of $\mathbb B_A$ using approaches similar to those introduced in \cite{kos2021correlations, kos2021chaos}.

To begin with, let us note that if all the gates composing $\mathbb B_A$ in Eq.~\eqref{eq:Bmap} are dual unitary we have  
\be
\tr[\mathbb B_A^t] = 1,\qquad\forall t\,. 
\ee
This means that the spectrum of $\mathbb B_A$ contains just two points, ${\rm Sp}(\mathbb B_A)=\{0,1\}$, and, moreover, the eigenvalue $1$ has geometric and algebraic multiplicity equal to one; the corresponding eigenvector is given by \eqref{eq:inftemp}. Note that these facts directly imply that $\tau_{\Lambda}$ in Eq.~\eqref{eq:taulambda} is exactly zero for dual-unitary circuits. In addition we also have 
\be
\mathbb B_A^{t} = \ket{\mcirc\ldots\mcirc}\bra{\mcirc\ldots\mcirc},\qquad t\geq L_A\,,
\ee
which means that the Jordan blocks relative to the eigenvalue $0$ have size bounded by $L_A$. Therefore the subsystem $A$ thermalises exactly for times $t\geq L_A$.

Let us now consider perturbed dual-unitary gates of the form~\cite{kos2021correlations}
\be
(U_\eta)_{x,\tau}=(U_{\rm du})_{x,\tau} e^{i \eta P_{x,\tau}}\, ,
\label{eq:pertU}
\ee
where $U_{\rm du}$ is dual-unitary gate, $\eta\ll 1$ the perturbation strength and $P_{x,\tau}$ generic Hermitian two-qudit operators. The gates \eqref{eq:pertU} produce double gates of the form 
\be
(W_\eta)_{x,\tau}= (W_{\rm du})_{x,\tau} + (\delta W_\eta)_{x,\tau}.
\label{eq:splitW}
\ee
Here we introduced ``defects'' defined by
\be
\!\!\!(\delta W_\eta)_{x,\tau} = (W_{\rm du})_{x,\tau} (e^{i \eta P_{x,\tau}}\!\otimes e^{-i \eta P^T_{x,\tau}}-\1)= \begin{tikzpicture}[baseline=(current  bounding  box.center), scale=0.75]
\def\eps{0.5};
\Wgategrey{0}{0};
\Text[x=-0-0.25,y=-0.75]{}
\end{tikzpicture},
\ee
where $(\cdot)^T$ denotes transposition. Note that defects are such that 
\be
(\delta W_\eta)_{x,\tau} \ket{\mcirc\mcirc}=0 = \bra{\mcirc\mcirc}(\delta W_\eta)_{x,\tau}\,.
\label{eq:conditiondef} 
\ee

As we shall see, the perturbation \eqref{eq:pertU} leads to the appearance of additional non-zero eigenvalues of $\mathbb B_A$ which are proportional to fractional powers of $\eta$. In the following we develop perturbative predictions for the latter and compare them with exact numerics. \ch{For simplicity, even though this is not necessary, in the rest of this section we use the same perturbed gate everywhere in the region $A$, i.e. $W_{x,\tau}=W$.}
In particular, in Sec.~\ref{sec:leadingpt} we compute the first non-trivial correction in $\eta$ to the spectrum of $\mathbb B_A$. Instead, in Sec.~\ref{sec:skeleton} we present an uncontrolled approximation based on a particular truncation of the matrix $\mathbb B_A$, which gives quantitatively accurate results even at intermediate times and at moderate perturbation strengths. 
Finally, in Sec.~\ref{sec:results} we discuss the predictions that these approximation schemes yield concerning spectral and thermalisation properties of the system.

%-----------------------------------------
\subsection{Limit of vanishing perturbation strength $\eta$ at fixed $L_A$}
%-----------------------------------------
\label{sec:leadingpt}

Let us begin by evaluating the first non-trivial correction in $\eta$ to $\tr \mathbb{B}_A^t$. 
\ch{This correction is reflected in eigenvalues of $\mathbb{B}_A$ of size proportional to $\eta^{{2}/{(2L_A-1)}}$.  Interestingly, these eigenvalues give non-zero contribution to $\tr \mathbb{B}_A^t$ only at times, which are multiples of $2L_A-1$. To access qualitative predictions at other times, one need to access higher order corrections to the  eigenvalues of $\mathbb{B}_A$, which we do in the subsequent subsection \ref{sec:skeleton}.
To evaluate the first non-trivial correction we consider the diagrammatic representation  }
\be
{\rm tr}[\mathbb B_A^t] =
\begin{tikzpicture}[baseline=(current  bounding  box.center), scale=0.55]
\foreach \i in {5.5, 6.5, 7.5}{
\draw[thick, dotted] (2*\i+2-12.5+0.255,-.5) -- (2*\i+2-12.5+0.255,9.65);}
\foreach \i in {5.5, 6.5}{
\draw[thick, dotted] (2*\i+2-11.5-0.255,-0.5) -- (2*\i+2-11.5-0.255,9.65);}
\foreach \i in {5.5, 6.5}{
\draw[very thick] (2*\i+2-11.5,9.5) arc (-45:175:0.15);
\draw[very thick] (2*\i+2-11.5,-.5) arc (315:180:0.15);
}
\foreach \i in {5.5, 6.5, 7.5}{
\draw[very thick] (2*\i+2-0.5-12,-.5) arc (-135:0:0.15);
\draw[very thick] (2*\i+4-2.5-12,9.5) arc (225:0:0.15);
}
\foreach \i in {0,...,4}{
\Wgategreen{0}{0+2*\i}\Wgategreen{2}{0+2*\i}\Wgategreen{4}{0+2*\i}
\Wgategreen{1}{1+2*\i}\Wgategreen{3}{1+2*\i}\Wgategreen{5}{1+2*\i}
}
\foreach \i in {0,...,4}{
\draw[thick, fill=white] (-0.5,-0.5+2*\i) circle (0.1cm); 
\draw[thick, fill=white] (-0.5,0.5+2*\i) circle (0.1cm); 
\draw[thick, fill=white] (5.5,1.5+2*\i) circle (0.1cm); 
\draw[thick, fill=white] (5.5,0.5+2*\i) circle (0.1cm); 
}
\label{eq:trBdiag}
\end{tikzpicture}\, ,
\ee 
%\ew
and write each \ch{green gate (they are all the same)} as the sum of \ch{orange} dual unitary part   and \ch{grey}  defect \ch{(cf. \eqref{eq:splitW})}
\be
 \begin{tikzpicture}[baseline=(current  bounding  box.center), scale=0.75]
\def\eps{0.5};
\Wgategreen{0}{0};
\Text[x=-0-0.25,y=-0.75]{}
\end{tikzpicture} = \begin{tikzpicture}[baseline=(current  bounding  box.center), scale=0.75]
\def\eps{0.5};
\Wgateorange{0}{0};
\Text[x=-0-0.25,y=-0.75]{}
\end{tikzpicture}
+
 \begin{tikzpicture}[baseline=(current  bounding  box.center), scale=0.75]
\def\eps{0.5};
\Wgategrey{0}{0};
\Text[x=-0-0.25,y=-0.75]{}
\end{tikzpicture}
\ee
generating the following sum over distributions of defects in the network 
\be
{\rm tr}[\mathbb B_A^t] = \sum_{\{\mathcal C\}} W_{\mathcal C}\,.
\label{eq:sumdefects}
\ee
Here ${\mathcal C}$ is a set of positions on the tensor network and $W_{\mathcal C}$ is a tensor network like \eqref{eq:trBdiag} with all dual unitary gates except for defects at the positions specified by ${\mathcal C}$. 

As discussed before, the contribution with no defect gives one and our objective is to determine the terms with the lowest number of defects that give non-vanishing contribution. 

For $t= 0\, {\rm mod}(2L_A-1)$ \ch{the only non-zero} such configurations are those with $2t/(2L_A-1)$ defects placed at the two boundaries of $A$ in an alternating fashion and are separated in time by $L_A-1/2$ steps, for instance
\be
W_{\{(3,2),(1/2,9/2)\}}= \begin{tikzpicture}[baseline=(current  bounding  box.center), scale=0.55]
\foreach \i in {5.5, 6.5, 7.5}{
\draw[thick, dotted] (2*\i+2-12.5+0.255,-.5) -- (2*\i+2-12.5+0.255,9.65);}
\foreach \i in {5.5, 6.5}{
\draw[thick, dotted] (2*\i+2-11.5-0.255,-0.5) -- (2*\i+2-11.5-0.255,9.65);}
\foreach \i in {5.5, 6.5}{
\draw[very thick] (2*\i+2-11.5,9.5) arc (-45:175:0.15);
\draw[very thick] (2*\i+2-11.5,-.5) arc (315:180:0.15);
}
\foreach \i in {5.5, 6.5, 7.5}{
\draw[very thick] (2*\i+2-0.5-12,-.5) arc (-135:0:0.15);
\draw[very thick] (2*\i+4-2.5-12,9.5) arc (225:0:0.15);
}
\foreach \i in {0,...,4}{
\Wgateorange{0}{0+2*\i}\Wgateorange{2}{0+2*\i}\Wgateorange{4}{0+2*\i}
\Wgateorange{1}{1+2*\i}\Wgateorange{3}{1+2*\i}\Wgateorange{5}{1+2*\i}
}
\Wgategrey{5}{3}
\Wgategrey{0}{8}
\foreach \i in {0,...,4}{
\draw[thick, fill=white] (-0.5,-0.5+2*\i) circle (0.1cm); 
\draw[thick, fill=white] (-0.5,0.5+2*\i) circle (0.1cm); 
\draw[thick, fill=white] (5.5,1.5+2*\i) circle (0.1cm); 
\draw[thick, fill=white] (5.5,0.5+2*\i) circle (0.1cm); 
}
\end{tikzpicture}.
\ee
These diagrams simplify due to dual-unitarity, leaving us with one dimensional zig-zag like quantum channel that in Ref.~\cite{kos2021correlations} has been dubbed \emph{skeleton diagram}. For example the one above simplifies to 
\be
W_{\{(3,2),(1/2,9/2)\}}= \begin{tikzpicture}[baseline=(current  bounding  box.center), scale=0.55]
\foreach \i in {5.5}{
%\draw[very thick, dotted] (2*\i+2-12.5+0.255,-1.75-0.1) -- (2*\i+2-12.5+0.255,4.25-0.1);
\draw[thick, dotted] (2*\i+2-11.5-0.255,-0.5) -- (2*\i+2-11.5-0.255,9.65);}
\foreach \i in {5.5}{
\draw[very thick] (2*\i+2-11.5,9.5) arc (-45:175:0.15);
\draw[very thick] (2*\i+2-11.5,-.5) arc (315:180:0.15);
}
\foreach \i in {1,...,3}{
\Wgateorange{\i+1}{\i-1}
\draw[thick, fill=white] (0.5+\i,-0.5+\i) circle (0.1cm); 
\draw[thick, fill=white] (1.5+\i,-1.5+\i) circle (0.1cm); 
}
\foreach \i in {3,...,6}{
\Wgateorange{-\i+7}{\i+1}
\draw[thick, fill=white] (6.5-\i,\i+.5) circle (0.1cm); 
\draw[thick, fill=white] (7.5-\i,\i+1.5) circle (0.1cm); 
}
\Wgategrey{5}{3}
\Wgategrey{0}{8}
\Wgateorange{1}{9}
\draw[thick, fill=white] (1.5,8.5) circle (0.1cm); 
\draw[thick, fill=white] (.5,9.5) circle (0.1cm); 
\draw[thick, fill=white] (-.5,8.5) circle (0.1cm); 
\draw[thick, fill=white] (-.5,7.5) circle (0.1cm); 
\draw[thick, fill=white] (5.5,3.5) circle (0.1cm); 
\draw[thick, fill=white] (5.5,2.5) circle (0.1cm); 
\end{tikzpicture}\,,
\label{eq:skeleton}
\ee
\ch{where empty circles still depict the vectorised identity operators \eqref{eq:Id}.}
Summing all these terms and noting that for $t\neq 0\, {\rm mod}(2L_A-1)$ one needs at least $2 \left \lfloor{t/(2L_A-1)}\right \rfloor +2$ defects to have a non-zero contribution we obtain 
\be
\tr \mathbb{B}_A^t = 1+ S_0 \left(\tfrac{t}{2L_A-1},L_A\right) (\delta_{t \, {\rm mod}\, 2L_A-1,0} + \mathcal{O}(\eta^2))\, ,
\label{eq:traceB}
\ee
where $S_0(k,L_A) \propto \eta^{2k}$ is evaluated contracting the skeleton diagrams. For example, in the translational invariant case we have 
\be
S_0(k,L_A) = (2L_A-1) {\rm tr}\big[C^{k}\big], 
\ee
where $C$ denotes the quantum channel which begins and ends at the same point, i.e., 
\be
C = d_{l} (m_-)^{2L_A-2} d_{r} (m_+)^{2L_A-2},
\ee
and we introduced the quantum channels
 \begin{align}
m_+ &=\begin{tikzpicture}[baseline=(current  bounding  box.center), scale=0.7]
\Wgateorange{0}{0}\draw[thick, fill=white] (0.5,-0.5) circle (0.1cm);
\draw[thick, fill=white] (-0.5,0.5) circle (0.1cm); 
\end{tikzpicture} \, ,
& 
m_- &=\begin{tikzpicture}[baseline=(current  bounding  box.center), scale=0.7]
\Wgateorange{0}{0}
\draw[thick, fill=white] (-0.5,-0.5) circle (0.1cm);
\draw[thick, fill=white] (0.5,0.5) circle (0.1cm); 
\end{tikzpicture}\,, 
\label{eq:mpm}
\\
d_l &=\begin{tikzpicture}[baseline=(current  bounding  box.center), scale=0.7]
\Wgategrey{0}{0}\draw[thick, fill=white] (-0.5,-0.5) circle (0.1cm);
\draw[thick, fill=white] (-0.5,0.5) circle (0.1cm); 
\end{tikzpicture} \, ,
& 
d_r &=\begin{tikzpicture}[baseline=(current  bounding  box.center), scale=0.7]
\Wgategrey{0}{0}
\draw[thick, fill=white] (0.5,-0.5) circle (0.1cm);
\draw[thick, fill=white] (0.5,0.5) circle (0.1cm); 
\end{tikzpicture}\,.
\end{align}
Note that $d_l, d_r$ are ${\mathcal O}(\eta)$ for $\eta\to0$. 

The matrix $C$ has ${d^2-1}$ nontrivial eigenvalues
$\left\{ \lambda_{C,i} \right\}_{i=0}^{d^2-2}$ and a trivial eigenvalue $0$ corresponding to the eigenvector $\mcirc$ \ch{(cf. \eqref{eq:Id})}. Rewriting the leading order of \eqref{eq:traceB} at times $t=k (2L_A-1)$ we then have 
\begin{align}
&{\rm tr}[\mathbb B_A^{k(2L_A-1)}] -1=\sum_l \Lambda_l^{k(2L_A-1)}\notag\\
&= (2L_A-1) \sum_{i=0}^{d^2-2} \lambda_{C,i}^k\qquad \forall k>0\,.
\end{align}
\ch{Here, slightly differently from before, we denoted by $\Lambda_l$ all eigenvalues of $\mathbb B_A$ which appear at the first nontrivial order of $\eta$ except for the eigenvalue $1$.}
Since this holds for all $k>0$ and \eqref{eq:traceB} is of higher order in $\eta$ at times $t \, {\rm mod} \, (2L_A-1) \neq 0$, we obtain (apart from the leading eigenvalue one) $(d^2-1) (2L_A-1)$ sub-leading eigenvalues of $\mathbb B_A$
\be
\Lambda_{(2L_A-1)j+k} =  e^{\frac{i 2\pi k}{2L_A-1}} (\lambda_{C,j})^{1/(2L_A-1)} 
\label{eq:firstLambda}
\ee
where $k=0, \dots 2L_A-2$, $j=0, \dots d^2-2$. We illustrate the agreement with numerical results in Fig.~\ref{fig:Spectrum}. 

\begin{figure}[t]
    \centering
    \includegraphics[scale=0.65]{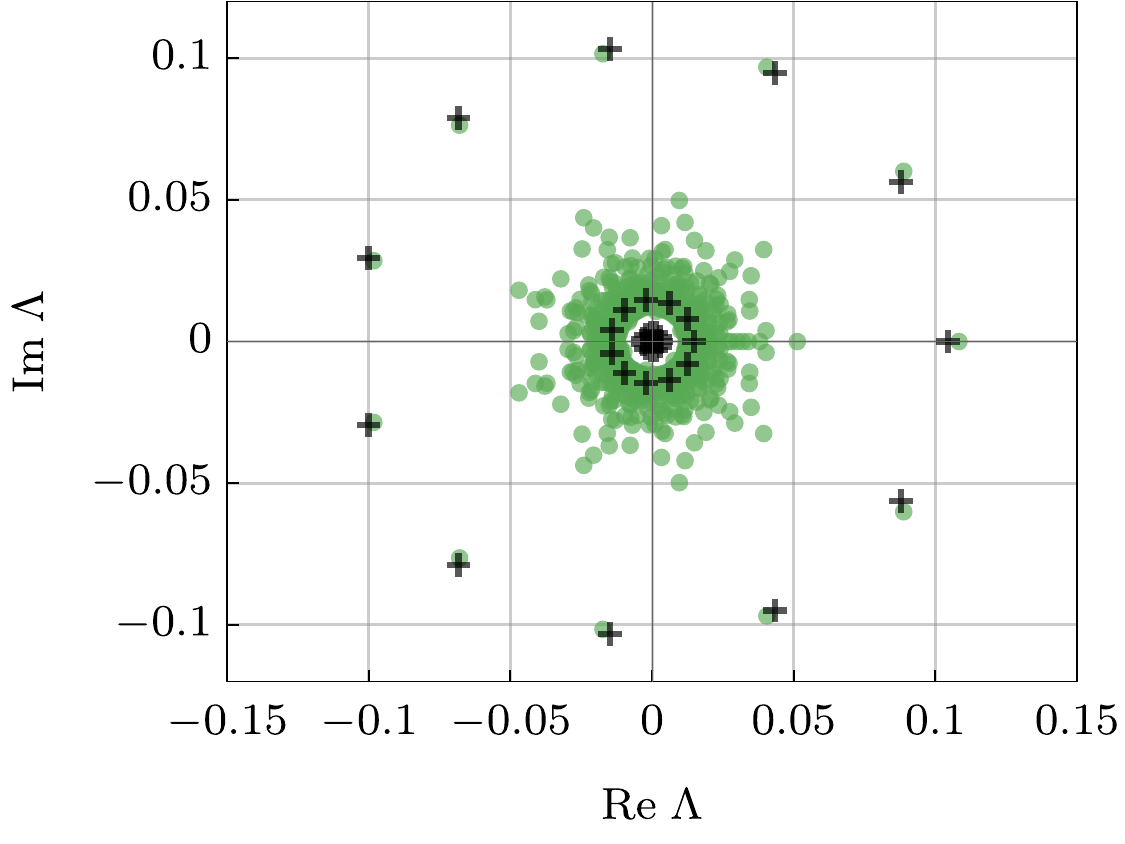}
    \caption{Comparison between the spectrum of the result based on the first order given by Eq.~\eqref{eq:firstLambda} depicted in black and exact numerical results for \ch{homogeneously} perturbed dual-unitary circuits of largest $200$ eigenvalues depicted in green.  Notice the $2L_A-1$ fold symmetry of the leading part of the spectrum. $2L_A=12$, $\eta=0.0001$ and the gate parameters are in the first entry of Tab.~\ref{tab:PrxGates}.
    }
    \label{fig:Spectrum}
\end{figure}

\ch{Note that, for large enough $L_A$ the leading eigenvalue of $C$ satisfy} 
\be
\lambda_{C,{\rm lead}} \propto  \eta^2\lambda_+^{2L_A-2} \lambda_-^{2L_A-2}, 
\label{eq:leadLambdaC}
\ee
where $\lambda_\pm$ are the first sub-leading eigenvalues of $m_\pm$ \ch{from Eq. \eqref{eq:mpm}}. This gives 
\be
\Lambda_{k,{\rm lead}} \propto e^{\frac{i 2\pi k}{2L_A-1}} \lambda_+ \lambda_- \left[\frac{\eta^{2}}{\lambda_+\lambda_-}\right]^{{1}/{(2L_A-1)}},
\label{eq:leadLambdas}
\ee
with $k=0,1, \dots 2L_A-2$, and, in turn
\be
\tau_{\Lambda}\propto -\frac{2L_A-1}{\ln \eta^{2} - (2L_A-2)\log(\lambda_+\lambda_-)}\simeq -\frac{2L_A-1}{\ln \eta^{2}}\, . 
\label{eq:tauPertDU}
\ee

The result \eqref{eq:leadLambdas} can also be understood in terms of standard perturbation theory for a matrix, $\mathbb B_A$, with non-trivial Jordan structure. In particular, one can verify that at the first non-trivial order the perturbation couples a Jordan block of size $L$ and one of size $L-1$ as 
\be 
\begin{aligned}
\braket{L_A,L_A|\mathbb B_A|L_A-1,1} &\sim \eta\, , \\
\braket{L_A-1,L_A-1|\mathbb B_A|L_A,1} &\sim \eta \, .
\end{aligned}
\ee 
This produces $2L_A-1$ eigenvalues of equal magnitude $\propto \eta^{2/(2L_A-1)}$ and phases $e^{i 2\pi k/(2L_A-1)}$ for $k=0,\ldots,2L_A-2$.

%-----------------------------------------
\subsection{\ch{All skeleton diagrams}}
\label{sec:skeleton}
%-----------------------------------------

\begin{figure*}[t]
    \centering
    \includegraphics[width=0.47\textwidth]{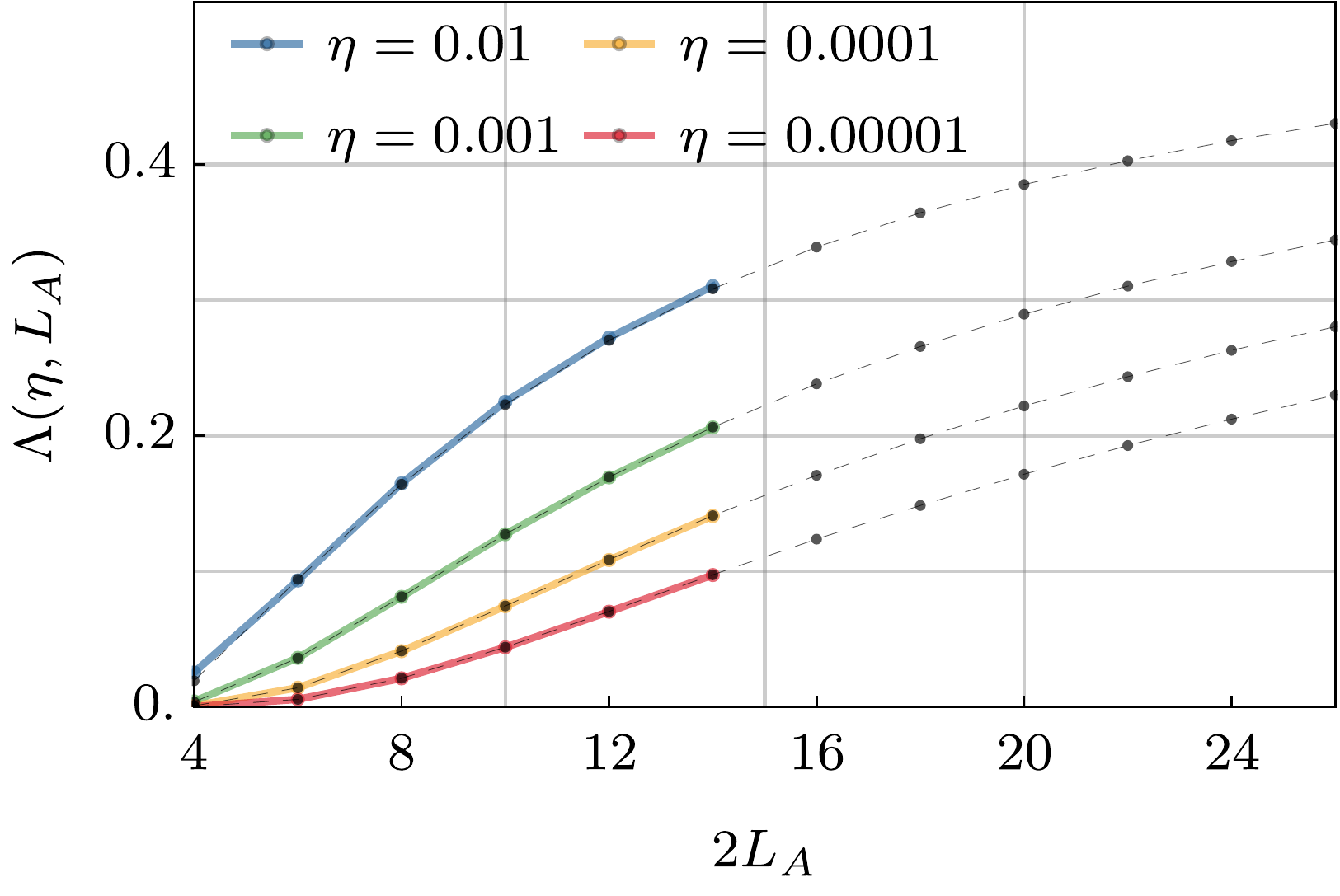}
    \includegraphics[width=0.47\textwidth]{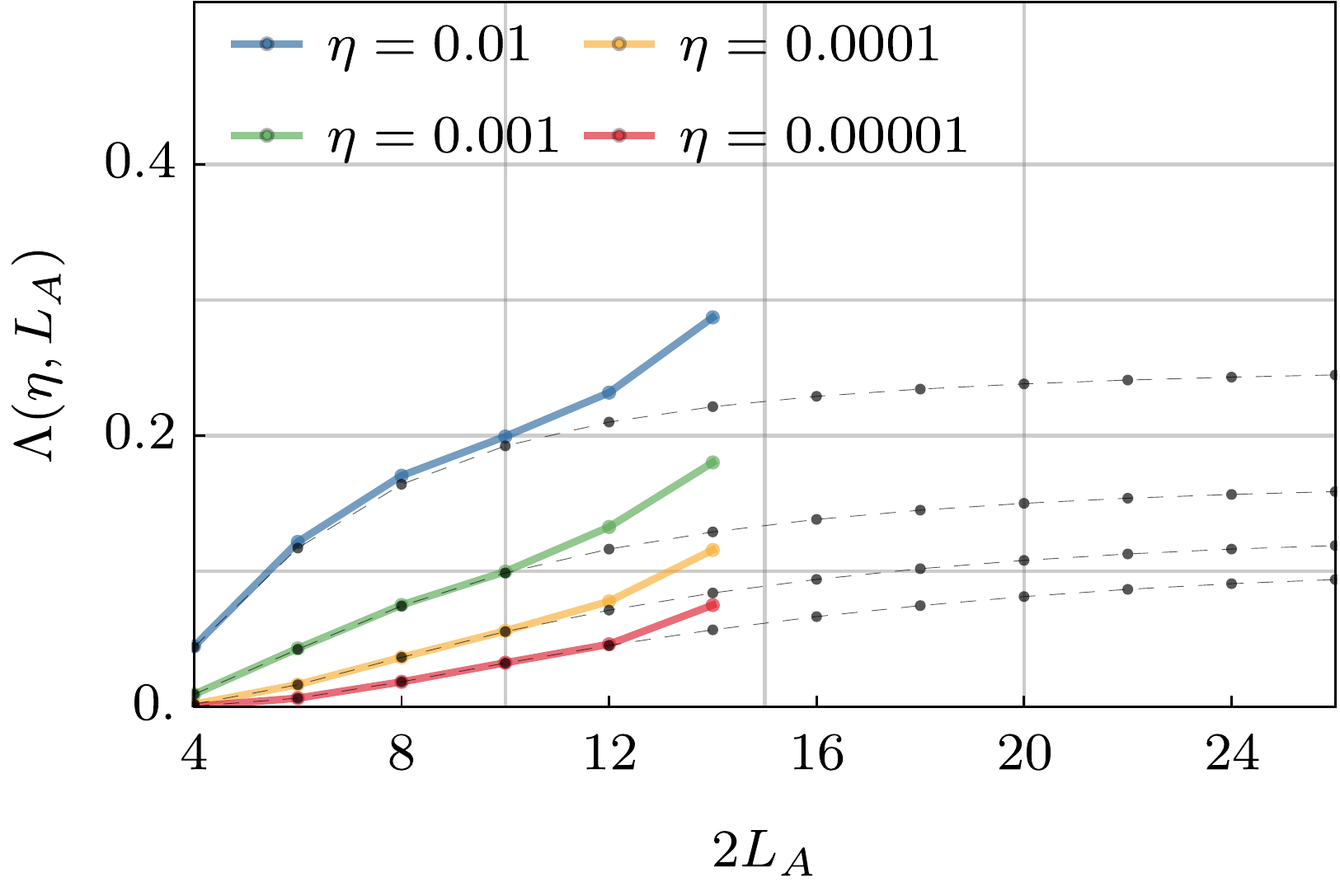}
    \caption{
    The first non-trivial (sub-leading) eigenvalue $\lambda(\eta,L_A)$ versus $L_A$ for different perturbation strength $\eta$ at $d=2$. In black, we show the predictions from \ch{Eq.~\eqref{eq:BM1} (all skeleton diagrams)}. The left (right) panel shows an example where the \ch{bare bones} condition is (not) fulfilled. 
    When the conditions are not fulfilled, the effect of the sub-leading terms become visible by increasing $L_A$, as we explain in the main text.
    The parameters of the gates are given in Tab.~\ref{tab:PrxGates}. }
    \label{fig:lambda}
\end{figure*}

For the special class of perturbed dual-unitary circuits characterised in Ref.~\cite{kos2021correlations},  we can go beyond the first order result presented in the previous subsection by following the methods introduced in \cite{kos2021correlations, kos2021chaos} for the calculation of correlation functions. 

The idea is to assume that in the sum \eqref{eq:sumdefects} over the possible distributions of defects, \ch{those that result in}  skeleton diagrams like \eqref{eq:skeleton} are the most important. That is, they are not only giving the correct first order result as shown in the previous subsection but, at large enough $L_A$ and $t$, they give a very good approximation of all subsequent orders. The intuitive reason is that configurations originating from non-skeleton diagrams are expressed in terms of quantum channels with maximal eigenvalues strictly smaller than $\lambda_{\pm}$ (cf.~\eqref{eq:leadLambdaC}). This means that they are suppressed exponentially in $L_A$ and $t$ with respect to the skeleton diagrams. We will refer to these cases as those where the \ch{``\emph{bare bones} condition''} holds.

\ch{The aforementioned  approximation is uncontrolled}, meaning that it does not fully reproduce any order of the perturbative expansion in $\eta$ higher than the the leading one, but, as we now show, works extremely well for the class of circuits identified in \cite{kos2021correlations} \ch{(see also the comparisons with exact numerics in Sec.~\ref{sec:results}). In fact, the fact that this approximation works so well can be justified by means of a simple energetic argument~\cite{kos2021correlations}.} We begin by noting that the approximation under discussion is equivalent to restricting $\mathbb B_A$ to the sector spanned by operators with \ch{single site support} (\ch{dubbed as $m=1$ sector in~\cite{kos2021chaos}}). 
\ch{This sector is spanned by the following basis
\be
\{ \ket{
\underbrace{\mcirc \dots \mcirc}_{\text{$2 x -1$}}
 o^\alpha \underbrace{\mcirc \dots \mcirc}_{\text{$2L_A-2 x $}}}\},
\label{eq:Hm1}
\ee
where ${x={1}/{2},1,\dots L_A-{1}/{2}}$, ${\alpha=1,\dots,d^2-1}$, and $o^\alpha$ are elements of the Hilbert-Schmidt orthonormal basis of the traceless part of ${\rm End}({\mathbb C}^d)$, e.g. the set of generalised Gell-Mann matrices.
Specifically, the aforementioned approximation} corresponds to the following replacement 
\be
\mathbb B_A \mapsto (\mathbb B_A)_{m=1}.
\label{eq:SKreplacement}
\ee
The restricted matrix $(\mathbb B_A)_{m=1}$ has dimensions $(d^2-1) L_A \times (d^2-1) L_A$ and is made out of three different \ch{blocks and can be arranged as follows}
\be
(\mathbb B_A)_{m=1}=\begin{pmatrix}
\begin{tikzpicture}[baseline=0pt, scale=\scale*1.4]
\draw[  thick] (0,0) rectangle (2,-3);\node at (1,-1.5) {$E_1$}; 
\draw[thick] (2,-2) rectangle (4,-6); \node at (3,-4) {$B$}; 
\draw[thick] (4,-4) rectangle (6,-8);\node at (5,-6) {$B$}; 
\node at (7,-8) {$\ddots$};
\draw[ thick] (8,-8) rectangle (10,-12);\node at (9,-10) {$B$};
\draw[ thick] (10,-10) rectangle (11,-12); \node at (10.5,-11) {$E_2$}; 
\end{tikzpicture}
\end{pmatrix}\,.
\label{eq:BM1}
\ee
Here the blocks $B$, $E_1$, and $E_2$ are written in terms of ${8(d^2-1)^2}$ different matrix elements as %($72$ for $d=2$)
\be
\begin{aligned}
B&=\begin{pmatrix}
\begin{tikzpicture}[baseline=0pt, scale=\scale] \draw[thick] (0,0) -- (0,1)--(1,2) ;\end{tikzpicture}_{\, r} &
\begin{tikzpicture}[baseline=0pt, scale=\scale]    \draw[thick] (0,0) -- (1,1)--(2,2) ;\end{tikzpicture}_{\, l}
\\
\begin{tikzpicture}[baseline=0pt, scale=\scale] \draw[thick] (2,0) --(2,2) ;\end{tikzpicture}_{\, r} &
\begin{tikzpicture}[baseline=0pt, scale=\scale]    \draw[thick] (0,0) -- (1,1)--(1,2) ;\end{tikzpicture}_{\, l}
\\
\begin{tikzpicture}[baseline=0pt, scale=\scale] \draw[thick] (1,0) -- (1,1)--(2,2) ;\end{tikzpicture}_{\, r} &
\begin{tikzpicture}[baseline=0pt, scale=\scale]    \draw[thick] (0,0) --(0,2) ;\end{tikzpicture}_{\, l}
\\
\begin{tikzpicture}[baseline=0pt, scale=\scale] \draw[thick] (2,0) -- (1,1)--(0,2) ;\end{tikzpicture}_{\, r} &
\begin{tikzpicture}[baseline=0pt, scale=\scale]    \draw[thick] (1,0) -- (1,1)--(0,2) ;\end{tikzpicture}_{\, l}
\end{pmatrix}\, ,
\\
E_1&=\begin{pmatrix}
\begin{tikzpicture}[baseline=0pt, scale=\scale] \draw[thick] (2,0) --(2,2) ;\end{tikzpicture}_{\, r} &
\begin{tikzpicture}[baseline=0pt, scale=\scale]    \draw[thick] (0,0) -- (1,1)--(1,2) ;\end{tikzpicture}_{\, l}
\\
\begin{tikzpicture}[baseline=0pt, scale=\scale] \draw[thick] (1,0) -- (1,1)--(2,2) ;\end{tikzpicture}_{\, r} &
\begin{tikzpicture}[baseline=0pt, scale=\scale]    \draw[thick] (0,0) --(0,2) ;\end{tikzpicture}_{\, l}
\\
\begin{tikzpicture}[baseline=0pt, scale=\scale] \draw[thick] (2,0) -- (1,1)--(0,2) ;\end{tikzpicture}_{\, r} &
\begin{tikzpicture}[baseline=0pt, scale=\scale]    \draw[thick] (1,0) -- (1,1)--(0,2) ;\end{tikzpicture}_{\, l}
\end{pmatrix}\, ,
\quad
E_2=\begin{pmatrix} %0_{3\times3} \\
\begin{tikzpicture}[baseline=0pt, scale=\scale] \draw[thick] (0,0) -- (0,1)--(1,2) ;\end{tikzpicture}_{\, r} \\
\begin{tikzpicture}[baseline=0pt, scale=\scale] \draw[thick] (2,0) --(2,2) ;\end{tikzpicture}_{\, r} 
\end{pmatrix}\, ,
\end{aligned}
\ee
where we denoted the $(d^2-1)\times (d^2-1)$ blocks by the quantum channels, which they represent. The labels $l$ and $r$ respectively denote cases where the initial site is to the left and right of the gate. The straight vertical lines represent the straight channels $m_{i,l}$ and $m_{i,r}$, whereas the tilted segments represent the quantum channels similar to $m_{\pm}$ but made out of general (green) gates. For instance
\be 
\Bigg( \, \begin{tikzpicture}[baseline=(current  bounding  box.center), scale=1.3*\scale] \draw[thick] (1,0) -- (1,1)--(2,2);
\end{tikzpicture}_{\, r} \, \Bigg)_{\alpha\beta}
%= \bra{ o^\alpha} m_+ m_{l} \ket{o^\beta}
=
\begin{tikzpicture}[baseline=(current  bounding  box.center), scale=0.55]
\Wgategreen{0}{0}\Wgategreen{1}{1}
\draw[thick, fill=white] (-0.5,-0.5) circle (0.1cm);
\draw[thick, fill=white] (-0.5,0.5) circle (0.1cm);
\draw[thick, fill=white] (0.5,1.5) circle (0.1cm);
\draw[thick, fill=white] (1.5,.5) circle (0.1cm);
\draw[thick, fill=black] (0.5,-0.5) circle (0.1cm);
\draw[thick, fill=black] (1.5,1.5) circle (0.1cm);
\node[text width=1cm] at (1.1,-0.99){$o^\beta$};
\node[text width=1cm] at (2.1,1.99){$o^\alpha$};
\end{tikzpicture}  \,,
\ee
where \ch{the black filled circles  labeled by $o^\alpha$, $o^\beta$ come from the basis give in \eqref{eq:Hm1} for $L_A=1$.
}
%the set}
%\be
%\{\ket{o^\alpha}\}_{\alpha=1,\ldots,d^2-1}.
%\ee
%\ch{This is} a Hilbert-Schmidt orthonormal basis of the traceless part of ${\rm End}({\mathbb C}^d)$, e.g. the set of generalised Gell-Mann matrices. 
The matrix $(\mathbb B_A)_{m=1}$ can be evaluated and diagonalised numerically with negligible cost.

In Fig.~\ref{fig:lambda} we show the agreement between the leading eigenvalues of $(\mathbb B_A)_{m=1}$ and those of the full matrix $\mathbb B_A$, which are obtained through exact numerical evaluation. 

  \begin{figure}
    \centering
    \includegraphics[scale=0.65]{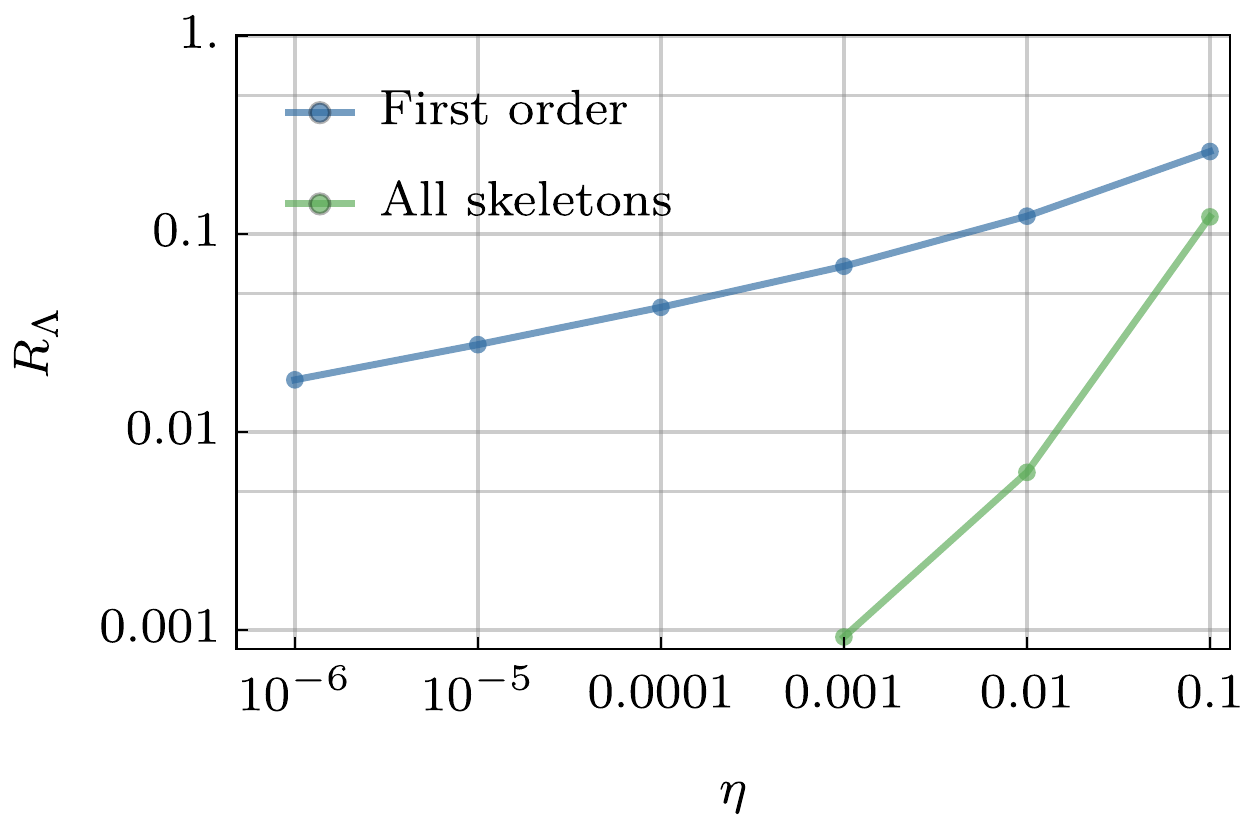}
    \caption{Relative error \eqref{eq:RelErr} of $\Lambda$ if we use only leading order in $\eta$ or all skeleton diagrams for the case $1$ in Tab.~\ref{tab:PrxGates}. We see that \eqref{eq:SKreplacement} vastly outperforms the first order approximation. $2L_A=14$. 
    %[second: rel err of trace at times mod $0$; $2L=10$]
    }
    \label{fig:higherOrders}
\end{figure}

A natural question is whether this uncontrolled approximation gives a more accurate prediction than the first order result. This is indeed the case as shown in Fig.~\ref{fig:higherOrders}, which depicts the relative error
\be
R_\Lambda=\left| \frac{|\Lambda_p|-|\Lambda_n|}{\Lambda_n} \right| \, .
\label{eq:RelErr}
\ee
\ch{
In the above equation $\Lambda_p$ is the predicted first sub-leading eigenvalue, using the replacement \eqref{eq:SKreplacement} or the first order  result \eqref{eq:firstLambda}}, and $\Lambda_n$ is the one computed numerically. 

As discussed before, the simple skeleton diagrams give a good approximation if the \ch{bare bones condition holds}. Technically speaking (see Ref.~\cite{kos2021correlations} for more details), this condition demands that the largest nontrivial eigenvalue of the ``correlation functions' transfer matrix'' $\lambda_{\rm true}$ coincides with the largest nontrivial eigenvalue of the single site ``diagonal transfer matrix'' $\lambda_{1}$. When this condition is not satisfied, the first sub-leading term is wrong by a factor exponentially large in $L_A$. For instance, at $t=k(2L_A-1)$ the first sub-leading correction to ${\rm tr}\, \mathbb{B}_A^t $ is of the form 
\be
\eta^{2t/(2L_A-1)} \lambda_1^{t-2t/(2L_A-1)-2} \eta^2 \left|\frac{\lambda_{\rm true}}{\lambda_1}\right|^{2L_A-3}.
\ee
%}
We illustrate this fact in Fig.~\ref{fig:DeltaLambda}. Therefore, if one fixes $\eta$ and increases $L_A$, one obtains wrong predictions. As discussed in~\cite{kos2021correlations} one could obtain useful approximations also in these cases by extending the above analysis to ``thickened''  skeletons, i.e. by increasing the support of the operators retained in the truncated $\mathbb B_A$ matrix.

\begin{figure}
    \centering
    \includegraphics[width=0.47\textwidth]{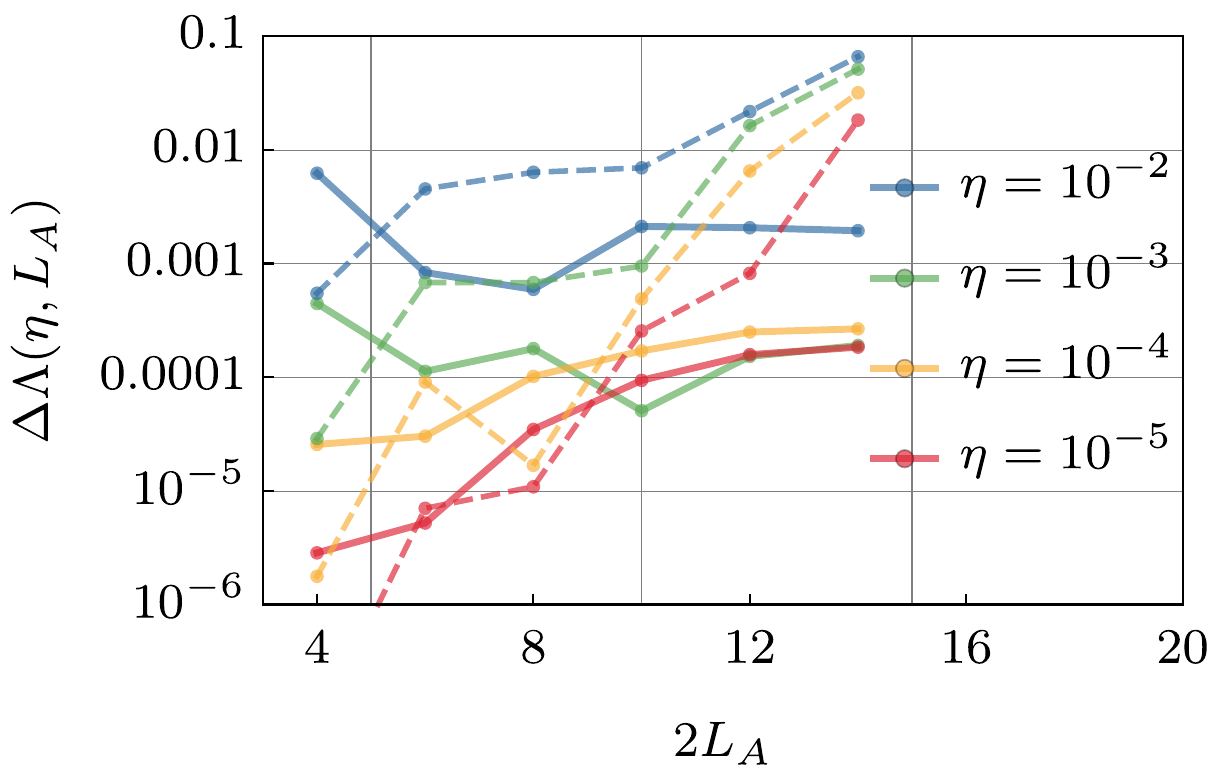}
    \caption{The absolute size of the errors in the  sub-leading eigenvalues of $\mathbb{B}_A$ by computing them from $(\mathbb{B}_A)_{m=1}$. Solid lines correspond to a case where the \ch{bare bones} condition holds while dashes lines to a case where the latter condition does not hold. In the second case, we see indications that the deviations grow quickly (perhaps exponentially) with $L_A$, as suggested by the arguments of the main text. The gate parameters are given in Tab.~\ref{tab:PrxGates}.}
    \label{fig:DeltaLambda}
\end{figure}

\subsection{Results on spectral statistics and thermalisation dynamics}
\label{sec:results}

Finally, let us look at the predictions that \eqref{eq:firstLambda} and \eqref{eq:SKreplacement} yield concerning physical properties. 

We begin by looking at the spectral form factor~\eqref{eq:SFFBA}: in this case the first order approximation \eqref{eq:firstLambda} produces useful qualitative predictions only at times that are multiples of $(2L_A-1)$
\ch{ 
\be
\label{eq:Kt1st}
K(t)=t (1+ \delta_{t \, {\rm mod}\, 2L_A-1,0} (2L_A-1) \Lambda_{0,{\rm lead}} ^t +\dots) ,
\ee
with $\Lambda_{0,{\rm lead}}$ defined in Eq. \eqref{eq:leadLambdas}. Since $\Lambda_{0,{\rm lead}}$ can be made arbitrarily small --- it scales with the perturbation strength --- we have that \eqref{eq:Kt1st} approaches the CUE spectral form factor for large times. In particular, from  \eqref{eq:leadLambdas} we find that the Thouless time (cf. \eqref{eq:taulambda}) scales as $\tau_{\Lambda} \propto -1/\log(\eta)$.} 

\ch{At times that are not multiples of $2L_A-1$, $K(t)-t$ becomes of higher order in $\eta$. Nonetheless, if the \ch{bare bones} condition holds, it can be accessed using the approximation \eqref{eq:SKreplacement}. Namely we have 
\be
K(t) \approx \tr (\mathbb B_A)_{m=1}^t, 
\ee
which can be evaluated with negligible computational cost. The accuracy of this approximation is illustrated in Fig.~\ref{fig:SFF}, where it is compared with exact numerics.}

\ch{A similar picture emerges when looking at the thermalisation of local observables: in perturbed dual-unitary circuits local observables thermalise in a timescale $\propto -1/\log(\eta)$. Moreover, if the \ch{bare bones} condition holds, their thermalisation dynamics can be quantitatively described with negligible computational cost. For instance, Fig.~\ref{fig:thermalisation} reports the comparison between the results of \eqref{eq:SKreplacement} and exact numerics for the time-evolution of \ch{third Pauli matrix acting on the central site of region $A$ (at $L_A/2$)} $\braket{Z_{L_A/2}(t)}$ after a quench from the state \eqref{eq:psi0diagram} with $d=2$, $\chi = 1$, and  
\be
[T]_{ab}^{cd} = \delta_{b,1}\delta_{c,1}\,. 
\ee
The latter is nothing but a fully polarised state in the subsystem $A$. Once again when the \ch{bare bones} condition holds the observed agreement is remarkable.}

\ch{Finally we recall that, even though here we assumed the perturbed gates to be the same for all $x$ and $t$, our analysis can be performed also for position-dependent perturbations. The main differences are: (i) the skeleton diagrams, like the one in Eq. \eqref{eq:skeleton}, are made out of position dependent quantum channels; (ii) the bulk matrices $B$ in Eq. \eqref{eq:BM1} are position dependent.}

\begin{figure}
    \centering
    \includegraphics[scale=0.65]{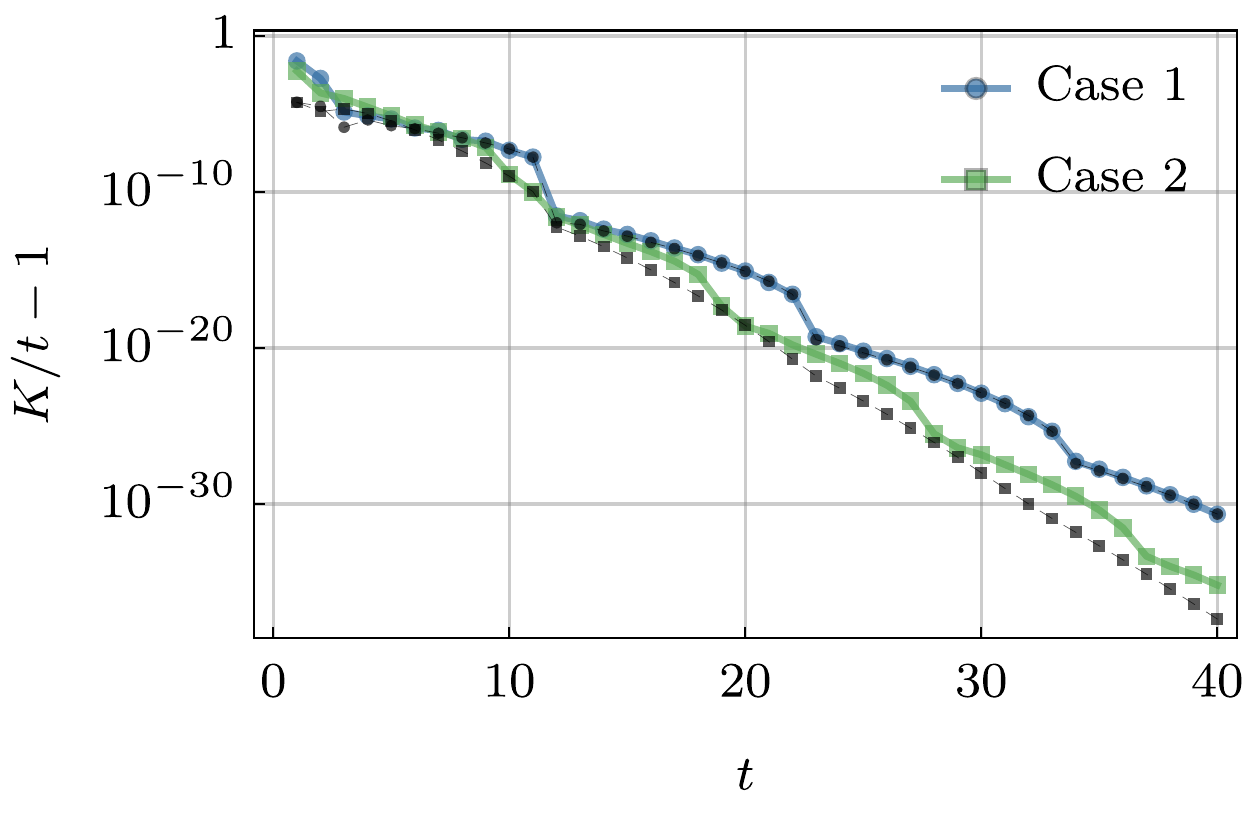}
    \caption{Numerically computed deviations of $K(t)$ from RMT (in color) and comparison with the approximation \eqref{eq:SKreplacement} (in black) at $\eta=0.001$. In case $1$ ($2$), the \ch{bare bones} condition is (not) fulfilled. 
    The gate parameters are given in Tab.~\ref{tab:PrxGates}, $2L_A=12$.}
    \label{fig:SFF}
\end{figure}

\begin{figure}
    \centering
    \includegraphics[scale=0.65]{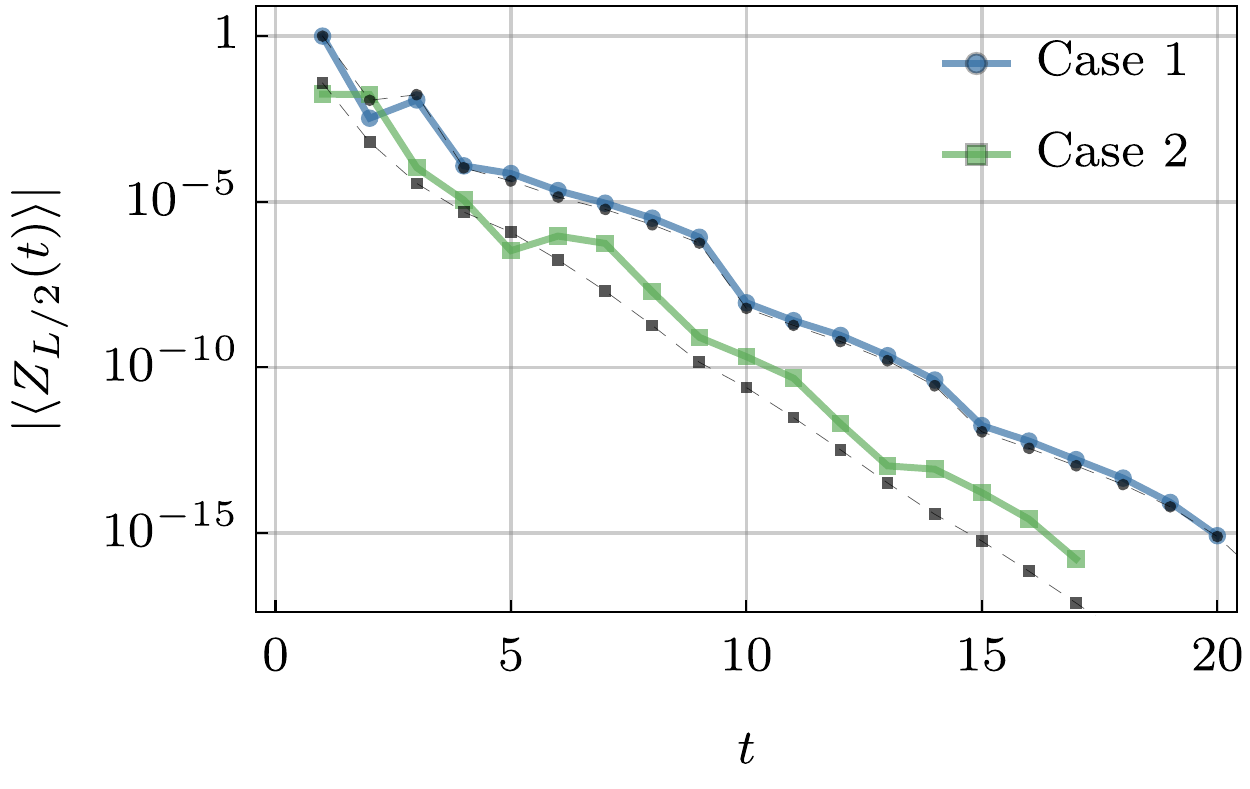}
    \caption{Numerically computed thermalisation of $\braket{Z_{L_A/2}(t)}$ (in color) and comparison with the approximation \eqref{eq:SKreplacement} (in black) at $\eta=0.001$. In case $1$ ($2$), the \ch{bare bones} condition is (not) fulfilled. We start from the exactly polarised initial state at the $A$ part of the system and solvable tensors at $B$. 
    The gate parameters are given in Tab.~\ref{tab:PrxGates}, $2L_A=12$.}
    \label{fig:thermalisation}
\end{figure}

%-------------------------------------------------------------------------------------------------------------
\section{Strong non-ergodicity}
\label{sec:StrongL}
%-------------------------------------------------------------------------------------------------------------

In the previous section we analysed a class of local gates for which all non-trivial eigenvalues of $\mathbb{B}_A$ are perturbatively small. In this section we focus on what can be considered the opposite extreme case: gates for which $\mathbb{B}_A$ has additional eigenvalues of magnitude exactly equal to $1$. \ch{In Sec. \ref{sec:map} we discussed how this implies absence of thermalisaing and chaotic behaviour.}
First we obtain some necessary conditions on $\{U_{x,\tau}\}$ valid for general $d$ \ch{(on site Hilbert space dimension)}. Then we specialise the treatment to translational invariant systems and $d=2$. In this case we find all explicit solutions of the constraints and use them to pinpoint all possible \ch{strongly non-ergodic} gates.

More specifically, here we are interested in finding the gates $U_{x,\tau}$ for which we have \ch{additional eigenvalues of magnitude one}
\be
\label{eq:eigenvector}
\mathbb B_A \ket{\Lambda} = e^{i \theta} \ket{\Lambda} \,,\qquad \theta\in\mathbb R\,,
\ee
with \ch{any vectorised operator $\ket{\Lambda}$ fulfilling}
\be
\label{eq:Lambda}
\|\ket{\Lambda}\|=1,\qquad \braket{\mcirc\ldots\mcirc|\Lambda}=0\,.
\ee
\ch{To clarify the physical meaning of these conditions let us make a simple observation that follows immediately from the unitarity of the gates: If (\ref{eq:eigenvector}, \ref{eq:Lambda}) hold the operator corresponding to $\ket{\Lambda}$ upon vectorisation is a \emph{still soliton} for \eqref{eq:diagramFloquet2} (cf. Ref.~\cite{bertini2020operator2}). Namely it is a local operator that maintains the same form when evolved with \eqref{eq:diagramFloquet2}. In general it acquires a phase $e^{i \theta t}$ while it is a \emph{bona fine} local conserved operator, i.e.\ a local integral of motion, for $\theta=0$.}

\ch{ The fact that the circuit admits non-trivial still solitons immediately explains the lack of thermalisation. The still solitons are able to ``store'' information locally preventing complete dissipation and hence relaxation to the infinite temperature state.} \ch{As discussed in Ref.~\cite{bertini2020operator2} this situation is similar to what happens in kinetically constrained models, like the celebrated PXP~\cite{turner2018weak, bernien2017probing}. In the latter case, however, the constraints imposed by the local conserved operators are typically weak enough to allow for large ergodic subspaces of the Hilbert space. Instead, strongly non-ergodic circuits can have a number of local constraints sufficiently large to ensure complete localisation.} 

\ch{In fact, also the relation between strong non-ergodicity and localisation is worth a comment. We begin by noting that strong non-ergodicity does not generically imply localisation: still solitons are not necessarily exactly conserved (because $\theta\neq0$) and even when conserved they do not necessarily form a complete set as it is required to have localisation~\cite{abanin2019many}. Perhaps more surprisingly, however, also the converse is true: localisation \emph{does not} imply strong non-ergodicity. Indeed, in the presence of localisation $1-|\Lambda|$ decays rapidly (i.e.\ exponentially) in the limit of infinite $L_A$ but it is not generically exactly equal to 0 as required by (\ref{eq:eigenvector}, \ref{eq:Lambda}). As an example, we illustrate the shrinking of the gap with increasing system size $L_A$ for an Anderson localising system: a kicked Ising model with random transverse field. This system can be written in circuit form with the following local gates
\be
U=  e^{i J Z\otimes Z}(e^{i h_x X} \otimes e^{i h_x' X}) e^{i J Z\otimes Z},
\label{eq:TFKI}
\ee
where we took $d=2$, $\{X, Y, Z\}$ denotes the  triple of Pauli matrices, and $h_x$ and $h_x'$ are random fields. In Fig.~\ref{fig:DisLambda} we show disordered-averaged gap versus the accessible system sizes.}

\begin{figure}
    \centering
    \includegraphics[width=0.45\textwidth]{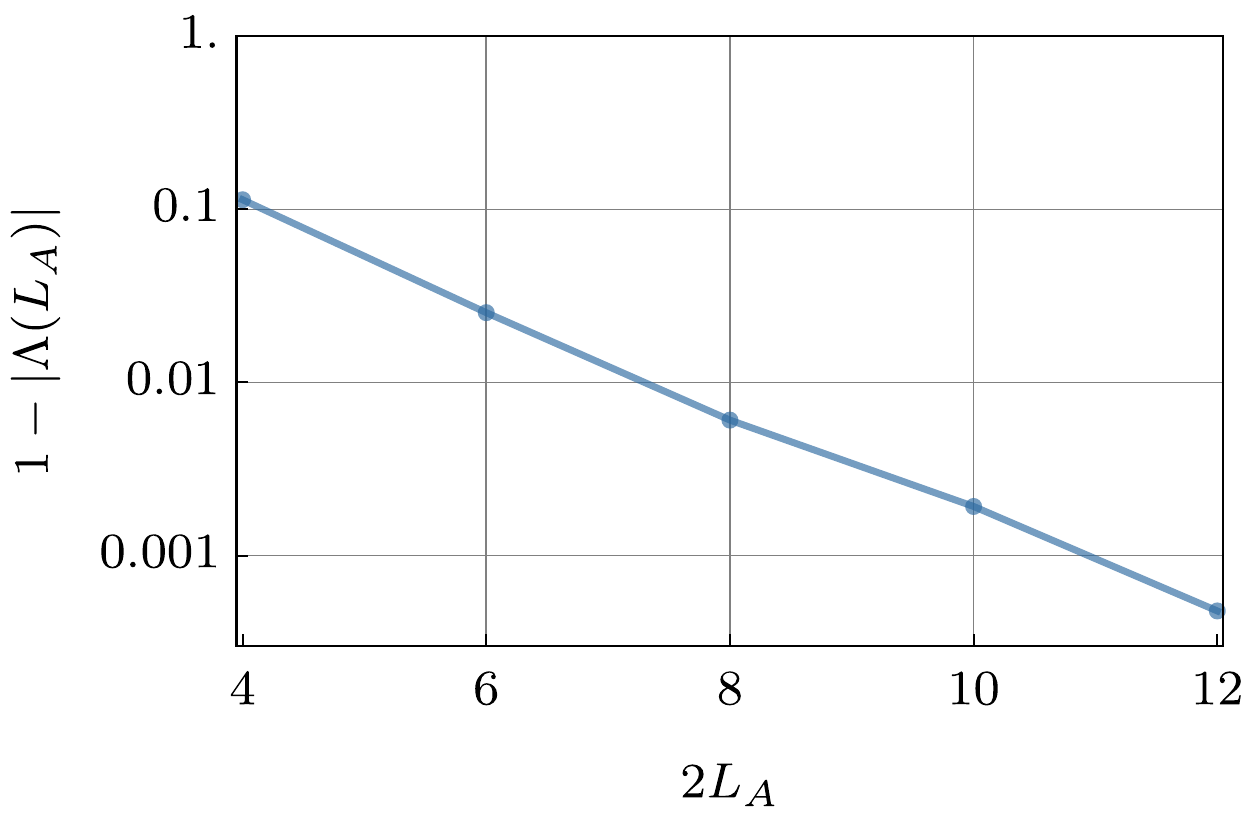}
    \caption{
    The size of the gap between the second largest eigenvalue of a $\mathbb B_A$ matrix built out of gates \eqref{eq:TFKI} and the unit circle versus system size. We averaged over $100$ realisations of the random fields and took $J=0.4$.}
    \label{fig:DisLambda}
\end{figure}

Let us now proceed finding some necessary conditions for (\ref{eq:eigenvector}, \ref{eq:Lambda}) to occur. 

%-----------------------------------------
\subsection{Necessary conditions for \ch{strong non-ergodicity}}
\label{sec:NecSL}
%-----------------------------------------
\ch{The relations (\ref{eq:eigenvector}, \ref{eq:Lambda}) imply non-trivial conditions on the local gates. Specifically, the latter are summarised by the following theorem}
%-------------------------------------------------------------
\begin{theorem}
If (\ref{eq:eigenvector}, \ref{eq:Lambda}) hold, there exist $x\leq y\in \mathbb Z_{2L_A}/2$ and four \ch{single-site} Hilbert-Schmidt normalised \ch{traceless} operators $a, b, a', b'$:% orthogonal to the identity for which: 
\begin{align}
&U_{x,\tau_x} (\1 \otimes a ) U_{x,\tau_x}^\dagger= \1\otimes a', \label{eq:StillMaps1x}\\
&U_{y,\tau_y}^\dagger ( b \otimes \1) U_{y,\tau_y}=  b' \otimes \1\,,
\label{eq:StillMaps1y}
\end{align}
where $\tau_{x,y}\in\{1/2,1\}$ and $x+\tau_{x},y+\tau_{y}\in \mathbb Z$ \ch{and $\1$ is a single-site identy operator}. Moreover, there exist two-site operators $a_2,b_2,a_2',b_2'$ such that 
\begin{align}
&U_{x+1/2,\tau'_x}^\dagger a_2 U_{x+1/2,\tau_x'}= a_2', \label{eq:StillMaps2x}\\
&U_{y-1/2,\tau_y'} b_2 U_{y-1/2,\tau_y'}^\dagger=  b_2'\,,
\label{eq:StillMaps2y}
\end{align}
and 
\begin{align}
&{\rm tr}_l(a_2 (a \otimes \1))\neq0,& &{\rm tr}_l(a_2' (a'\otimes \1))\neq0,\label{eq:a2site}\\
&{\rm tr}_r(b_2 (\1 \otimes b))\neq0, & &{\rm tr}_r(b_2' (\1 \otimes b') )\neq0.
\end{align}
Here ${\rm tr}_l$ and ${\rm tr}_r$ denotes the partial trace over the left and right site, and $\tau_{x,y}'=(\tau_{x,y}+1/2)\, {\rm mod}\, 1$. 
\label{thm:necessary}
\end{theorem}
%-------------------------------------------------------------

\begin{proof}%--------------------------------------------------------------
If (\ref{eq:eigenvector}, \ref{eq:Lambda}) hold, then 
\be
\bra{\Lambda}\mathbb B_A^\dagger \mathbb B_A^{\phantom{\dag}} \ket{\Lambda}=\bra{\Lambda}\mathbb B_A^{\phantom{\dag}} \mathbb B_A^\dag \ket{\Lambda} =1,
\ee
combining with the fact that $\mathbb B_A$ is non-expanding we have that $\ket{\Lambda}$ has to be a fixed point of both $\mathbb B_A^\dag \mathbb B^{\phantom{\dag}}_A$ and $\mathbb B^{\phantom{\dag}}_A \mathbb B_A^\dag$ which in turns imply that it has to be an eigenvector of $\mathbb B_A^\dag$. In formulae 
\be
\!\!\mathbb B_A^\dag \mathbb B^{\phantom{\dag}}_A \ket{\Lambda} =  \mathbb B^{\phantom{\dag}}_A \mathbb B_A^\dag \ket{\Lambda} = \ket{\Lambda},\,\quad \mathbb B_A^\dag \ket{\Lambda} = e^{-i \theta} \ket{\Lambda}\!.
\label{eq:BBdaglambda}
\ee
In particular, let us start by looking at $\mathbb B_A^\dagger \mathbb B_A^{\phantom{\dag}}$
\begin{align}
\mathbb B_A^\dagger \mathbb B_A^{\phantom{\dag}} &=
\begin{tikzpicture}[baseline=(current  bounding  box.center), scale=0.55]
\Wgategreen{0}{0}%\Wgategreen{2}{0}\Wgategreen{4}{0}\Wgategreen{1}{1}\Wgategreen{3}{1}
\Wgategreen{5}{1}
\Wgategreen{4}{0}
\Wgatedagger{0}{4} \Wgatedagger{5}{3}\Wgatedagger{4}{4}
\draw[thick] (.5,0.5) -- (.5,3.5) ; 
%\draw[thick] (4.5,-0.5) -- (4.5,0.5) ; 
\draw[thick] (4.5,1.5) -- (4.5,2.5) ; 
%\draw[thick] (4.5,3.5) -- (4.5,4.5) ; 
\draw[thick] (1.5,-.5) -- (1.5,4.5) ; 
\draw[thick] (2.5,-.5) -- (2.5,4.5) ;
\draw[thick] (3.5,0.5) -- (3.5,3.5) ;
\draw[thick, fill=white] (-0.5,-0.5) circle (0.1cm); 
\draw[thick, fill=white] (-0.5,0.5) circle (0.1cm); 
\draw[thick, fill=white] (5.5,1.5) circle (0.1cm); 
\draw[thick, fill=white] (5.5,0.5) circle (0.1cm); 
\draw[thick, fill=white] (-0.5,4-0.5) circle (0.1cm); 
\draw[thick, fill=white] (-0.5,4+0.5) circle (0.1cm); 
\draw[thick, fill=white] (5.5,2+1.5) circle (0.1cm); 
\draw[thick, fill=white] (5.5,2+0.5) circle (0.1cm); 
\end{tikzpicture}\\
&= M_{1/2,l} \otimes \1 \otimes \dots \otimes \tilde{M}_{L_A-1/2,r},
\label{eq:BdB}
\end{align} 
where we denoted $W_{x,\tau}^\dagger$ by red squares and introduced 
\begin{align}
&\!\!\!\!M_{1/2,l} = m_{1/2,l}^\dagger m^{\phantom \dagger}_{1/2,l}, \\
&\!\!\!\!\tilde{M}_{L_A-1/2,r} = W_{L_A-1/2,1/2}^\dag (\1 \otimes M_{L_A-1/2,r}) W_{L_A-1/2,1/2}^{\phantom \dagger}.
\end{align}
Similarly
\begin{align}
&\mathbb B_A^{\phantom{\dag}} \mathbb B_A^\dagger\notag\\
&= \!W_{1,1} m_{1/2,l}^{\phantom \dagger} m_{1/2,l}^\dagger W_{1,1}^\dagger \!\otimes\! \1 \!\dots \otimes\! m_{L-_A1/2,r}^{\phantom \dagger} m_{L_A-1/2,r}^\dagger \notag\\
&=:\! \tilde{M}_{1/2,l} \otimes \1 \dots \otimes M_{L_A-1/2,r}.
\end{align}
From the above and \eqref{eq:BBdaglambda} it follows that 
\be
\ket{\Lambda} \propto \sum_{i,k} \ket{a_i}\otimes\ket{\Lambda_{i,k}}\otimes \ket{b_{2,k}},
\label{eq:lambdasum}
\ee
where $\{\ket{a_i}\}\in \mathbb C^{d^2}$ and $\{\ket{b_{2,k}}\}\in \mathbb C^{d^2}\otimes \mathbb C^{d^2}$ are orthonormal fixed points (eigenvectors of eigenvalue 1) of $M_{1/2,l}$ and $\tilde{M}_{L_A-1/2,r}$, respectively. Noting that $\ket{a_0}=\ket{\mcirc}$ and $\ket{b_{2,0}}=\ket{\mcirc\mcirc}$ are  (trivial) fixed points of $M_{1/2,l}$ and $\tilde{M}_{L_A-1/2,r}$ and looking at the left boundary we then have two possibilities:

(i) Some $\ket{a_i}\neq \ket{\mcirc}$ appears in the sum \eqref{eq:lambdasum}, this immediately implies that \eqref{eq:StillMaps1x} must hold with $x=1/2$ and $\tau_x=1/2$. Indeed, we must have 
\be
m_{1/2,l} \ket{a_i} = \ket{a'_i}
\ee 
for some $\ket{a'_i}$ fulfilling $\braket{a'_i|a'_i}=1$ and $\braket{\mcirc|a'_i}=0$. Moreover, repeating the same reasoning for $\mathbb B_A^{\phantom{\dag}} \mathbb B_A^\dag$ we then have that 
\be
\ket{\Lambda} = \sum_{i,k} \ket{a_{2,i}}\otimes\ket{\tilde{\Lambda}_{i,k}}\otimes \ket{b_{k}},
\ee
where $\{\ket{a_{2,i}}\}\in \mathbb C^{d^2}\otimes \mathbb C^{d^2}$ and $\{\ket{b_{k}}\}\in \mathbb C^{d^2}$ are orthonormal fixed points of $\tilde{ M}_{1/2,l}$ and ${M_{L_A-1/2,r}}$.
This is compatible with \eqref{eq:lambdasum} only if the sum over $i$ includes some $n\neq0$ such that 
\be
{\rm tr}_l(a_{2, n} (a \otimes \1))\neq0.
\label{eq:condan}
\ee
To conclude we note that if $\ket{a_{2, n}}$ is a fixed point of $\tilde{M}_{1/2,l}$ and $a_{2,n}$ satisfies \eqref{eq:condan}, then $a_{2, n}$ must fulfil \eqref{eq:StillMaps2x}, with $x=1/2$, $\tau_x=1/2$ and for some $a_{2, n}'$ fulfilling the second of \eqref{eq:a2site}.

(ii) Only $\ket{a_0}=\ket{\mcirc}$ appears in the sum \eqref{eq:lambdasum}. Namely, we have  
\be
\mathbb{B}_A \ket{\mcirc}\otimes\ket{\Lambda'} = e^{i \theta} \ket{\mcirc}\otimes\ket{\Lambda'},
\ee
with $\ket{\Lambda'}\in(\mathbb C^{d^2})^{\otimes 2L_A-2}$. Tracing out the first site (i.e. contracting with $\ket{\mcirc}$) we get 
\be
\mathbb{B}_{A'} \ket{\Lambda'} = e^{i \theta} \ket{\Lambda'},
\ee
with $A'=A\setminus \{1/2\}$. Explicitly we have 
\begin{align}
\mathbb B_{A'} &= 
\begin{tikzpicture}[baseline=(current  bounding  box.center), scale=0.55]
%\Wgategreen{0}{0}
\Wgategreen{2}{0}\Wgategreen{4}{0}
\Wgategreen{1}{1}\Wgategreen{3}{1}\Wgategreen{5}{1}
\draw[thick, fill=white] (.5,1.5) circle (0.1cm); 
\draw[thick, fill=white] (.5,0.5) circle (0.1cm); 
\draw[thick, fill=white] (5.5,1.5) circle (0.1cm); 
\draw[thick, fill=white] (5.5,0.5) circle (0.1cm); 
\draw [thick, black,decorate,decoration={brace,amplitude=5pt,mirror},xshift=0.0pt,yshift=-0.0pt](4.5,1.7) -- (1.5,1.7) node[black,midway,yshift=0.4cm] { $A'$};
\end{tikzpicture}\,.
\label{eq:Bmap'}
\end{align}
Now we start the reasoning from the beginning for a smaller block $A'$ and repeat the analysis until we get to a case where the point (i) holds. This must happen for some $x\in \mathbb Z_{2L_A}/2$ because $\braket{\mcirc\ldots\mcirc|\Lambda}=0$.  

Repeating the same reasoning for the right boundary gives \eqref{eq:StillMaps1y} and \eqref{eq:StillMaps2y} for some $y\geq x$. 
 \end{proof}

%-----------------------------------------
\subsection{\ch{Strong non-ergodicity} for homogeneous circuits of qubits}
%-----------------------------------------
\label{sec:SLd2}

\ch{In Appendix \ref{app:d2homogenous} we derive all solutions to (\ref{eq:eigenvector}, \ref{eq:Lambda}) in the case of qubits, i.e.\ $d=2$, and for translationally invariant circuits, i.e.\ $U_{x,\tau}\equiv U$. There we derive additional constraints that almost entirely fix the form of the gate $U$. A simple numerical search leaves us with the following two families of strongly non-ergodic homogenous circuits in $d=2$}
\begin{align}
&\!\!\!\!U_1 \!=\!  (e^{i \phi Z} X^s \!\otimes\! \1 )e^{i J Z\otimes Z}, \label{eq:Still1}\\
&\!\!\!\!U_2  \!=\!  ( e^{i \pi Y/4}  Z^{s}\!\otimes\! \1)e^{i J Z\otimes Z} (e^{-i \pi Y/4}\! \otimes\! Z^{r}), 
\label{eq:Still2}
\end{align}
where $s,r\in\{0,1\}$. For both these families one can construct exponentially many (in $L_A$) still solitons, i.e., exponentially many non-trivial $\ket{\Lambda}$ fulfilling (\ref{eq:eigenvector}, \ref{eq:Lambda}). 

Explicitly, recalling that \ch{there are $2L_A-1$ sites in $A$}, we have that the $2^{2L_A-1}$ operators   
\be
Z^{s_1}\otimes\cdots\otimes Z^{s_{(2L_A)-1}},\quad s_i\in\{0,1\},
\ee
are still solitons for $\mathbb B_A(U_1)$, while the $2^{L_A}$ operators 
\begin{align}
\!\!\!&\!\!\1 \otimes (Z\otimes X)^{s_1}\otimes\cdots\otimes (Z\otimes X)^{s_{L_A-1}},\\
\!\!\!&\!\! (Z\otimes X)^{s_1}\otimes\cdots\otimes (Z\otimes X)^{s_{L_A-1}}\otimes \1,
\end{align}
${s_i\in\{0,1\}}$ are still solitons for a map $\mathbb B_A(U_2)$. Note that the latter operators are not the only still solitons of $\mathbb B_A(U_2)$. For example, also the operator $Z\otimes Y\otimes X$ is a still soliton of $\mathbb B_{A}(U_2)$ for $A=\{1/2,1,3/2\}$. 

\ch{This also shows that there are exponentially many (in $L_A$) unimodular eigenvalues of $\mathbb B_A$. The dynamics therefore has exponentially many non-decaying modes, and the spectral form factor is exponentially large in $L_A$, signalling non-chaotic behaviour.}

%----------------------------------------------------------------------------------------
\section{Conclusions}
\label{sec:conclusion}

We have studied the thermalisation and spectral properties of generic quantum circuits of a finite size that are attached to infinitely large dual-unitary circuits. Our strategy has been to express the evolution of relevant quantities in terms of an unital quantum channel (unital CPT map) $\mathbb B_A$. The spectrum of $\mathbb B_A$ determines whether subsystems thermalise and whether the spectral statistics of the total system \ch{(including the dual-unitary environment)} follows random matrix theory. We characterised the spectrum of $\mathbb B_A$ in two cases that display opposite extreme features: (i) nearly dual-unitary circuits and (ii) \ch{``{strongly non-ergodic}''} circuits, defined as those where $\mathbb B_A$ has additional unimodular eigenvalues at finite $L_A$.

For unperturbed dual-unitary circuits the spectrum of $\mathbb B_A$ contains only two eigenvalues: one, corresponding to the identity operator, and zero, corresponding to everything else. Instead, in the presence of small dual-unitarity breaking perturbations of strength $\eta\ll1$, some additional non-zero eigenvalues appear. We presented two increasingly more accurate schemes to determine such eigenvalues for a fixed subsystem size $L_A$ \ch{evolved with homogenous circuit (for simplicity)}. Firstly we computed them at the first non-trivial order in $\eta$ showing that they scale with a fractional power of $\eta$. Secondly, following Ref.~\cite{kos2021correlations}, we presented an uncontrolled approximation that gives highly accurate results when perturbing a specific family of dual-unitary gates. In general, we observed that as the size $L_A$ of the subsystem increases, the largest non-trivial eigenvalue $\Lambda$ grows initially, but then it appears to saturate at a value much smaller than $1$. This suggests that weak perturbations are not sufficient to localise dual unitary circuits. 

In contrast, we looked at \ch{strongly non-ergodic} circuits where $\mathbb B_A$ has additional unimodular eigenvalues.
These circuits admit non-trivial ``still solitons'', i.e.\ local operators that are left invariant by the time evolution (when evolving in time they are modified at most by a multiplicative phase). \ch{Thus they remain strictly localised.} \ch{Strongly non-ergodic} systems do not thermalise, and their spectral statistics does not follow the random matrix theory prediction. We characterised such circuits, determining a set of necessary conditions that the local gates have to satisfy for \ch{strong non-ergodicity} to emerge. For $d=2$, and in the presence of translational invariance, we explicitly solved the latter conditions and used them to find all possible \ch{strongly non-ergodic} gates. We identified two families: the first admits still solitons on a single site, while the second admits still solitons on two sites but not on one. 

Our results hold when a finite part of the circuit is immersed in an infinite dual-unitary circuit, which acts as a perfect thermalising reservoir. It is interesting to \ch{ask} whether similar results continue to hold even when the dual-unitary part is finite. Preliminary numerical observations seem to go in this direction, but we were not able to show this rigorously so far. A more careful analysis of this question, as well as the study of different semi-infinite systems providing more general reservoirs, is left to future research.   

Another exciting direction is exploiting the simple structure of \ch{strongly non-ergodic} gates in $d=2$ to compute their spectral statistics rigorously. Namely, one could evaluate the spectral form factor of a quantum circuit composed of \ch{strongly non-ergodic} gates proceeding in the same spirit of Refs.~\cite{bertini2021random, bertini2018exact}.

Finally, it would be interesting to understand further the fate of \ch{strong non-ergodicity} in the presence of small perturbations. Indeed, even though this property is by definition fragile, points in parameter space that are close enough to the \ch{strongly non-ergodic} manifold might still \ch{show non-ergodic behaviour}.

\section*{Acknowledgements}
We thank Felix Fritzsch for fruitful discussions and Sam Garratt for useful comments on the manuscript. This work has been supported by ERC Advanced grant 694544 -- OMNES and the program P1-0402 of Slovenian Research Agency (PK and TP) and by the Royal Society through the University Research Fellowship No.\ 201101 (BB).%-------------------------------------------------------------------------------------------------------------
\appendix
%-------------------------------------------------------------------------------------------------------------

\ch{
\section{Gauge transformation}
\label{app:gauge}
%Gauge-----------------
From the definition \eqref{eq:Floquet0} we see that the Floquet operator is covariant under certain transformations of the gates. In particular changing the gates as 
\begin{align}
U_{x,1/2} &\mapsto (v^{\dag}_{x-1/2}\otimes v^{\dag}_x) U_{x,1/2} (w_{x-1/2}\otimes w_x),\notag\\
U_{x,1}   &\mapsto (w^{\dag}_{x-1/2}\otimes w^{\dag}_x) U_{x,1} (v_{x-1/2}\otimes v_x),
\label{eq:gaugegen}
\end{align}
where $\{v_{x}, w_x\} \in{\rm U}(d)$ and $0 \equiv  L$, we have 
\be
\mathbb U \mapsto  \left(\bigotimes_x w_x^\dag\right) \cdot \mathbb U \cdot \left(\bigotimes_x w_x\right)\,.
\ee
This change does not affect the physical properties that we are interested to study and, therefore, we refer to \eqref{eq:gaugegen} as \emph{gauge transformation}. In particular, considering homogeneous systems $U_{x,\tau}=U$ and restricting to gauge transformations that preserve such symmetry we have 
\be
U\mapsto (u^\dag \otimes v^\dag) U (v \otimes u),\qquad v,u\in {\rm U}(d)\,.
\label{eq:gauge}
\ee
}
%-----------------------------------------------
\section{\ch{Classification of \ch{strong non-ergodicity} for homogeneous circuits of qubits}}
\label{app:d2homogenous}

Let us now proceed to find all solutions of the constraints (\ref{eq:StillMaps1x}, \ref{eq:StillMaps1y}, \ref{eq:StillMaps2x}, \ref{eq:StillMaps2y}) in the case of qubits, $d=2$, and for translationally invariant circuits, i.e. $U_{x,\tau}\equiv U$.

We begin by noting that we can choose a gauge transformation \ch{(see Appendix \ref{app:gauge})} that brings the conditions \eqref{eq:StillMaps1x} and \eqref{eq:StillMaps1y} to the following form  
\be
U^\dagger (\1 \otimes Z ) U= \1\otimes Z, \quad U^\dagger ( b \otimes \1) U= b' \otimes \1\ .
\label{eq:StillMaps}
\ee
while we recall that $\{X,Y,Z\}$ represents the Pauli basis. Using Property B.1 of Ref.~\cite{bertini2020operator2} we have that the first of \eqref{eq:StillMaps} implies that the gate can be written in the following form
\be
U = (v \otimes \1)\, e^{i J Z \otimes Z} (u \otimes \1),
\label{eq:gate}
\ee
where $u$ and $v$ are unitary $2\times 2$ matrices such that 
\be
 b= v Z v^\dag,\quad b'= u^\dag Z u,
\ee
and we use some remaining gauge freedom to remove $\1 \otimes e^{i \eta Z}$ from the gate.

Now we note that for the gate \eqref{eq:gate} and generic values of $J$ the constraint \eqref{eq:StillMaps1x} can only hold for $a,a'=\pm Z$, namely
\begin{lemma}
\label{lemmaunique}
For $J\in\,]0,\pi[\,\setminus\,\{\pi/2\}$ 
\be
U^\dag (\1\otimes  a) U  = \1\otimes a',
\ee
and ${\rm tr}[a]={\rm tr}[a']=0$ only if $a,a'=\pm Z$. 
\end{lemma}
\begin{proof}
This can be easily seen by writing explicitly  
\begin{align}
& {\rm tr}[  U (\1\otimes  a) U^\dag (\1\otimes a')]/4 \notag\\
&= \cos^2(J)\,  {\rm tr}[a a']/2  + \sin^2(J) {\rm tr}[ ZaZa']/2 \,.
\end{align}
The r.h.s.\ is zero for $a\in\{X,Y\}$ and $a'\in\{\1, Z\}$ (same for $a$ and $a'$ swapped), while it is one for $a,a'\in\{\1, Z\}$. Finally, for $a,a'\in\{X,Y\}$ we have 
\begin{align}
&|{\rm tr}[  U (\1\otimes  a) U^\dag (\1\otimes a')]/4|\notag\\
&=|\cos(2J)|  {\rm tr}[a a']/2\leq |\cos(2J)| < 1\,.
\end{align}
This proves the statement. 
\end{proof}

This lemma, and the fact that $U_{x,\tau}=U$ implies that the constraint \eqref{eq:StillMaps2x} must be fulfilled by $a_2$ and $a_2'$ of the form  
\begin{align}
a_2 &= c_0\, \1\otimes p+ c_z\, Z\otimes q,\\
a'_2 &= c'_0\, \1\otimes p'+ c'_z\, Z\otimes q',
\end{align}
where $p,q, p',q'$ are Hilbert-Schmidt normalised matrices and $c_0, c_0', c_z, c'_z$ complex coefficients, and $c_z, c'_z \neq 0$. This in turn means that defining the matrix
\be
[Q]_{ef}^{gh} := \frac{1}{4}{\rm tr}[  U^\dag (e\otimes  f) U (g\otimes h)]\,,
\label{eq:conditionf}
\ee
with $e,g\in\{\1,Z\}$ and $f,h\in\{\1,X,Y,Z\}$, we have that both $Q Q^\dag$ and $Q^\dag Q$ must have a fixed point of the form
\be
c_0 \ket{\mcirc p} + c_z \ket{Z q}
\ee
for some $\ket{p},\ket{q} \in\mathbb C^{d^2}$ and non-zero $c_z$.

Writing explicitly the matrix element  
\bw
\begin{align}
[Q]_{ef}^{gh} =& \frac{\cos^2(J)}{4} {\rm tr}[vu g u^\dag v^\dag e] {\rm tr}[fh] +\frac{\sin^2(J)}{4} {\rm tr}[v Z u g u^\dag Z v^\dag e] {\rm tr}[Z h Z f] \notag\\
&+\frac{i\cos(J)\sin(J)}{4} \left({\rm tr}[v Z u g u^\dag v^\dag e] {\rm tr}[f Z h] - {\rm tr}[v u g u^\dag Z v^\dag e] {\rm tr}[Z f h]\right),
\end{align}
we see that $Q$ is block diagonal in the Pauli basis. In particular ordering the basis as   
\be 
\{\ket{\mcirc\mcirc},\,\,\ket{\mcirc Z},\,\,\ket{Z \mcirc},\,\,\ket{Z Z},\,\,\ket{\mcirc X},\,\,\ket{Z Y},\,\,\ket{\mcirc Y},\,\,\ket{Z X}\},
\ee
we have 
\be
Q =
\begin{pmatrix}
\1_2 &  &\\
& K_1 & \\
& & K_2 \\
& & & Z  K_2 Z \\
\end{pmatrix},
\ee
where we introduced
\begin{align}
K_1 &= \begin{bmatrix}{\cos^2(J)} {\rm tr}[vu Z u^\dag v^\dag Z]/2 + {\sin^2(J)} {\rm tr}[v Z u Z u^\dag Z v^\dag Z]/2 & i \sin(2 J) \left({\rm tr}[vZu Z u^\dag v^\dag Z] - {\rm tr}[vuZ u^\dag Z v^\dag Z]\right)/4 \\
i \sin(2 J) \left({\rm tr}[vZu Z u^\dag v^\dag Z] - {\rm tr}[vuZ u^\dag Z v^\dag Z]\right)/4 &   {\cos^2(J)} {\rm tr}[vu Z u^\dag v^\dag Z]/2 + {\sin^2(J)} {\rm tr}[v Z u Z u^\dag Z v^\dag Z]/2\\
\end{bmatrix}\label{eq:K1matrix}\\
&= \begin{bmatrix}
\cos(\beta_1)\cos(\beta_2)-\sin(\beta_1)\sin(\beta_2)\cos(\alpha_1+\gamma_2)\cos(2J) & \sin(2J) \sin(\beta_1) \sin(\beta_2)  \sin(\alpha_1+\gamma_2)  \\
\sin(2J) \sin(\beta_1) \sin(\beta_2)  \sin(\alpha_1+\gamma_2) &  \cos(\beta_1)\cos(\beta_2)-\sin(\beta_1)\sin(\beta_2)\cos(\alpha_1+\gamma_2)\cos(2J)\\
\end{bmatrix}\notag\\
K_2 &= \begin{bmatrix}
\cos(2 J) & -\sin(2 J){{\rm tr}[v Z v^\dag Z]}/{2} \\
 {\sin(2 J)}{\rm tr}[u Z u^\dag Z]/{2} &  {\cos^2(J)} {\rm tr}[vu Z u^\dag v^\dag Z]/2-{\sin^2(J)} {\rm tr}[v Z u Z u^\dag Z v^\dag Z]/2\\
\end{bmatrix}\label{eq:K2matrix}\\
&= \begin{bmatrix}
\cos(2 J) & -\sin(2 J) \cos(\beta_2) \\
 {\sin(2 J)} \cos(\beta_1) &  -\sin(\beta_1)\sin(\beta_2)\cos(\alpha_1+\gamma_2)+\cos(2J)\cos(\beta_1)\cos(\beta_2)\\
\end{bmatrix}.\notag
\end{align}
\ew
where in the second steps of both \eqref{eq:K1matrix} and \eqref{eq:K2matrix} we parametrised the unitary matrices in terms of Euler angles 
\begin{align}
u & = e^{i \alpha_1 Z/2}e^{i \beta_1 Y/2}e^{i \gamma_1 Z/2},\\
v & = e^{i \alpha_2 Z/2}e^{i \beta_2 Y/2}e^{i \gamma_2 Z/2}.
\end{align}
The upper $2\times2$ block does not contribute to eigenvalues of the desired form. 
The second $2\times2$ block has two unimodular eigenvalues for 
\be
\beta_1,\beta_2 \in\{ 0, \pi\}.
\label{eq:consdstill}
\ee
This produces $QQ^\dag$ and $Q^\dag Q$ with eigenvalues one and eigenvectors of the desired form. Additional acceptable solutions are found for 
\be
\beta_1=\pm \pi/2,\quad\beta_2=\pm \pi/2, \quad\alpha_1+\gamma_2 \pm 2J\in\{0, \pi\}.
\label{eq:consdstill2}
\ee
The last two blocks produce $QQ^\dag$ and $Q^\dag Q$ with eigenvalues one for 
\begin{itemize}
\item[(i)] $\beta_1\in\{0, \pi\}$;
\item[(ii)] $\beta_2\in\{0, \pi\}$;
\item[(iii)] $\beta_1=\pm\beta_2,\,\, \gamma_1+\alpha_2\in\{0,\pi\}$;
\end{itemize}
but they are of the desired form only if either both (i) and (ii) hold simultaneously (i.e.~\eqref{eq:consdstill} holds) or (iii) holds.

\begin{figure}[t]
    \centering
    \includegraphics[scale=0.6]{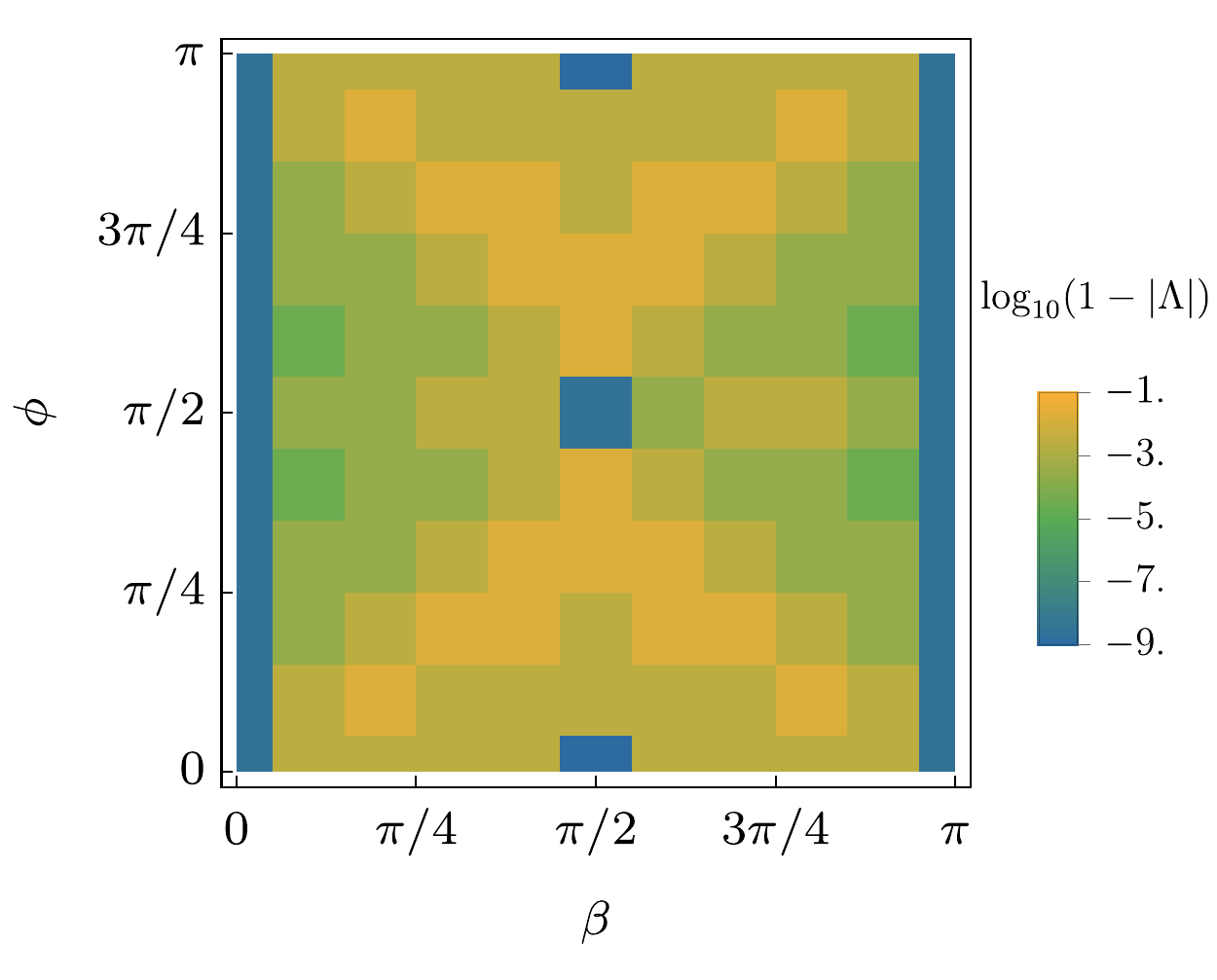}
    \caption{Gap between the largest eigenvalue and $1$ for the family of gates given in Eq.~\eqref{eq:gatesolitontwosites} at $2L_A=12$ and $J=0.3$. The families of solutions from Eqs.~\eqref{eq:Still1} and~\eqref{eq:Still2} appear with vanishing gap in blue. 
    }
    \label{fig:GapGrid}
\end{figure}

Putting all together we have only two possible choices for the gate \eqref{eq:gate}. If \eqref{eq:consdstill} holds the gate is gauge equivalent to  
\be
\tilde U_1 =  ( e^{i \phi Z} X^s \otimes \1)e^{i J Z\otimes Z}\,,\qquad s\in\{0,1\}\,.
\label{eq:gatesolitononesite}
\ee
If (iii) holds then the gate is gauge equivalent to 
\be
\!\!\!\!\tilde U_2 = ( e^{i \beta Y/2}  e^{i \phi Z}\otimes \1)e^{i J Z\otimes Z} (e^{-i \beta Y/2} \otimes Z^s)\,,\quad s\in\{0,1\}\,.
\label{eq:gatesolitontwosites}
\ee
It is easy to see that \eqref{eq:gatesolitononesite} and \eqref{eq:gatesolitontwosites} fulfil also the two remaining necessary conditions \eqref{eq:StillMaps1y} and \eqref{eq:StillMaps2y}. We begin by noting that, if a gate $U$ fulfils \eqref{eq:StillMaps1x} and \eqref{eq:StillMaps2x}, then $S U^\dag S$, where $S$ is the SWAP gate, fulfils \eqref{eq:StillMaps1y} and \eqref{eq:StillMaps2y}. 
Second we note that performing a gauge transformation
 one can bring $S \tilde U_1^\dag S$ to the form \eqref{eq:gatesolitononesite} and $S \tilde U_2^\dag S$ to the form 
\be
( e^{-i \beta \tilde Y/2}  e^{- i \phi Z}\otimes \1)e^{- i J Z\otimes Z} (e^{i \beta \tilde Y/2} \otimes Z^s)\,,\quad s\in\{0,1\},
\ee
where $\tilde a = e^{- i \phi Z} a e^{i \phi Z}$. This gate is  brought to the form \eqref{eq:gatesolitontwosites} with the change of basis $a\mapsto \tilde a$. This concludes the argument: since $\tilde U_1$ and $\tilde U_2$ are equivalent to $S \tilde U_1^\dag S$ and $S \tilde U_2^\dag S$ they also fulfil \eqref{eq:StillMaps1y} and \eqref{eq:StillMaps2y}.   

Having determined all possible solutions to the necessary conditions (\ref{eq:StillMaps1x}, \ref{eq:StillMaps1y}, \ref{eq:StillMaps2x}, \ref{eq:StillMaps2y}) let us identify which ones actually admit still solitons. It is easy to verify that the one-site operator $Z$ is the still soliton of minimal support for the map $\mathbb B_A(\tilde U_1)$, i.e. a map of the form \eqref{eq:Bmap} constructed with gates $\tilde U_1$. Moreover, the two site operator $Z\otimes X$ is the still soliton of minimal support for the map $\mathbb B_A(\tilde U_2)$ for 
\be
\beta=\pm \pi/2,\qquad 
 \phi\in\{0, \pi/2,\pi,3\pi/2\}. 
\label{eq:parametersSS}
\ee
An extensive numerical search considering $\mathbb B_A(\tilde U_2)$ for $2L_A\leq 14$ suggests that the parameter choice \eqref{eq:parametersSS} is indeed the only one giving still solitons for gates of the form \eqref{eq:gatesolitontwosites} (see Fig.~\ref{fig:GapGrid}).

This leaves us with the following two families of \ch{strongly non-ergodic} gates in $d=2$ 
\begin{align}
&\!\!\!\!U_1 \!=\!  (e^{i \phi Z} X^s \!\otimes\! \1 )e^{i J Z\otimes Z},\label{eq:Still1}\\
&\!\!\!\!U_2  \!=\!  ( e^{i \pi Y/4}  Z^{s}\!\otimes\! \1)e^{i J Z\otimes Z} (e^{-i \pi Y/4}\! \otimes\! Z^{r}),
\label{eq:Still2}
\end{align}
where $s,r\in\{0,1\}$.

%-------------------------------------------------
\section{Numerical methods and parameters of the gates}
%-------------------------------------------------

Numerical results were obtained by implementing an efficient action of $\mathbb{B}_A$ on a vector. This was achieved by contracting the vector by each gate separately, with the help of some basic functionalities of ITensor Library~\cite{fishman2020itensor}. To obtain the leading eigenvalues and eigenvectors, we implemented the power method and the Arnoldi iteration. Power method worked better with single sub-leading eigenvalue, whereas Arnoldi worked better in the case of multiple sub-leading eigenvalues of the same size.

In Tab. \ref{tab:PrxGates} we show the numerical values of the parameters of the dual-unitary perturbed gates.

\begin{table*}
\centering
\begin {tabular} {c | cc}
\text{Is {\em \ch{bare bones}}} \\
\text{{\em condition} fulfilled?} & $u$ & $v$\\
\hline
\text {Yes} & $\begin{pmatrix} -0.229466-0.562418 i & 0.507409\, -0.611202 i \\ 0.316533\, +0.728586 i & 0.377401\, -0.475959 i \end{pmatrix}$
& $\begin{pmatrix} -0.500163-0.421942 i & -0.0342946-0.755398 i \\  0.532173\, -0.537209 i & -0.653978-0.0226033 i  \end{pmatrix}$\\
\\
\text {No} & $\begin{pmatrix} -0.484203+0.476337 i & 0.207991\, +0.70384 i \\ -0.194545+0.707674 i & -0.493189-0.467026 i \\\end{pmatrix}$
& $\begin{pmatrix} 0.488082\, +0.160083 i & 0.0340135\, +0.857317 i \\ 0.197049\, +0.835058 i & 0.427305\, -0.285063 i \\  \end{pmatrix}$
\end{tabular}\\
  \caption{Parameters of the non-averaged gates, which were used in the figures.
  The gates are parametrized as $ (u \otimes u) V[\{J_i\}] (v \otimes v)$, with $V[\{J_j\}] \!=\! \exp\bigl[i(J_1 \,X\! \otimes  X + J_2 \,Y\! \otimes  Y+J_3 \,Z\! \otimes  Z)\bigr]$. The parameters are set to $J_1=J_2=\eta + \pi/4$ and $J_3=0.1$. The parameters match those in~\cite{kos2021correlations}.
  }
  \label{tab:PrxGates}
  \end{table*}

%-------------------------------------------------

\bibliography{bibliography}

\begin{thebibliography}{53}%
\makeatletter
\providecommand \@ifxundefined [1]{%
 \@ifx{#1\undefined}
}%
\providecommand \@ifnum [1]{%
 \ifnum #1\expandafter \@firstoftwo
 \else \expandafter \@secondoftwo
 \fi
}%
\providecommand \@ifx [1]{%
 \ifx #1\expandafter \@firstoftwo
 \else \expandafter \@secondoftwo
 \fi
}%
\providecommand \natexlab [1]{#1}%
\providecommand \enquote  [1]{``#1''}%
\providecommand \bibnamefont  [1]{#1}%
\providecommand \bibfnamefont [1]{#1}%
\providecommand \citenamefont [1]{#1}%
\providecommand \href@noop [0]{\@secondoftwo}%
\providecommand \href [0]{\begingroup \@sanitize@url \@href}%
\providecommand \@href[1]{\@@startlink{#1}\@@href}%
\providecommand \@@href[1]{\endgroup#1\@@endlink}%
\providecommand \@sanitize@url [0]{\catcode `\\12\catcode `\$12\catcode
  `\&12\catcode `\#12\catcode `\^12\catcode `\_12\catcode `\%12\relax}%
\providecommand \@@startlink[1]{}%
\providecommand \@@endlink[0]{}%
\providecommand \url  [0]{\begingroup\@sanitize@url \@url }%
\providecommand \@url [1]{\endgroup\@href {#1}{\urlprefix }}%
\providecommand \urlprefix  [0]{URL }%
\providecommand \Eprint [0]{\href }%
\providecommand \doibase [0]{https://doi.org/}%
\providecommand \selectlanguage [0]{\@gobble}%
\providecommand \bibinfo  [0]{\@secondoftwo}%
\providecommand \bibfield  [0]{\@secondoftwo}%
\providecommand \translation [1]{[#1]}%
\providecommand \BibitemOpen [0]{}%
\providecommand \bibitemStop [0]{}%
\providecommand \bibitemNoStop [0]{.\EOS\space}%
\providecommand \EOS [0]{\spacefactor3000\relax}%
\providecommand \BibitemShut  [1]{\csname bibitem#1\endcsname}%
\let\auto@bib@innerbib\@empty
%</preamble>
\bibitem [{\citenamefont {Polkovnikov}\ \emph {et~al.}(2011)\citenamefont
  {Polkovnikov}, \citenamefont {Sengupta}, \citenamefont {Silva},\ and\
  \citenamefont {Vengalattore}}]{polkovnikov2011}%
  \BibitemOpen
  \bibfield  {author} {\bibinfo {author} {\bibfnamefont {A.}~\bibnamefont
  {Polkovnikov}}, \bibinfo {author} {\bibfnamefont {K.}~\bibnamefont
  {Sengupta}}, \bibinfo {author} {\bibfnamefont {A.}~\bibnamefont {Silva}},\
  and\ \bibinfo {author} {\bibfnamefont {M.}~\bibnamefont {Vengalattore}},\
  }\bibfield  {title} {\bibinfo {title} {Colloquium: Nonequilibrium dynamics of
  closed interacting quantum systems},\ }\href
  {https://doi.org/10.1103/RevModPhys.83.863} {\bibfield  {journal} {\bibinfo
  {journal} {Rev. Mod. Phys.}\ }\textbf {\bibinfo {volume} {83}},\ \bibinfo
  {pages} {863} (\bibinfo {year} {2011})}\BibitemShut {NoStop}%
\bibitem [{\citenamefont {Eisert}\ \emph {et~al.}(2015)\citenamefont {Eisert},
  \citenamefont {Friesdorf},\ and\ \citenamefont
  {Gogolin}}]{eisert2015quantum}%
  \BibitemOpen
  \bibfield  {author} {\bibinfo {author} {\bibfnamefont {J.}~\bibnamefont
  {Eisert}}, \bibinfo {author} {\bibfnamefont {M.}~\bibnamefont {Friesdorf}},\
  and\ \bibinfo {author} {\bibfnamefont {C.}~\bibnamefont {Gogolin}},\
  }\bibfield  {title} {\bibinfo {title} {Quantum many-body systems out of
  equilibrium},\ }\href {http://www.nature.com/articles/nphys3215} {\bibfield
  {journal} {\bibinfo  {journal} {Nat. Phys.}\ }\textbf {\bibinfo {volume}
  {11}},\ \bibinfo {pages} {124} (\bibinfo {year} {2015})}\BibitemShut
  {NoStop}%
\bibitem [{\citenamefont {Gogolin}\ and\ \citenamefont
  {Eisert}(2016)}]{gogolin2016equilibration}%
  \BibitemOpen
  \bibfield  {author} {\bibinfo {author} {\bibfnamefont {C.}~\bibnamefont
  {Gogolin}}\ and\ \bibinfo {author} {\bibfnamefont {J.}~\bibnamefont
  {Eisert}},\ }\bibfield  {title} {\bibinfo {title} {Equilibration,
  thermalisation, and the emergence of statistical mechanics in closed quantum
  systems},\ }\href {https://doi.org/10.1088/0034-4885/79/5/056001} {\bibfield
  {journal} {\bibinfo  {journal} {Rep. Prog. Phys.}\ }\textbf {\bibinfo
  {volume} {79}},\ \bibinfo {pages} {056001} (\bibinfo {year}
  {2016})}\BibitemShut {NoStop}%
\bibitem [{\citenamefont {Calabrese}\ \emph {et~al.}(2016)\citenamefont
  {Calabrese}, \citenamefont {Essler},\ and\ \citenamefont
  {Mussardo}}]{calabrese2016introduction}%
  \BibitemOpen
  \bibfield  {author} {\bibinfo {author} {\bibfnamefont {P.}~\bibnamefont
  {Calabrese}}, \bibinfo {author} {\bibfnamefont {F.~H.}\ \bibnamefont
  {Essler}},\ and\ \bibinfo {author} {\bibfnamefont {G.}~\bibnamefont
  {Mussardo}},\ }\bibfield  {title} {\bibinfo {title} {Introduction to quantum
  integrability in out of equilibrium systems},\ }\href
  {https://doi.org/10.1088/1742-5468/2016/06/064001} {\bibfield  {journal}
  {\bibinfo  {journal} {J. Stat. Mech.}\ }\textbf {\bibinfo {volume} {2016}},\
  \bibinfo {pages} {064001} (\bibinfo {year} {2016})}\BibitemShut {NoStop}%
\bibitem [{\citenamefont {D'Alessio}\ \emph {et~al.}(2016)\citenamefont
  {D'Alessio}, \citenamefont {Kafri}, \citenamefont {Polkovnikov},\ and\
  \citenamefont {Rigol}}]{dalessio2016quantum}%
  \BibitemOpen
  \bibfield  {author} {\bibinfo {author} {\bibfnamefont {L.}~\bibnamefont
  {D'Alessio}}, \bibinfo {author} {\bibfnamefont {Y.}~\bibnamefont {Kafri}},
  \bibinfo {author} {\bibfnamefont {A.}~\bibnamefont {Polkovnikov}},\ and\
  \bibinfo {author} {\bibfnamefont {M.}~\bibnamefont {Rigol}},\ }\bibfield
  {title} {\bibinfo {title} {From quantum chaos and eigenstate thermalization
  to statistical mechanics and thermodynamics},\ }\href
  {https://doi.org/10.1080/00018732.2016.1198134} {\bibfield  {journal}
  {\bibinfo  {journal} {Adv. Phys.}\ }\textbf {\bibinfo {volume} {65}},\
  \bibinfo {pages} {239} (\bibinfo {year} {2016})}\BibitemShut {NoStop}%
\bibitem [{\citenamefont {Yukalov}(2011)}]{yukalov2011equilibration}%
  \BibitemOpen
  \bibfield  {author} {\bibinfo {author} {\bibfnamefont {V.}~\bibnamefont
  {Yukalov}},\ }\bibfield  {title} {\bibinfo {title} {Equilibration and
  thermalization in finite quantum systems},\ }\href@noop {} {\bibfield
  {journal} {\bibinfo  {journal} {Laser Physics Letters}\ }\textbf {\bibinfo
  {volume} {8}},\ \bibinfo {pages} {485} (\bibinfo {year} {2011})}\BibitemShut
  {NoStop}%
\bibitem [{\citenamefont {Essler}\ and\ \citenamefont
  {Fagotti}(2016)}]{essler2016quench}%
  \BibitemOpen
  \bibfield  {author} {\bibinfo {author} {\bibfnamefont {F.~H.}\ \bibnamefont
  {Essler}}\ and\ \bibinfo {author} {\bibfnamefont {M.}~\bibnamefont
  {Fagotti}},\ }\bibfield  {title} {\bibinfo {title} {Quench dynamics and
  relaxation in isolated integrable quantum spin chains},\ }\href
  {https://iopscience.iop.org/article/10.1088/1742-5468/2016/06/064002}
  {\bibfield  {journal} {\bibinfo  {journal} {J. Stat. Mech.}\ }\textbf
  {\bibinfo {volume} {2016}},\ \bibinfo {pages} {064002} (\bibinfo {year}
  {2016})}\BibitemShut {NoStop}%
\bibitem [{\citenamefont {Bertini}\ \emph
  {et~al.}(2019{\natexlab{a}})\citenamefont {Bertini}, \citenamefont {Kos},\
  and\ \citenamefont {Prosen}}]{bertini2019entanglement}%
  \BibitemOpen
  \bibfield  {author} {\bibinfo {author} {\bibfnamefont {B.}~\bibnamefont
  {Bertini}}, \bibinfo {author} {\bibfnamefont {P.}~\bibnamefont {Kos}},\ and\
  \bibinfo {author} {\bibfnamefont {T.}~\bibnamefont {Prosen}},\ }\bibfield
  {title} {\bibinfo {title} {Entanglement spreading in a minimal model of
  maximal many-body quantum chaos},\ }\href
  {https://doi.org/10.1103/physrevx.9.021033} {\bibfield  {journal} {\bibinfo
  {journal} {Phys. Rev. X}\ }\textbf {\bibinfo {volume} {9}},\ \bibinfo {pages}
  {021033} (\bibinfo {year} {2019}{\natexlab{a}})}\BibitemShut {NoStop}%
\bibitem [{\citenamefont {Piroli}\ \emph {et~al.}(2020)\citenamefont {Piroli},
  \citenamefont {Bertini}, \citenamefont {Cirac},\ and\ \citenamefont
  {Prosen}}]{piroli2020exact}%
  \BibitemOpen
  \bibfield  {author} {\bibinfo {author} {\bibfnamefont {L.}~\bibnamefont
  {Piroli}}, \bibinfo {author} {\bibfnamefont {B.}~\bibnamefont {Bertini}},
  \bibinfo {author} {\bibfnamefont {J.~I.}\ \bibnamefont {Cirac}},\ and\
  \bibinfo {author} {\bibfnamefont {T.}~\bibnamefont {Prosen}},\ }\bibfield
  {title} {\bibinfo {title} {Exact dynamics in dual-unitary quantum circuits},\
  }\href {https://doi.org/10.1103/physrevb.101.094304} {\bibfield  {journal}
  {\bibinfo  {journal} {Phys. Rev. B}\ }\textbf {\bibinfo {volume} {101}},\
  \bibinfo {pages} {094304} (\bibinfo {year} {2020})}\BibitemShut {NoStop}%
\bibitem [{\citenamefont {Suzuki}\ \emph {et~al.}(2021)\citenamefont {Suzuki},
  \citenamefont {Mitarai},\ and\ \citenamefont
  {Fujii}}]{suzuki2021computational}%
  \BibitemOpen
  \bibfield  {author} {\bibinfo {author} {\bibfnamefont {R.}~\bibnamefont
  {Suzuki}}, \bibinfo {author} {\bibfnamefont {K.}~\bibnamefont {Mitarai}},\
  and\ \bibinfo {author} {\bibfnamefont {K.}~\bibnamefont {Fujii}},\
  }\href@noop {} {\bibinfo {title} {Computational power of one- and
  two-dimensional dual-unitary quantum circuits}} (\bibinfo {year} {2021}),\
  \Eprint {https://arxiv.org/abs/2103.09211} {arXiv:2103.09211 [quant-ph]}
  \BibitemShut {NoStop}%
\bibitem [{\citenamefont {Bertini}\ \emph
  {et~al.}(2019{\natexlab{b}})\citenamefont {Bertini}, \citenamefont {Kos},\
  and\ \citenamefont {Prosen}}]{bertini2019exact}%
  \BibitemOpen
  \bibfield  {author} {\bibinfo {author} {\bibfnamefont {B.}~\bibnamefont
  {Bertini}}, \bibinfo {author} {\bibfnamefont {P.}~\bibnamefont {Kos}},\ and\
  \bibinfo {author} {\bibfnamefont {T.}~\bibnamefont {Prosen}},\ }\bibfield
  {title} {\bibinfo {title} {Exact correlation functions for dual-unitary
  lattice models in 1+1 dimensions},\ }\href
  {https://doi.org/10.1103/physrevlett.123.210601} {\bibfield  {journal}
  {\bibinfo  {journal} {Phys. Rev. Lett.}\ }\textbf {\bibinfo {volume} {123}},\
  \bibinfo {pages} {210601} (\bibinfo {year} {2019}{\natexlab{b}})}\BibitemShut
  {NoStop}%
\bibitem [{\citenamefont {Klobas}\ \emph {et~al.}(2021)\citenamefont {Klobas},
  \citenamefont {Bertini},\ and\ \citenamefont {Piroli}}]{klobas2021exact}%
  \BibitemOpen
  \bibfield  {author} {\bibinfo {author} {\bibfnamefont {K.}~\bibnamefont
  {Klobas}}, \bibinfo {author} {\bibfnamefont {B.}~\bibnamefont {Bertini}},\
  and\ \bibinfo {author} {\bibfnamefont {L.}~\bibnamefont {Piroli}},\
  }\bibfield  {title} {\bibinfo {title} {Exact thermalization dynamics in the
  ``rule 54'' quantum cellular automaton},\ }\href
  {https://doi.org/10.1103/PhysRevLett.126.160602} {\bibfield  {journal}
  {\bibinfo  {journal} {Phys. Rev. Lett.}\ }\textbf {\bibinfo {volume} {126}},\
  \bibinfo {pages} {160602} (\bibinfo {year} {2021})}\BibitemShut {NoStop}%
\bibitem [{\citenamefont {Klobas}\ and\ \citenamefont
  {Bertini}(2021{\natexlab{a}})}]{klobas2021exactII}%
  \BibitemOpen
  \bibfield  {author} {\bibinfo {author} {\bibfnamefont {K.}~\bibnamefont
  {Klobas}}\ and\ \bibinfo {author} {\bibfnamefont {B.}~\bibnamefont
  {Bertini}},\ }\href@noop {} {\bibinfo {title} {Exact relaxation to gibbs and
  non-equilibrium steady states in the quantum cellular automaton rule 54}}
  (\bibinfo {year} {2021}{\natexlab{a}}),\ \Eprint
  {https://arxiv.org/abs/2104.04511} {arXiv:2104.04511 [cond-mat.stat-mech]}
  \BibitemShut {NoStop}%
\bibitem [{\citenamefont {Klobas}\ and\ \citenamefont
  {Bertini}(2021{\natexlab{b}})}]{klobas2021entanglement}%
  \BibitemOpen
  \bibfield  {author} {\bibinfo {author} {\bibfnamefont {K.}~\bibnamefont
  {Klobas}}\ and\ \bibinfo {author} {\bibfnamefont {B.}~\bibnamefont
  {Bertini}},\ }\href@noop {} {\bibinfo {title} {Entanglement dynamics in rule
  54: exact results and quasiparticle picture}} (\bibinfo {year}
  {2021}{\natexlab{b}}),\ \Eprint {https://arxiv.org/abs/2104.04513}
  {arXiv:2104.04513 [cond-mat.stat-mech]} \BibitemShut {NoStop}%
\bibitem [{\citenamefont {Pozsgay}(2014)}]{pozsgay2014quantum}%
  \BibitemOpen
  \bibfield  {author} {\bibinfo {author} {\bibfnamefont {B.}~\bibnamefont
  {Pozsgay}},\ }\bibfield  {title} {\bibinfo {title} {Quantum quenches and
  generalized {G}ibbs ensemble in a {B}ethe ansatz solvable lattice model of
  interacting bosons},\ }\href
  {https://doi.org/10.1088/1742-5468/2014/10/p10045} {\bibfield  {journal}
  {\bibinfo  {journal} {J. Stat. Mech.}\ }\textbf {\bibinfo {volume} {2014}},\
  \bibinfo {pages} {P10045} (\bibinfo {year} {2014})}\BibitemShut {NoStop}%
\bibitem [{\citenamefont {Pozsgay}\ and\ \citenamefont
  {Eisler}(2016)}]{pozsgay2016real}%
  \BibitemOpen
  \bibfield  {author} {\bibinfo {author} {\bibfnamefont {B.}~\bibnamefont
  {Pozsgay}}\ and\ \bibinfo {author} {\bibfnamefont {V.}~\bibnamefont
  {Eisler}},\ }\bibfield  {title} {\bibinfo {title} {Real-time dynamics in a
  strongly interacting bosonic hopping model: Global quenches and mapping to
  the {XX} chain},\ }\href {https://doi.org/10.1088/1742-5468/2016/05/053107}
  {\bibfield  {journal} {\bibinfo  {journal} {J. Stat. Mech.}\ }\textbf
  {\bibinfo {volume} {2016}},\ \bibinfo {pages} {053107} (\bibinfo {year}
  {2016})}\BibitemShut {NoStop}%
\bibitem [{\citenamefont {Zadnik}\ and\ \citenamefont
  {Fagotti}(2021)}]{zadnik2021folded}%
  \BibitemOpen
  \bibfield  {author} {\bibinfo {author} {\bibfnamefont {L.}~\bibnamefont
  {Zadnik}}\ and\ \bibinfo {author} {\bibfnamefont {M.}~\bibnamefont
  {Fagotti}},\ }\bibfield  {title} {\bibinfo {title} {{The Folded Spin-1/2 XXZ
  Model: I. Diagonalisation, Jamming, and Ground State Properties}},\ }\href
  {https://doi.org/10.21468/SciPostPhysCore.4.2.010} {\bibfield  {journal}
  {\bibinfo  {journal} {SciPost Phys. Core}\ }\textbf {\bibinfo {volume} {4}},\
  \bibinfo {pages} {10} (\bibinfo {year} {2021})}\BibitemShut {NoStop}%
\bibitem [{\citenamefont {Zadnik}\ \emph {et~al.}(2021)\citenamefont {Zadnik},
  \citenamefont {Bidzhiev},\ and\ \citenamefont
  {Fagotti}}]{zadnik2021foldedII}%
  \BibitemOpen
  \bibfield  {author} {\bibinfo {author} {\bibfnamefont {L.}~\bibnamefont
  {Zadnik}}, \bibinfo {author} {\bibfnamefont {K.}~\bibnamefont {Bidzhiev}},\
  and\ \bibinfo {author} {\bibfnamefont {M.}~\bibnamefont {Fagotti}},\
  }\bibfield  {title} {\bibinfo {title} {{The Folded Spin-1/2 XXZ Model: II.
  Thermodynamics and Hydrodynamics with a Minimal Set of Charges}},\ }\href
  {https://doi.org/10.21468/SciPostPhys.10.5.099} {\bibfield  {journal}
  {\bibinfo  {journal} {SciPost Phys.}\ }\textbf {\bibinfo {volume} {10}},\
  \bibinfo {pages} {99} (\bibinfo {year} {2021})}\BibitemShut {NoStop}%
\bibitem [{\citenamefont {Garratt}\ and\ \citenamefont
  {Chalker}(2021{\natexlab{a}})}]{garratt2021local}%
  \BibitemOpen
  \bibfield  {author} {\bibinfo {author} {\bibfnamefont {S.~J.}\ \bibnamefont
  {Garratt}}\ and\ \bibinfo {author} {\bibfnamefont {J.~T.}\ \bibnamefont
  {Chalker}},\ }\bibfield  {title} {\bibinfo {title} {Local pairing of feynman
  histories in many-body floquet models},\ }\href
  {https://doi.org/10.1103/PhysRevX.11.021051} {\bibfield  {journal} {\bibinfo
  {journal} {Phys. Rev. X}\ }\textbf {\bibinfo {volume} {11}},\ \bibinfo
  {pages} {021051} (\bibinfo {year} {2021}{\natexlab{a}})}\BibitemShut
  {NoStop}%
\bibitem [{\citenamefont {Kos}\ \emph {et~al.}(2018)\citenamefont {Kos},
  \citenamefont {Ljubotina},\ and\ \citenamefont {Prosen}}]{kos2018many}%
  \BibitemOpen
  \bibfield  {author} {\bibinfo {author} {\bibfnamefont {P.}~\bibnamefont
  {Kos}}, \bibinfo {author} {\bibfnamefont {M.}~\bibnamefont {Ljubotina}},\
  and\ \bibinfo {author} {\bibfnamefont {T.}~\bibnamefont {Prosen}},\
  }\bibfield  {title} {\bibinfo {title} {Many-body quantum chaos: Analytic
  connection to random matrix theory},\ }\href
  {https://doi.org/10.1103/PhysRevX.8.021062} {\bibfield  {journal} {\bibinfo
  {journal} {Phys. Rev. X}\ }\textbf {\bibinfo {volume} {8}},\ \bibinfo {pages}
  {021062} (\bibinfo {year} {2018})}\BibitemShut {NoStop}%
\bibitem [{\citenamefont {Bertini}\ \emph {et~al.}(2018)\citenamefont
  {Bertini}, \citenamefont {Kos},\ and\ \citenamefont
  {Prosen}}]{bertini2018exact}%
  \BibitemOpen
  \bibfield  {author} {\bibinfo {author} {\bibfnamefont {B.}~\bibnamefont
  {Bertini}}, \bibinfo {author} {\bibfnamefont {P.}~\bibnamefont {Kos}},\ and\
  \bibinfo {author} {\bibfnamefont {T.}~\bibnamefont {Prosen}},\ }\bibfield
  {title} {\bibinfo {title} {Exact spectral form factor in a minimal model of
  many-body quantum chaos},\ }\href
  {https://doi.org/10.1103/physrevlett.121.264101} {\bibfield  {journal}
  {\bibinfo  {journal} {Phys. Rev. Lett.}\ }\textbf {\bibinfo {volume} {121}},\
  \bibinfo {pages} {264101} (\bibinfo {year} {2018})}\BibitemShut {NoStop}%
\bibitem [{\citenamefont {Bertini}\ \emph {et~al.}(2021)\citenamefont
  {Bertini}, \citenamefont {Kos},\ and\ \citenamefont
  {Prosen}}]{bertini2021random}%
  \BibitemOpen
  \bibfield  {author} {\bibinfo {author} {\bibfnamefont {B.}~\bibnamefont
  {Bertini}}, \bibinfo {author} {\bibfnamefont {P.}~\bibnamefont {Kos}},\ and\
  \bibinfo {author} {\bibfnamefont {T.}~\bibnamefont {Prosen}},\ }\bibfield
  {title} {\bibinfo {title} {Random matrix spectral form factor of dual-unitary
  quantum circuits},\ }\bibfield  {journal} {\bibinfo  {journal}
  {Communications in Mathematical Physics}\ }\href
  {https://doi.org/10.1007/s00220-021-04139-2} {10.1007/s00220-021-04139-2}
  (\bibinfo {year} {2021})\BibitemShut {NoStop}%
\bibitem [{\citenamefont {Chan}\ \emph {et~al.}(2018)\citenamefont {Chan},
  \citenamefont {De~Luca},\ and\ \citenamefont {Chalker}}]{chan2018spectral}%
  \BibitemOpen
  \bibfield  {author} {\bibinfo {author} {\bibfnamefont {A.}~\bibnamefont
  {Chan}}, \bibinfo {author} {\bibfnamefont {A.}~\bibnamefont {De~Luca}},\ and\
  \bibinfo {author} {\bibfnamefont {J.~T.}\ \bibnamefont {Chalker}},\
  }\bibfield  {title} {\bibinfo {title} {Spectral statistics in spatially
  extended chaotic quantum many-body systems},\ }\href
  {https://doi.org/10.1103/PhysRevLett.121.060601} {\bibfield  {journal}
  {\bibinfo  {journal} {Phys. Rev. Lett.}\ }\textbf {\bibinfo {volume} {121}},\
  \bibinfo {pages} {060601} (\bibinfo {year} {2018})}\BibitemShut {NoStop}%
\bibitem [{\citenamefont {{\v S}untajs}\ \emph {et~al.}(2020)\citenamefont {{\v
  S}untajs}, \citenamefont {Bon{\v c}a}, \citenamefont {Prosen},\ and\
  \citenamefont {Vidmar}}]{suntajs2020quantum}%
  \BibitemOpen
  \bibfield  {author} {\bibinfo {author} {\bibfnamefont {J.}~\bibnamefont {{\v
  S}untajs}}, \bibinfo {author} {\bibfnamefont {J.}~\bibnamefont {Bon{\v c}a}},
  \bibinfo {author} {\bibfnamefont {T.}~\bibnamefont {Prosen}},\ and\ \bibinfo
  {author} {\bibfnamefont {L.}~\bibnamefont {Vidmar}},\ }\bibfield  {title}
  {\bibinfo {title} {Quantum chaos challenges many-body localization},\ }\href
  {https://doi.org/10.1103/PhysRevE.102.062144} {\bibfield  {journal} {\bibinfo
   {journal} {Phys. Rev. E}\ }\textbf {\bibinfo {volume} {102}},\ \bibinfo
  {pages} {062144} (\bibinfo {year} {2020})}\BibitemShut {NoStop}%
\bibitem [{\citenamefont {Friedman}\ \emph {et~al.}(2019)\citenamefont
  {Friedman}, \citenamefont {Chan}, \citenamefont {De~Luca},\ and\
  \citenamefont {Chalker}}]{friedman2019spectral}%
  \BibitemOpen
  \bibfield  {author} {\bibinfo {author} {\bibfnamefont {A.~J.}\ \bibnamefont
  {Friedman}}, \bibinfo {author} {\bibfnamefont {A.}~\bibnamefont {Chan}},
  \bibinfo {author} {\bibfnamefont {A.}~\bibnamefont {De~Luca}},\ and\ \bibinfo
  {author} {\bibfnamefont {J.~T.}\ \bibnamefont {Chalker}},\ }\bibfield
  {title} {\bibinfo {title} {Spectral statistics and many-body quantum chaos
  with conserved charge},\ }\href
  {https://doi.org/10.1103/PhysRevLett.123.210603} {\bibfield  {journal}
  {\bibinfo  {journal} {Phys. Rev. Lett.}\ }\textbf {\bibinfo {volume} {123}},\
  \bibinfo {pages} {210603} (\bibinfo {year} {2019})}\BibitemShut {NoStop}%
\bibitem [{\citenamefont {Roy}\ and\ \citenamefont
  {Prosen}(2020)}]{roy2020random}%
  \BibitemOpen
  \bibfield  {author} {\bibinfo {author} {\bibfnamefont {D.}~\bibnamefont
  {Roy}}\ and\ \bibinfo {author} {\bibfnamefont {T.}~\bibnamefont {Prosen}},\
  }\bibfield  {title} {\bibinfo {title} {Random matrix spectral form factor in
  kicked interacting fermionic chains},\ }\href
  {https://doi.org/10.1103/PhysRevE.102.060202} {\bibfield  {journal} {\bibinfo
   {journal} {Phys. Rev. E}\ }\textbf {\bibinfo {volume} {102}},\ \bibinfo
  {pages} {060202} (\bibinfo {year} {2020})}\BibitemShut {NoStop}%
\bibitem [{\citenamefont {Garratt}\ and\ \citenamefont
  {Chalker}(2021{\natexlab{b}})}]{garratt2021many}%
  \BibitemOpen
  \bibfield  {author} {\bibinfo {author} {\bibfnamefont {S.~J.}\ \bibnamefont
  {Garratt}}\ and\ \bibinfo {author} {\bibfnamefont {J.~T.}\ \bibnamefont
  {Chalker}},\ }\bibfield  {title} {\bibinfo {title} {Many-body delocalization
  as symmetry breaking},\ }\href
  {https://doi.org/10.1103/PhysRevLett.127.026802} {\bibfield  {journal}
  {\bibinfo  {journal} {Phys. Rev. Lett.}\ }\textbf {\bibinfo {volume} {127}},\
  \bibinfo {pages} {026802} (\bibinfo {year} {2021}{\natexlab{b}})}\BibitemShut
  {NoStop}%
\bibitem [{\citenamefont {Flack}\ \emph {et~al.}(2020)\citenamefont {Flack},
  \citenamefont {Bertini},\ and\ \citenamefont {Prosen}}]{flack2020statistics}%
  \BibitemOpen
  \bibfield  {author} {\bibinfo {author} {\bibfnamefont {A.}~\bibnamefont
  {Flack}}, \bibinfo {author} {\bibfnamefont {B.}~\bibnamefont {Bertini}},\
  and\ \bibinfo {author} {\bibfnamefont {T.}~\bibnamefont {Prosen}},\
  }\bibfield  {title} {\bibinfo {title} {Statistics of the spectral form factor
  in the self-dual kicked ising model},\ }\href
  {https://doi.org/10.1103/physrevresearch.2.043403} {\bibfield  {journal}
  {\bibinfo  {journal} {Phys. Rev. Res.}\ }\textbf {\bibinfo {volume} {2}},\
  \bibinfo {pages} {043403} (\bibinfo {year} {2020})}\BibitemShut {NoStop}%
\bibitem [{\citenamefont {Chan}\ \emph {et~al.}(2021)\citenamefont {Chan},
  \citenamefont {De~Luca},\ and\ \citenamefont {Chalker}}]{chan2021spectral}%
  \BibitemOpen
  \bibfield  {author} {\bibinfo {author} {\bibfnamefont {A.}~\bibnamefont
  {Chan}}, \bibinfo {author} {\bibfnamefont {A.}~\bibnamefont {De~Luca}},\ and\
  \bibinfo {author} {\bibfnamefont {J.~T.}\ \bibnamefont {Chalker}},\
  }\bibfield  {title} {\bibinfo {title} {Spectral lyapunov exponents in chaotic
  and localized many-body quantum systems},\ }\href
  {https://doi.org/10.1103/PhysRevResearch.3.023118} {\bibfield  {journal}
  {\bibinfo  {journal} {Phys. Rev. Research}\ }\textbf {\bibinfo {volume}
  {3}},\ \bibinfo {pages} {023118} (\bibinfo {year} {2021})}\BibitemShut
  {NoStop}%
\bibitem [{\citenamefont {Singh}\ \emph {et~al.}(2021)\citenamefont {Singh},
  \citenamefont {Ware}, \citenamefont {Vasseur},\ and\ \citenamefont
  {Friedman}}]{singh2021subdiffusion}%
  \BibitemOpen
  \bibfield  {author} {\bibinfo {author} {\bibfnamefont {H.}~\bibnamefont
  {Singh}}, \bibinfo {author} {\bibfnamefont {B.}~\bibnamefont {Ware}},
  \bibinfo {author} {\bibfnamefont {R.}~\bibnamefont {Vasseur}},\ and\ \bibinfo
  {author} {\bibfnamefont {A.~J.}\ \bibnamefont {Friedman}},\ }\href@noop {}
  {\bibinfo {title} {Subdiffusion and many-body quantum chaos with kinetic
  constraints}} (\bibinfo {year} {2021}),\ \Eprint
  {https://arxiv.org/abs/2108.02205} {arXiv:2108.02205 [cond-mat.stat-mech]}
  \BibitemShut {NoStop}%
\bibitem [{\citenamefont {Casati}\ \emph {et~al.}(1980)\citenamefont {Casati},
  \citenamefont {Valz-Gris},\ and\ \citenamefont {Guarnieri}}]{casati1980on}%
  \BibitemOpen
  \bibfield  {author} {\bibinfo {author} {\bibfnamefont {G.}~\bibnamefont
  {Casati}}, \bibinfo {author} {\bibfnamefont {F.}~\bibnamefont {Valz-Gris}},\
  and\ \bibinfo {author} {\bibfnamefont {I.}~\bibnamefont {Guarnieri}},\
  }\bibfield  {title} {\bibinfo {title} {On the connection between quantization
  of nonintegrable systems and statistical theory of spectra},\ }\href
  {https://doi.org/10.1007/BF02798790} {\bibfield  {journal} {\bibinfo
  {journal} {Lettere al Nuovo Cimento (1971-1985)}\ }\textbf {\bibinfo {volume}
  {28}},\ \bibinfo {pages} {279} (\bibinfo {year} {1980})}\BibitemShut
  {NoStop}%
\bibitem [{\citenamefont {Berry}(1981)}]{berry1981quantizing}%
  \BibitemOpen
  \bibfield  {author} {\bibinfo {author} {\bibfnamefont {M.}~\bibnamefont
  {Berry}},\ }\bibfield  {title} {\bibinfo {title} {Quantizing a classically
  ergodic system: Sinai's billiard and the kkr method},\ }\href
  {https://doi.org/https://doi.org/10.1016/0003-4916(81)90189-5} {\bibfield
  {journal} {\bibinfo  {journal} {Annals of Physics}\ }\textbf {\bibinfo
  {volume} {131}},\ \bibinfo {pages} {163} (\bibinfo {year}
  {1981})}\BibitemShut {NoStop}%
\bibitem [{\citenamefont {Bohigas}\ \emph {et~al.}(1984)\citenamefont
  {Bohigas}, \citenamefont {Giannoni},\ and\ \citenamefont
  {Schmit}}]{bohigas1984characterization}%
  \BibitemOpen
  \bibfield  {author} {\bibinfo {author} {\bibfnamefont {O.}~\bibnamefont
  {Bohigas}}, \bibinfo {author} {\bibfnamefont {M.~J.}\ \bibnamefont
  {Giannoni}},\ and\ \bibinfo {author} {\bibfnamefont {C.}~\bibnamefont
  {Schmit}},\ }\bibfield  {title} {\bibinfo {title} {Characterization of
  chaotic quantum spectra and universality of level fluctuation laws},\ }\href
  {https://doi.org/10.1103/PhysRevLett.52.1} {\bibfield  {journal} {\bibinfo
  {journal} {Phys. Rev. Lett.}\ }\textbf {\bibinfo {volume} {52}},\ \bibinfo
  {pages} {1} (\bibinfo {year} {1984})}\BibitemShut {NoStop}%
\bibitem [{\citenamefont {Cirac}\ \emph {et~al.}(2020)\citenamefont {Cirac},
  \citenamefont {Perez-Garcia}, \citenamefont {Schuch},\ and\ \citenamefont
  {Verstraete}}]{cirac2020matrix}%
  \BibitemOpen
  \bibfield  {author} {\bibinfo {author} {\bibfnamefont {I.}~\bibnamefont
  {Cirac}}, \bibinfo {author} {\bibfnamefont {D.}~\bibnamefont {Perez-Garcia}},
  \bibinfo {author} {\bibfnamefont {N.}~\bibnamefont {Schuch}},\ and\ \bibinfo
  {author} {\bibfnamefont {F.}~\bibnamefont {Verstraete}},\ }\bibfield  {title}
  {\bibinfo {title} {Matrix product states and projected entangled pair states:
  Concepts, symmetries, and theorems},\ }\href
  {https://arxiv.org/abs/2011.12127} {\bibfield  {journal} {\bibinfo  {journal}
  {arXiv:2011.12127}\ } (\bibinfo {year} {2020})},\ \Eprint
  {https://arxiv.org/abs/2011.12127} {arXiv:2011.12127} \BibitemShut {NoStop}%
\bibitem [{\citenamefont {Fritzsch}\ and\ \citenamefont
  {Prosen}(2021)}]{fritzsch2021eigenstate}%
  \BibitemOpen
  \bibfield  {author} {\bibinfo {author} {\bibfnamefont {F.}~\bibnamefont
  {Fritzsch}}\ and\ \bibinfo {author} {\bibfnamefont {T.}~\bibnamefont
  {Prosen}},\ }\bibfield  {title} {\bibinfo {title} {Eigenstate thermalization
  in dual-unitary quantum circuits: Asymptotics of spectral functions},\ }\href
  {https://doi.org/10.1103/PhysRevE.103.062133} {\bibfield  {journal} {\bibinfo
   {journal} {Phys. Rev. E}\ }\textbf {\bibinfo {volume} {103}},\ \bibinfo
  {pages} {062133} (\bibinfo {year} {2021})}\BibitemShut {NoStop}%
\bibitem [{\citenamefont {Bertini}\ \emph
  {et~al.}(2020{\natexlab{a}})\citenamefont {Bertini}, \citenamefont {Kos},\
  and\ \citenamefont {Prosen}}]{bertini2020operator}%
  \BibitemOpen
  \bibfield  {author} {\bibinfo {author} {\bibfnamefont {B.}~\bibnamefont
  {Bertini}}, \bibinfo {author} {\bibfnamefont {P.}~\bibnamefont {Kos}},\ and\
  \bibinfo {author} {\bibfnamefont {T.}~\bibnamefont {Prosen}},\ }\bibfield
  {title} {\bibinfo {title} {{Operator Entanglement in Local Quantum Circuits
  I: Chaotic Dual-Unitary Circuits}},\ }\href
  {https://doi.org/10.21468/SciPostPhys.8.4.067} {\bibfield  {journal}
  {\bibinfo  {journal} {SciPost Phys.}\ }\textbf {\bibinfo {volume} {8}},\
  \bibinfo {pages} {67} (\bibinfo {year} {2020}{\natexlab{a}})}\BibitemShut
  {NoStop}%
\bibitem [{\citenamefont {Bertini}\ \emph
  {et~al.}(2020{\natexlab{b}})\citenamefont {Bertini}, \citenamefont {Kos},\
  and\ \citenamefont {Prosen}}]{bertini2020operator2}%
  \BibitemOpen
  \bibfield  {author} {\bibinfo {author} {\bibfnamefont {B.}~\bibnamefont
  {Bertini}}, \bibinfo {author} {\bibfnamefont {P.}~\bibnamefont {Kos}},\ and\
  \bibinfo {author} {\bibfnamefont {T.}~\bibnamefont {Prosen}},\ }\bibfield
  {title} {\bibinfo {title} {Operator entanglement in local quantum circuits
  ii: Solitons in chains of qubits},\ }\href
  {https://doi.org/10.21468/scipostphys.8.4.068} {\bibfield  {journal}
  {\bibinfo  {journal} {SciPost Physics}\ }\textbf {\bibinfo {volume} {8}},\
  \bibinfo {pages} {68} (\bibinfo {year} {2020}{\natexlab{b}})}\BibitemShut
  {NoStop}%
\bibitem [{\citenamefont {Claeys}\ and\ \citenamefont
  {Lamacraft}(2020)}]{claeys2020maximum}%
  \BibitemOpen
  \bibfield  {author} {\bibinfo {author} {\bibfnamefont {P.~W.}\ \bibnamefont
  {Claeys}}\ and\ \bibinfo {author} {\bibfnamefont {A.}~\bibnamefont
  {Lamacraft}},\ }\bibfield  {title} {\bibinfo {title} {Maximum velocity
  quantum circuits},\ }\href {https://doi.org/10.1103/physrevresearch.2.033032}
  {\bibfield  {journal} {\bibinfo  {journal} {Phys. Rev. Res.}\ }\textbf
  {\bibinfo {volume} {2}},\ \bibinfo {pages} {033032} (\bibinfo {year}
  {2020})}\BibitemShut {NoStop}%
\bibitem [{\citenamefont {Bertini}\ and\ \citenamefont
  {Piroli}(2020)}]{bertini2020scrambling}%
  \BibitemOpen
  \bibfield  {author} {\bibinfo {author} {\bibfnamefont {B.}~\bibnamefont
  {Bertini}}\ and\ \bibinfo {author} {\bibfnamefont {L.}~\bibnamefont
  {Piroli}},\ }\bibfield  {title} {\bibinfo {title} {Scrambling in random
  unitary circuits: Exact results},\ }\href
  {https://doi.org/10.1103/physrevb.102.064305} {\bibfield  {journal} {\bibinfo
   {journal} {Phys. Rev. B}\ }\textbf {\bibinfo {volume} {102}},\ \bibinfo
  {pages} {064305} (\bibinfo {year} {2020})}\BibitemShut {NoStop}%
\bibitem [{\citenamefont {Gopalakrishnan}\ and\ \citenamefont
  {Lamacraft}(2019)}]{gopalakrishnan2019unitary}%
  \BibitemOpen
  \bibfield  {author} {\bibinfo {author} {\bibfnamefont {S.}~\bibnamefont
  {Gopalakrishnan}}\ and\ \bibinfo {author} {\bibfnamefont {A.}~\bibnamefont
  {Lamacraft}},\ }\bibfield  {title} {\bibinfo {title} {Unitary circuits of
  finite depth and infinite width from quantum channels},\ }\href
  {https://doi.org/10.1103/physrevb.100.064309} {\bibfield  {journal} {\bibinfo
   {journal} {Phys. Rev. B}\ }\textbf {\bibinfo {volume} {100}},\ \bibinfo
  {pages} {064309} (\bibinfo {year} {2019})}\BibitemShut {NoStop}%
\bibitem [{\citenamefont {Jonay}\ \emph {et~al.}(2021)\citenamefont {Jonay},
  \citenamefont {Khemani},\ and\ \citenamefont
  {Ippoliti}}]{jonay2021triunitary}%
  \BibitemOpen
  \bibfield  {author} {\bibinfo {author} {\bibfnamefont {C.}~\bibnamefont
  {Jonay}}, \bibinfo {author} {\bibfnamefont {V.}~\bibnamefont {Khemani}},\
  and\ \bibinfo {author} {\bibfnamefont {M.}~\bibnamefont {Ippoliti}},\
  }\href@noop {} {\bibinfo {title} {Tri-unitary quantum circuits}} (\bibinfo
  {year} {2021}),\ \Eprint {https://arxiv.org/abs/2106.07686} {arXiv:2106.07686
  [quant-ph]} \BibitemShut {NoStop}%
\bibitem [{\citenamefont {Claeys}\ and\ \citenamefont
  {Lamacraft}(2021)}]{claeys2021ergodic}%
  \BibitemOpen
  \bibfield  {author} {\bibinfo {author} {\bibfnamefont {P.~W.}\ \bibnamefont
  {Claeys}}\ and\ \bibinfo {author} {\bibfnamefont {A.}~\bibnamefont
  {Lamacraft}},\ }\bibfield  {title} {\bibinfo {title} {Ergodic and nonergodic
  dual-unitary quantum circuits with arbitrary local hilbert space dimension},\
  }\href {https://doi.org/10.1103/physrevlett.126.100603} {\bibfield  {journal}
  {\bibinfo  {journal} {Phys. Rev. Lett.}\ }\textbf {\bibinfo {volume} {126}},\
  \bibinfo {pages} {100603} (\bibinfo {year} {2021})}\BibitemShut {NoStop}%
\bibitem [{\citenamefont {Kos}\ \emph {et~al.}(2021{\natexlab{a}})\citenamefont
  {Kos}, \citenamefont {Bertini},\ and\ \citenamefont
  {Prosen}}]{kos2021correlations}%
  \BibitemOpen
  \bibfield  {author} {\bibinfo {author} {\bibfnamefont {P.}~\bibnamefont
  {Kos}}, \bibinfo {author} {\bibfnamefont {B.}~\bibnamefont {Bertini}},\ and\
  \bibinfo {author} {\bibfnamefont {T.}~\bibnamefont {Prosen}},\ }\bibfield
  {title} {\bibinfo {title} {Correlations in perturbed dual-unitary circuits:
  Efficient path-integral formula},\ }\href
  {https://doi.org/10.1103/physrevx.11.011022} {\bibfield  {journal} {\bibinfo
  {journal} {Phys. Rev. X}\ }\textbf {\bibinfo {volume} {11}},\ \bibinfo
  {pages} {011022} (\bibinfo {year} {2021}{\natexlab{a}})}\BibitemShut
  {NoStop}%
\bibitem [{\citenamefont {Reid}\ and\ \citenamefont
  {Bertini}(2021)}]{reid2021entanglement}%
  \BibitemOpen
  \bibfield  {author} {\bibinfo {author} {\bibfnamefont {I.}~\bibnamefont
  {Reid}}\ and\ \bibinfo {author} {\bibfnamefont {B.}~\bibnamefont {Bertini}},\
  }\bibfield  {title} {\bibinfo {title} {Entanglement barriers in dual-unitary
  circuits},\ }\href {https://doi.org/10.1103/PhysRevB.104.014301} {\bibfield
  {journal} {\bibinfo  {journal} {Phys. Rev. B}\ }\textbf {\bibinfo {volume}
  {104}},\ \bibinfo {pages} {014301} (\bibinfo {year} {2021})}\BibitemShut
  {NoStop}%
\bibitem [{\citenamefont {Bengtsson}\ \emph {et~al.}(2008)\citenamefont
  {Bengtsson}, \citenamefont {Zyczkowski},\ and\ \citenamefont
  {Milburn}}]{bengtsson2008geometry}%
  \BibitemOpen
  \bibfield  {author} {\bibinfo {author} {\bibfnamefont {I.}~\bibnamefont
  {Bengtsson}}, \bibinfo {author} {\bibfnamefont {K.}~\bibnamefont
  {Zyczkowski}},\ and\ \bibinfo {author} {\bibfnamefont {G.}~\bibnamefont
  {Milburn}},\ }\bibfield  {title} {\bibinfo {title} {Geometry of quantum
  states: an introduction to quantum entanglement},\ }\href
  {https://doi.org/10.26421/QIC8.8-9-12} {\bibfield  {journal} {\bibinfo
  {journal} {Quantum Information and Computation}\ }\textbf {\bibinfo {volume}
  {8}},\ \bibinfo {pages} {860} (\bibinfo {year} {2008})}\BibitemShut {NoStop}%
\bibitem [{\citenamefont {Lerose}\ \emph {et~al.}(2021)\citenamefont {Lerose},
  \citenamefont {Sonner},\ and\ \citenamefont {Abanin}}]{lerose2020influence}%
  \BibitemOpen
  \bibfield  {author} {\bibinfo {author} {\bibfnamefont {A.}~\bibnamefont
  {Lerose}}, \bibinfo {author} {\bibfnamefont {M.}~\bibnamefont {Sonner}},\
  and\ \bibinfo {author} {\bibfnamefont {D.~A.}\ \bibnamefont {Abanin}},\
  }\bibfield  {title} {\bibinfo {title} {Influence matrix approach to many-body
  floquet dynamics},\ }\href {https://doi.org/10.1103/PhysRevX.11.021040}
  {\bibfield  {journal} {\bibinfo  {journal} {Phys. Rev. X}\ }\textbf {\bibinfo
  {volume} {11}},\ \bibinfo {pages} {021040} (\bibinfo {year}
  {2021})}\BibitemShut {NoStop}%
\bibitem [{\citenamefont {Kos}\ \emph {et~al.}(2021{\natexlab{b}})\citenamefont
  {Kos}, \citenamefont {Bertini},\ and\ \citenamefont {Prosen}}]{kos2021chaos}%
  \BibitemOpen
  \bibfield  {author} {\bibinfo {author} {\bibfnamefont {P.}~\bibnamefont
  {Kos}}, \bibinfo {author} {\bibfnamefont {B.}~\bibnamefont {Bertini}},\ and\
  \bibinfo {author} {\bibfnamefont {T.}~\bibnamefont {Prosen}},\ }\bibfield
  {title} {\bibinfo {title} {Chaos and ergodicity in extended quantum systems
  with noisy driving},\ }\href {https://doi.org/10.1103/PhysRevLett.126.190601}
  {\bibfield  {journal} {\bibinfo  {journal} {Phys. Rev. Lett.}\ }\textbf
  {\bibinfo {volume} {126}},\ \bibinfo {pages} {190601} (\bibinfo {year}
  {2021}{\natexlab{b}})}\BibitemShut {NoStop}%
\bibitem [{\citenamefont {Ippoliti}\ \emph {et~al.}(2021)\citenamefont
  {Ippoliti}, \citenamefont {Rakovszky},\ and\ \citenamefont
  {Khemani}}]{ippoliti2021fractal}%
  \BibitemOpen
  \bibfield  {author} {\bibinfo {author} {\bibfnamefont {M.}~\bibnamefont
  {Ippoliti}}, \bibinfo {author} {\bibfnamefont {T.}~\bibnamefont
  {Rakovszky}},\ and\ \bibinfo {author} {\bibfnamefont {V.}~\bibnamefont
  {Khemani}},\ }\href@noop {} {\bibinfo {title} {Fractal, logarithmic and
  volume-law entangled non-thermal steady states via spacetime duality}}
  (\bibinfo {year} {2021}),\ \Eprint {https://arxiv.org/abs/2103.06873}
  {arXiv:2103.06873 [cond-mat.stat-mech]} \BibitemShut {NoStop}%
\bibitem [{\citenamefont {Ippoliti}\ and\ \citenamefont
  {Khemani}(2021)}]{ippoliti2021postselection}%
  \BibitemOpen
  \bibfield  {author} {\bibinfo {author} {\bibfnamefont {M.}~\bibnamefont
  {Ippoliti}}\ and\ \bibinfo {author} {\bibfnamefont {V.}~\bibnamefont
  {Khemani}},\ }\bibfield  {title} {\bibinfo {title} {Postselection-free
  entanglement dynamics via spacetime duality},\ }\href
  {https://doi.org/10.1103/PhysRevLett.126.060501} {\bibfield  {journal}
  {\bibinfo  {journal} {Phys. Rev. Lett.}\ }\textbf {\bibinfo {volume} {126}},\
  \bibinfo {pages} {060501} (\bibinfo {year} {2021})}\BibitemShut {NoStop}%
\bibitem [{\citenamefont {Turner}\ \emph {et~al.}(2018)\citenamefont {Turner},
  \citenamefont {Michailidis}, \citenamefont {Abanin}, \citenamefont {Serbyn},\
  and\ \citenamefont {Papi{\'c}}}]{turner2018weak}%
  \BibitemOpen
  \bibfield  {author} {\bibinfo {author} {\bibfnamefont {C.~J.}\ \bibnamefont
  {Turner}}, \bibinfo {author} {\bibfnamefont {A.~A.}\ \bibnamefont
  {Michailidis}}, \bibinfo {author} {\bibfnamefont {D.~A.}\ \bibnamefont
  {Abanin}}, \bibinfo {author} {\bibfnamefont {M.}~\bibnamefont {Serbyn}},\
  and\ \bibinfo {author} {\bibfnamefont {Z.}~\bibnamefont {Papi{\'c}}},\
  }\bibfield  {title} {\bibinfo {title} {Weak ergodicity breaking from quantum
  many-body scars},\ }\href@noop {} {\bibfield  {journal} {\bibinfo  {journal}
  {Nature Physics}\ }\textbf {\bibinfo {volume} {14}},\ \bibinfo {pages} {745}
  (\bibinfo {year} {2018})}\BibitemShut {NoStop}%
\bibitem [{\citenamefont {Bernien}\ \emph {et~al.}(2017)\citenamefont
  {Bernien}, \citenamefont {Schwartz}, \citenamefont {Keesling}, \citenamefont
  {Levine}, \citenamefont {Omran}, \citenamefont {Pichler}, \citenamefont
  {Choi}, \citenamefont {Zibrov}, \citenamefont {Endres}, \citenamefont
  {Greiner} \emph {et~al.}}]{bernien2017probing}%
  \BibitemOpen
  \bibfield  {author} {\bibinfo {author} {\bibfnamefont {H.}~\bibnamefont
  {Bernien}}, \bibinfo {author} {\bibfnamefont {S.}~\bibnamefont {Schwartz}},
  \bibinfo {author} {\bibfnamefont {A.}~\bibnamefont {Keesling}}, \bibinfo
  {author} {\bibfnamefont {H.}~\bibnamefont {Levine}}, \bibinfo {author}
  {\bibfnamefont {A.}~\bibnamefont {Omran}}, \bibinfo {author} {\bibfnamefont
  {H.}~\bibnamefont {Pichler}}, \bibinfo {author} {\bibfnamefont
  {S.}~\bibnamefont {Choi}}, \bibinfo {author} {\bibfnamefont {A.~S.}\
  \bibnamefont {Zibrov}}, \bibinfo {author} {\bibfnamefont {M.}~\bibnamefont
  {Endres}}, \bibinfo {author} {\bibfnamefont {M.}~\bibnamefont {Greiner}},
  \emph {et~al.},\ }\bibfield  {title} {\bibinfo {title} {Probing many-body
  dynamics on a 51-atom quantum simulator},\ }\href@noop {} {\bibfield
  {journal} {\bibinfo  {journal} {Nature}\ }\textbf {\bibinfo {volume} {551}},\
  \bibinfo {pages} {579} (\bibinfo {year} {2017})}\BibitemShut {NoStop}%
\bibitem [{\citenamefont {Abanin}\ \emph {et~al.}(2019)\citenamefont {Abanin},
  \citenamefont {Altman}, \citenamefont {Bloch},\ and\ \citenamefont
  {Serbyn}}]{abanin2019many}%
  \BibitemOpen
  \bibfield  {author} {\bibinfo {author} {\bibfnamefont {D.~A.}\ \bibnamefont
  {Abanin}}, \bibinfo {author} {\bibfnamefont {E.}~\bibnamefont {Altman}},
  \bibinfo {author} {\bibfnamefont {I.}~\bibnamefont {Bloch}},\ and\ \bibinfo
  {author} {\bibfnamefont {M.}~\bibnamefont {Serbyn}},\ }\bibfield  {title}
  {\bibinfo {title} {Colloquium: Many-body localization, thermalization, and
  entanglement},\ }\href {https://doi.org/10.1103/RevModPhys.91.021001}
  {\bibfield  {journal} {\bibinfo  {journal} {Rev. Mod. Phys.}\ }\textbf
  {\bibinfo {volume} {91}},\ \bibinfo {pages} {021001} (\bibinfo {year}
  {2019})}\BibitemShut {NoStop}%
\bibitem [{\citenamefont {Fishman}\ \emph {et~al.}(2020)\citenamefont
  {Fishman}, \citenamefont {White},\ and\ \citenamefont
  {Stoudenmire}}]{fishman2020itensor}%
  \BibitemOpen
  \bibfield  {author} {\bibinfo {author} {\bibfnamefont {M.}~\bibnamefont
  {Fishman}}, \bibinfo {author} {\bibfnamefont {S.~R.}\ \bibnamefont {White}},\
  and\ \bibinfo {author} {\bibfnamefont {E.~M.}\ \bibnamefont {Stoudenmire}},\
  }\href@noop {} {\bibinfo {title} {The \mbox{ITensor} software library for
  tensor network calculations}} (\bibinfo {year} {2020}),\ \Eprint
  {https://arxiv.org/abs/2007.14822} {arXiv:2007.14822} \BibitemShut {NoStop}%
\end{thebibliography}%
%\printbibliography

\end{document}